\documentclass[journal]{IEEEtran}
\UseRawInputEncoding
\usepackage{cite}
\usepackage{breqn}
\usepackage{multirow}
\usepackage{graphicx}
\usepackage{amssymb,amsmath,times,wasysym}
\usepackage[ruled,vlined]{algorithm2e}
\usepackage{color}
\usepackage[english]{babel}

\usepackage{pb-diagram}
\usepackage{psfrag,balance,verbatim}
\usepackage{srcltx,color} 
\usepackage{float} 
\usepackage{subfigure}
\usepackage{tabularx}
\usepackage{booktabs}
\usepackage{amsthm}
\usepackage{nomencl}
\usepackage{multirow}
 \usepackage[utf8]{inputenc}
 
\newtheorem{theorem}{Theorem}

\newtheorem{lemma}{Lemma}

\usepackage[labelfont=bf,format=plain,justification=raggedright,singlelinecheck=false]{caption}

\def\BibTeX{{\rm B\kern-.05em{\sc i\kern-.025em b}\kern-.08em
    T\kern-.1667em\lower.7ex\hbox{E}\kern-.125emX}}

\newcounter{appendx}

\newcommand{\real}[1]{\Re\left({#1}\right)}

\renewcommand{\exp}[1]{\mathrm{e}^{#1}}

\newcommand{\E}[2][]{\mathbb{E}_{#1}\!\left[{#2}\right]}

\makenomenclature
\usepackage{etoolbox}
\renewcommand\nomgroup[1]{%
  \item[\itshape
  \ifstrequal{#1}{N}{Notation}{%
  \ifstrequal{#1}{A}{Acronyms}}%
]}
\begin{document}
\bstctlcite{IEEEexample:BSTcontrol}

\title{ On the Number of Maintenance Cycles in Systems with Critical and Non-Critical Components}

\author{Guanchen Li, \IEEEmembership{Student Member, IEEE}, and Dimitri Kagaris, \IEEEmembership{Member, IEEE}
\thanks{Guanchen Li and Dimitri Kagaris are with the School of Electrical, Computer, and Biomedical Engineering, Southern Illinois University, Carbondale, IL 62901, USA. Email: guanchen.li@siu.edu, \,kagaris@engr.siu.edu. }}

% ===========================================================================
% document content
% ===========================================================================

\maketitle
 \thispagestyle{plain}
 \pagestyle{plain}
% abstract & keywords
\begin{abstract}

We present a novel mathematical framework for computing the number of maintenance cycles in a system with critical and non-critical components, where ``critical'' (CR) means that the component's failure is fatal for the system's operation and renders any more repairs inapplicable, whereas ``non-critical" (NC) means that the component  can undergo corrective maintenance (replacement or minimal repair) whenever it fails, provided that the CR component is still in operation. Whenever the NC component fails, the CR component can optionally be preventively replaced.
 We  extend traditional renewal theory (whether classical or generalized) for various maintenance scenarios for a system composed of one CR and one NC component in order to compute the average number of renewals of NC under the restriction (``bound") necessitated by CR. 
We also develop approximations in closed form for the proposed ``bounded" renewal functions. We validate our formulas by simulations on a variety of component lifetime distributions, including actual lifetime distributions of wind turbine components.

\end{abstract}

\IEEEpeerreviewmaketitle
%%%%%%%%%%%%%%%%%%%NOTATIONS%%%%%
\nomenclature[A,01]{NC}{Non-critical}
\nomenclature[A,02]{CR}{Critical}
\nomenclature[A,03]{MR}{Minimal repair}
\nomenclature[A,04]{CM}{Corrective maintenance}
\nomenclature[A,05]{PM}{Preventive maintenance}
\nomenclature[A,06]{g-renewal}{Generalized renewal }

\nomenclature[N,01]{\(T\), \(T_i\)}{lifetime of NC, $i$th lifetime of NC}
\nomenclature[N,02]{\(\tilde T\), \(\tilde T_i\)}{lifetime of CR, $i$th lifetime of CR}
\nomenclature[N,03]{\(F_T(\cdot)\), \(\bar F_T(\cdot)\), \(f_T(\cdot)\)}{CDF, survival function, and pdf for NC}
\nomenclature[N,04]{\( F_{\tilde T}(\cdot)\), \(\bar  F_{\tilde T}(\cdot)\), \(f_{\tilde T}(\cdot)\)}{CDF, survival function, and pdf for CR}

\nomenclature[N,05]{\(S_n\)}{Summation of the NC's lifetimes  till  the $n$th renewal
}
\nomenclature[N,06]{\(M(\cdot)\)}{Classical renewal function}
\nomenclature[N,07]{\(M_g(\cdot)\)}{g-renewal function}
\nomenclature[N,08]{\(\hat M(\cdot)\)}{Approximated classical renewal function}
\nomenclature[N,09]{\(\hat M_g(\cdot)\)}{Approximated g-renewal function}
\nomenclature[N,10]{\(N_X(\cdot)\)}{Number of renewals for Process ``X''}
\nomenclature[N,11]{\(W_X(\cdot)\)}{Bounded renewal function for Process ``X''}
\nomenclature[N,12]{\(\hat W_X(\cdot)\)}{ Approximated \(W_X(\cdot)\)}
\nomenclature[N,13]{\(p_B\)}{ Probability of NC's lifetime being greater than that of CR in Process B}
\nomenclature[N,14]{\(p_{\mathcal B_{MR}}(y)\)}{ Probability of NC being greater than CR given that NC has been working for time $y$ in Process B-MR}

\nomenclature[N,16]{\(\tau(n)\)}{Approximated average time until the $n$th renewal in Process B-MR}
\nomenclature[N,17]{\(\rho_X(n)\)}{Probability of NC being greater than CR for $n$ renewals in Process B-MR under  consideration ``X''}

\printnomenclature

\section{Introduction.}\label{S:intro} 
Dependencies existing among the components of a multi-unit system also complicate its maintenance. Usually, when a component needs to undergo corrective maintenance (CM) upon failure, or preventive maintenance (PM) after a given period of time, it is cost-effective to perform some form of preventive maintenance on other components too. For maintenance planning and cost management, it is extremely helpful to know the number of maintenance sessions that a component has to undergo until a desired time $t$. Renewal theory \cite{rausand2003system}  plays an important role in that regard. Based on that, the fundamental question is how to minimize the overall maintenance cost given the general trade-off between expensive replacements or repairs upon failure, versus scheduled or opportunistic maintenance actions while components are still functioning. Recent works on the investigation of efficient maintenance policies include:
 Xu et al. \cite{XU2018141} examined the relationship between defects in a two-unit system with repairs. To simulate how the hazard rate of one component affects the health condition of another, a state-space discretization approach was suggested.
Castanier et al. \cite{CASTANIER2005109} developed a condition-based maintenance approach for a two-unit degrading system in which each unit is maintained through successive non-periodic examinations and can be maintained by  full replacements.
Zhang et al. \cite{zhang2011condition} proposed an iterative process to reduce overall maintenance costs through a condition-based and opportunistic maintenance strategy for a sereal two-unit degrading system.
Li et al. \cite{li2018maintenance} provided an approach to address the optimization of system maintenance and steady state availability of a system with two units connected in series.
%For a two-unit repairable system, most papers assumed that all maintenance actions were carried out at an inspection time.
Zequeira  and Berenguer \cite{zequeira2004maintenance} regarded as a two-unit standby system that can function well if at least one units is not damaged.
%Component failures are detected only by periodic testing or inspections. When a component is found failed, a corrective maintenance action is made. Besides periodic inspections, preventive maintenance actions are scheduled for the system at a fixed time after the last preventive or corrective maintenance action. 
Salari and Makis \cite{salari2017optimal} provided a preventive and opportunistic maintenance strategy for a two-unit model with economic dependency that may be used for distribution networks of electric power. %where the maintenance action would be carried out at discrete inspection time epochs when failure occurred during an inspection interval.
%However, the failed component is maintained immediately in a two-unit system should also consider the preventive maintenance simultaneously for the rest component.
Shafiee and Finkelstein \cite{shafiee2015optimal} proposed an ideal age-based group maintenance plan for a multi-unit series system that reduces the average long-term maintenance expense per unit of time for the system. Corrective maintenance (replacement) is performed on a component when it reaches a certain degree of degradation, and preventative maintenance is also performed on the other components at the same time. Additionally, the entire system receives routine preventative maintenance.
Wang et al. \cite{WANG20191} analyzed the age- and condition-based opportunistic maintenance strategy for a serial two-unit system where one unit is health monitored.
%They also derived a formula for the average maintenance cost and determined the optimal thresholds for preventive and opportunistic maintenance of the two units that minimizing the long run expected average cost. 
Zhang et al. \cite{zhang2018optimal} considered a system in which  the failure of the first unit  is likely to cause the failure of the second unit, whereas the failure of the second unit is fatal for the entire system.
By taking into account a two-component system with failure interactions, they were able to establish an optimum strategy for minimizing the long-term average maintenance cost.
%However, to the best of our knowledge, the previously work have not studied the analysis of asymptotic number of renewals, but in the cost minimization perspective. We focus on the number of renewals of CM that need to be implemented during the operating time of the system. We highlight that the stochastic process analysis is of great research interest in finding the important system reliability measure, i.e. the mean number of failures during a time interval.  It is essential to know renewal function, which is the expected number of arrivals/renewals in a given time slot in renewal theory. The renewal function can fundamentally provide key insights for long term behavior of the stochastic processes and helpfully guide people’s life and industrial production. For example, it helps to determine how often to replace worn-out machinery in a factory which achieves a good trade-off between maintenance cost and risk.  However, none of them have studied the more fundamental theoretical analysis of renewal function for the CM-PM model of two-unit.In existing research, optimization of PM policies is carried out under assumptions that PM is subject to a specific level of maintenance effectiveness, while CM is usually considered to be minimal. Other combinations of different maintenance models should be investigated. 
Syamsundar et al. \cite{syamsundar2021estimating} examined PM and CM models on a repairable one-unit system based on various combinations of maintenance types and calculated maintenance efficacy by providing failure times and repairs data.
Wang et al.\cite{wang2017preventive} 
examined two preventive maintenance models by means of the generalized geometric process.
Sun et al.\cite{sun2018scheduling} introduced two preventive maintenance models to describe the fluctuations of the saturation effect and determine the long-run overall cost.
Yu et al. \cite {6045309} considered various degrees of maintenance efficacy for imperfect Kijima and non-linear preventive maintenance models with limited CM operations.
Doyen et al. \cite{6046108} established a  framework to model and assess  aging and maintenance efficiency in  systems with corrective and preventative maintenance.
Kijima et al. \cite{KIJIMA1988194} compared a periodic replacement issue with a general repair, where a system is only changed when necessary and is fixed anytime it breaks down.
%Apart from this, there is little research estimating maintenance feasibility in a model with a consideration of both NC and CR of PM and CM. Properly researching into it will lead to more cost-effective decisions on maintenance policies. 
In our recent work \cite{IJRQSE}, we focused on renewals in a system that has two units that are coupled in series. In the event of a breakdown, any unit receives corrective maintenance (i.e., complete replacement, perfect repair, or minimal repair). When one of the units fails and must be repaired, the other unit may  optionally be replaced  (if  the cost is allowable). 

In this study, we take into account a two-unit system with a critical (CR) and a non-critical (NC) component. 
In contrast to \cite{IJRQSE}, where system operation continues indefinitely, when the CR component fails, the entire system is presumed to stop operating, and no further maintenance or repair is required. As long as the CR has not failed, the NC component may perform corrective maintenance (i.e., it may be completely replaced/perfectly repaired to an ``as-good-as-new" condition or minimally repaired to an ``as-bad-as-old" state).
Since the whole system is assumed to shut down when the CR fails, CR can only undertake preventive full replacement (on the occasions when NC fails) or simply left as is. 

If one considers the NC component alone (without regard to CR), the average number of renewals of NC is readily calculated  (see, e.g., \cite{rausand2003system}).
The average number of NC renewals under the restriction (or ``bound") set by CR, however, has not been studied (to the best of our knowledge). In this work, we first formulate new ``bounded" renewal functions to address the above goal.
  Then we establish closed-form approximations that work for arbitrary lifetime distributions of NC and CR. We validate our formulas via Monte Carlo simulations for various combinations of distributions, including actual lifetime distributions for wind turbine components. Preliminary results of this research appear in \cite{Kaga2206}.

The rest of the paper is organized as follows: A brief review of renewal theory is given in Section \ref{S:Pre}. The derivations of the renewal functions and the approximate models for perfectly repaired NC (Process A) and minimally repaired NC (Process A-MR) without preventive maintenance for CR are respectively given in Section \ref{S:ModelA} and \ref{S:ModelAMR}. Section \ref{S:ModelB} and \ref{S:ModelBMR} provide the derivations of renewal functions and the approximations of the models for perfectly repaired NC (Process B) and minimally repaired NC (Process B-MR) with preventive maintenance for CR, respectively. 
%Section \ref{S:Approx} proposes an  approximation for the renewal function defined in Section \ref{S:Model2}. 
Monte Carlo simulations are used to confirm that the derivations are correct for various component lifetime selections.
In Section \ref{S:case}, the illustrations of the application of our formulas
on actual lifetime distributions that have been obtained for
wind turbine components are provided \cite{shafiee2015optimal}, \cite{WANG20191}. 
Finally, Section  \ref{S:Conclusion} gives conclusions and future research. All  proofs are given in the Appendix (Section \ref{S:Appendix}).

\section{Overview of renewal theory.}\label{S:Pre} 

 \subsection{Overview of classical renewal theory.}\label{S:class} 

 A renewal process (RP) is a counting process where the inter-arrival intervals of the event in concern are independent and identically distributed  random variables with arbitrary distribution (see, for instance, \cite{rausand2003system}).

Let each inter-arrival interval be represented by the random variable $T_i, i \ge 1$. Let the underlying cumulative distribution function (CDF) of each $T_i$  be  $F_T(t)=P{(T_i\leq t)}$, and the corresponding probability density function (pdf) be   $f_T(t)$.

Let the time at the completion of  the $n$th renewal be  $S_n= T_1 +T_2 +...+ T_n = \sum_{i=1}^nT_i$. 
 The distribution function $L^{(n)}(t)$ of $S_n$ is given by the convolution  of the distribution function of $S_{n-1}$ and that of the $n$th lifetime $T_n$ respectively, i.e.,  
\begin{eqnarray}  \label{eqn:conofclass}
L^{(n)}(t)=\int_0^tL^{(n-1)}(t-x)f_T(x)dx.
\end{eqnarray}

In the classical renewal process, each renewal can be associated with a full replacement or perfect repair of the component. Let $N(t)$ be the counting process that gives the number of renewals in the time interval (0,$t$]:
$
N(t) = \max\{ n: S_n \leq t \}
$, and it also follows that $P{(N(t)\geq n)}=P{(S_n\leq t)}=L^{(n)}(t)$.

Let $M(t)$ be the mean number of renewals in the time interval (0,$t$]. Then the renewal function can be expressed as: 
\begin{eqnarray}  \label{eqn:class}
\!\!\!\!\!\!\!\!\!\!\!\!&&M(t) \!\!=\!\!\E{N(t)} \!\!=\!\!\sum_{n=1}^{\infty}\!\!P(N(t)\!\!\geq\!\! n) \!\!=\!\!\sum_{n=1}^{\infty}\!\!P(S_n\!\!\leq\!\! t)\!\!= \!\!\sum_{n=1}^{\infty}\!\!L^{(n)}(t) \nonumber\\
\!\!\!\!\!\!\!\!\!\!\!\!&&\iff M(t) = F_T(t) + \int_0^t M(t-x) dF_T(x).
\end{eqnarray}

The computation of Eq. \ref{eqn:class} is in general difficult since it is recursive. If the Laplace transform and its inverse can be calculated, one approach would be to use
\begin{dmath}\label{lapclass}
M^*(s)=\frac{f_T^*(s)}{s(1-f_T^*(s))}
.\end{dmath}
%The matching inverse Laplace transform should then be calculated. However, depending on the functions, obtaining this may still be difficult.
Otherwise, approximation methods like those in \cite{Bartholomew,deligonul_1985,tortorella2005numerical,ran2006some,garg1998approximations} can be used. 
%Appropriate integrals, quadrature rules for Stieltjes integrals, rational polynomial approximation, power series expansion techniques, generating functions, direct numerical computation,  and numerical inversion of the Laplace transform are some methods that have been suggested to approximate the classical renewal function $M(t)$.(e.g., \cite{Bartholomew,deligonul_1985,tortorella2005numerical,ran2006some,garg1998approximations}). For the purpose of our work here, we use the approximation of \cite{Bartholomew}. 

%The renewal density $M(t)$ satisfies:
%\begin{eqnarray}  
%M(t) = \frac{d}{dt}M(t)=f_T(t) + \int_0^t h(t-x) f_T(x)dx.
%\end{eqnarray} 

\subsection{Overview of Generalized Renewal Theory.} 

The generalized renewal (g-renewal) theory of \cite{Kijima1986,kijima2002generalized} allows the lifetime distribution of  the $n+1$th interval to depend on the value of $S_n=\sum_{i=1}^n T_i$, in contrast with the classical renewal theory where each inter-occurrence interval is  independent of the others.
Let $P(T_i \le t) =F_T(t|y=S_{i-1})$ be the 
 conditional CDF of each $T_i$  with density $f_T (x|y) $, and let the corresponding  survival function be $\bar F_T(t|y)=1-F_T(t|y)$.
%In the g-renewal process, each renewal can be associated with a minimal repair of the component, which after its $n$th failure (and minimal repair) has remaining lifetime  $F_T(t|y)=P (T\leq t | y=S_{n-1})$.

Let $N_g(t)$ be the g-renewal counting process. Let $P_n(y,t)$ be the probability that $N_g(t)$ increases to at least  $n$  in the interval $(y,t]$: 
 \begin{eqnarray} \label{eqn:pky}
 P_n(y,t)=\int_y^t P_{n-1}(z,t)f_T (z-y|y)dz.
 \end{eqnarray}
Then the following relationship holds: 
  \begin{eqnarray} \label{eqn:pky1}
  P_n(0 ,t)=P{(N_g(t)\geq n)}=P{(S_n\leq t)}.
 \end{eqnarray}
 
The g-renewal function $M_g (t)$ for counting the number of minimal repairs is given as 
\begin{eqnarray}\label{eqn:g-renew}
\!\!\!\!\!\!\!\!\!\!\!\!\!\!&&M_g (t) \!\!=\!\!\E{N_g(t)} \!\!=\!\!\sum_{n=1}^{\infty}\!\!P(N_g(t)\!\!\geq\!\! n) \!\!=\!\!\sum_{n=1}^{\infty}\!\!P(S_n\!\!\leq \!\!t)\!\!=\!\!\sum_{n=1}^{\infty}\!\!P_n(0,t) \nonumber\\
\!\!\!\!\!\!\!\!\!\!\!\!\!\!&&\iff M_g (t)= F_T(t|0) + \int_0^{t}  F_T(t-y| y) d M_g(y).
\end{eqnarray}

(We note that Eq. \ref{eqn:g-renew} generalizes the classical renewal function, since for the latter, $f_T(t| y)=f_T(t)$ and $F_T(t| y)=F_T(t)$.)

An approximation for the g-renewal function 
%$M_g (t) = F_T(t|0) + \int_0^{t} m(y) F_T(t-y| y)dy$ 
(Eq. \ref{eqn:g-renew}) has been given in \cite{KIJIMA1988194} utilizing the approximation of  \cite{deligonul_1985} for classical renewals.

\section{Bounded renewal process for perfectly repaired NC under limited lifetime of CR \\(Process A).}\label{S:ModelA} 
 
\subsection {Model assumptions.}

(a) The NC and CR start working ``as good as new.''

(b) When NC fails, it is  fully replaced (perfectly repaired). 

(c) When NC fails and is replaced, CR is simply left to operate  as is (with reevaluation of its remaining lifetime or not).

(d) When CR fails, the whole system operation stops (its mission comes to an end).

(e) All replacements/repairs take negligible time.

An example for this model is a battery operated sensor node that monitors some phenomenon, and which communicates with a battery-operated relay node. The sensor node is assumed to be placed at a location that is impractical or costly to access, whereas the relay node is easily accessible. If the battery of the sensor node dies, the system mission terminates. If the battery of the relay node dies while the sensor node is still in operation, it is just replaced immediately. The question then is how many battery replacements of the relay node we will have by a given time $t$. Another example is shown in Section \ref{S:casestudy1}.

An illustration of Process A is given in Figure~\ref{Fig.PA}. In this process, there are two scenarios: (i) Whenever NC fails, CR simply continues to work without any further action; and (ii) When NC fails, CR continues as usual but the distribution of its remaining lifetime is reevaluated based on the time it has spent in operation so far (``memory effect").  By time $t_1$, NC has failed and been fully replaced (or perfectly repaired) four times. At time $t_1$, CR fails and the system operation terminates. Consequently the number of renewals of NC will continue to be  4 for any time $t>t_1$

\begin{figure}[tb]
\centering  %图片全局居中
\subfigure[]{
\label{Fig.PA.1}
\def\svgwidth{230pt}
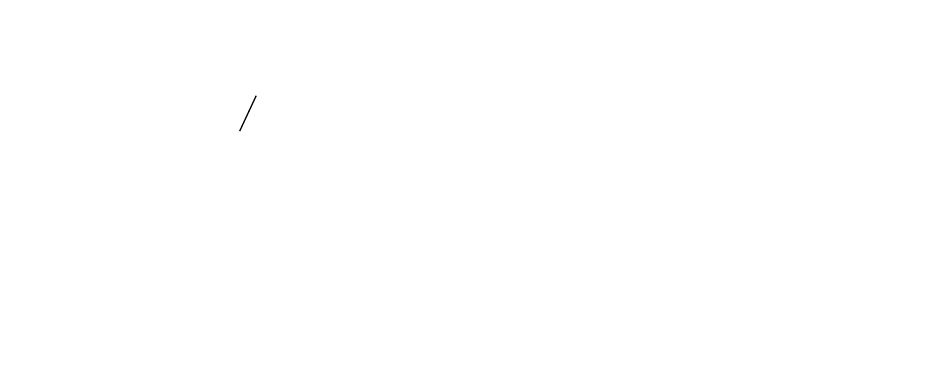
}
\subfigure[]{
\label{Fig.PA.2}
\def\svgwidth{230pt}
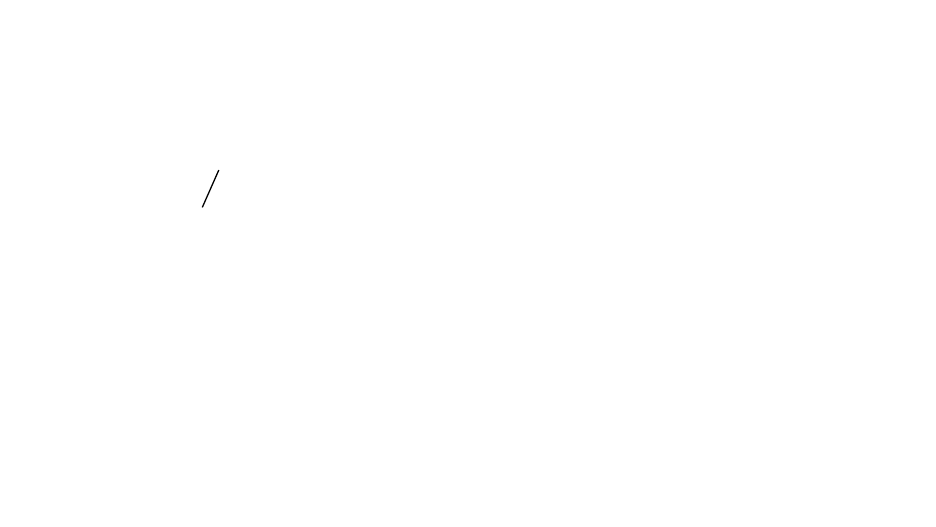}
\caption{Illustration of Process A: (a) Without reevaluation of CR's remaining lifetime after each failure of NC; (b) With reevaluation of CR's remaining lifetime after each failure of NC.}
\vspace{-4mm}\label{Fig.PA}
\end{figure}

Let $T_i$ denote the $i$-th lifetime of NC, and let  $F_{T}(t)$, $\bar F_{ T}(t)$, and $f_{ T}(t)$ denote the underlying CDF, survival function, and pdf of NC. 
Similarly, let $\tilde T$ denote the lifetime of CR, with  CDF $F_{\tilde T}(t)$,  survival function   $\bar F_{\tilde T}(t)= 1-P(\tilde T \hiderel \leq t)=1- F_{\tilde T}(t)$, and pdf  $f_{\tilde T}(t) = d(F_{\tilde T}(t))/dt$.
%  Let $F_{\tilde T}(x|y)$ as the conditional CDF:
%  \begin{dmath} 
% F_{\tilde T}(x|y)= P(\tilde T\hiderel \leq x+y|\tilde T \hiderel \geq y)=\frac{F_{\tilde T}(x+y)-F_{\tilde T}(y)}{\bar F_{\tilde T}(y)},
% \end{dmath}
% the conditional survival function is: 
%  \begin{dmath} \label{eqn:FbarTtilde}
% \bar F_{\tilde T}(x|y)= P(\tilde T\hiderel \geq x+y |\tilde T \hiderel \geq y)=\frac{\bar F_{\tilde T}(x+y)}{\bar F_{\tilde T}(y)},
% \end{dmath}
% and the conditional pdf is:
%  \begin{dmath} 
% f_{\tilde T}(x|y)=\frac{f_{\tilde T}(x+y)}{1-F_{\tilde T}(y)},
% \end{dmath}
% where $F_{\tilde T}(x)$ (pdf $f_{\tilde T}(x)$) is the initial lifetime distribution (real age $S_0=0$). 

\begin{lemma}\footnote{All proofs of lemmas and theorems are given in the Appendix.} \label{lemma:A}
The probability  $P(  N_{A}(t) \geq n)$ of having at least $n$ renewals by time $t$ in Process A does not depend on whether the lifetime of CR is reevaluated or not at every time that NC fails.
\end{lemma}
%{\color{blue}{See Appendix \ref{app:lemma1_proof} for the proof.}}

\subsection {Derivation of renewal function for Process A.}

\begin{theorem}\label{theo:A}
For Process A, the mean number of renewals $W_A(t)$ is :
\begin{dmath}\label{eqn:theorem1}
 W_{A}(t)\!\!=M(t)  \bar F_{\tilde T}(t)+\int_0^t  M(x) f_{\tilde T}(x)dx,
\end{dmath}
where $M(t)=F_T(t) + \int_0^t M(t-x) dF_T(x)$ is the classical renewal function.
\end{theorem}

%{\color{blue}{See Appendix \ref{app:theorem1_proof} for the proof.}}

Note that if CR never fails, 
$ W_{A}(t)$ becomes identical to the classical renewal function $M(t)$ for component NC. 

\subsection {Approximation of bounded renewal function for Process A.}
Given that the  renewal function of Process A (Eq. \ref{eqn:theorem1}) involves the classical renewal function, we approximate it by using  Bartholomew's  approximation $\hat M(t)$ for the classical renewal function $M(t)$:

 %Bartholomew derived the following   approximation $\hat M(t)$ to the classical renewal function $M(t)$:
\begin{dmath}\label{eqn:appB}
M(t)\hiderel \approx \hat M(t)\hiderel= F_T(t)+\int_0^t\frac{F_T^2(x)}{\int_0^x \bar F_T(u)du}dx.
\end{dmath}
Thus, we approximate Eq. \ref{eqn:theorem1} by: 
% by substituting the approximation $\hat M(t)$ for $M(t)$, it is simple to derive an approximation formula $\hat W_A(t)$ for the renewal function of Process A (Eq. \ref{eqn:theorem1}). That is,
\begin{dmath} \label{eqn:approxA1}
\hat W_{A}(t)\!\!=\hat M(t)  \bar F_{\tilde T}(t)+\int_0^t  \hat M(x) f_{\tilde T}(x)dx.
\end{dmath}

\subsection{Numerical Results.}

We used Monte Carlo simulations (MATLAB 2020a platform) to evaluate how well  the derived bounded renewal function for Process A (Eq. \ref{eqn:theorem1}) and its approximation (Eq. \ref{eqn:approxA1}) fit the data.
For the examples in
  Figure \ref{Fig.exampleA}(a) and \ref{Fig.exampleA}(b), the distributions of NC and CR allow for the  computation (via Laplace transforms) of the exact bounded renewal function curve, against which the simulation and approximation curves  can be compared. As can be seen, all are very close together. For the example in Figure \ref{Fig.A.5} and \ref{Fig.A.6}, the exact bounded renewal curve cannot be obtained, and consequently the approximation curve can only be compared to the simulation  curve, which it follows again very closely. 
 
\begin{figure}[tb]
\vspace{-4mm}
\centering  %图片全局居中
\begin{comment}
\subfigure[$T\!\!\sim\!\!$ Exp($\alpha\!\!=\!\!10$); $\tilde T \!\!\sim\!\!$ Exp($\beta\!\!=\!\!0.5$).]{
\label{Fig.A.1}
\includegraphics[width=0.23\textwidth]{A_expexp_10_0.5.pdf}}
\subfigure[$T\!\! \sim\!\!$ Exp($\alpha\!\!=\!\!1$); $\tilde T\!\! \sim\!\!$ Exp($\beta\!\!=\!\!0.1$).]{
\label{Fig.A.2}
\includegraphics[width=0.23\textwidth]{A_expexp_1_0.1.pdf}}
\end{comment}
\subfigure[
$T \!\sim\!$ Gamma($\kappa_1=2, \theta_1=10$); $\tilde T \sim$ Gamma($ \kappa_2=2, \theta_2=0.5$).]{
\label{Fig.A.3}
\includegraphics[width=0.22\textwidth]{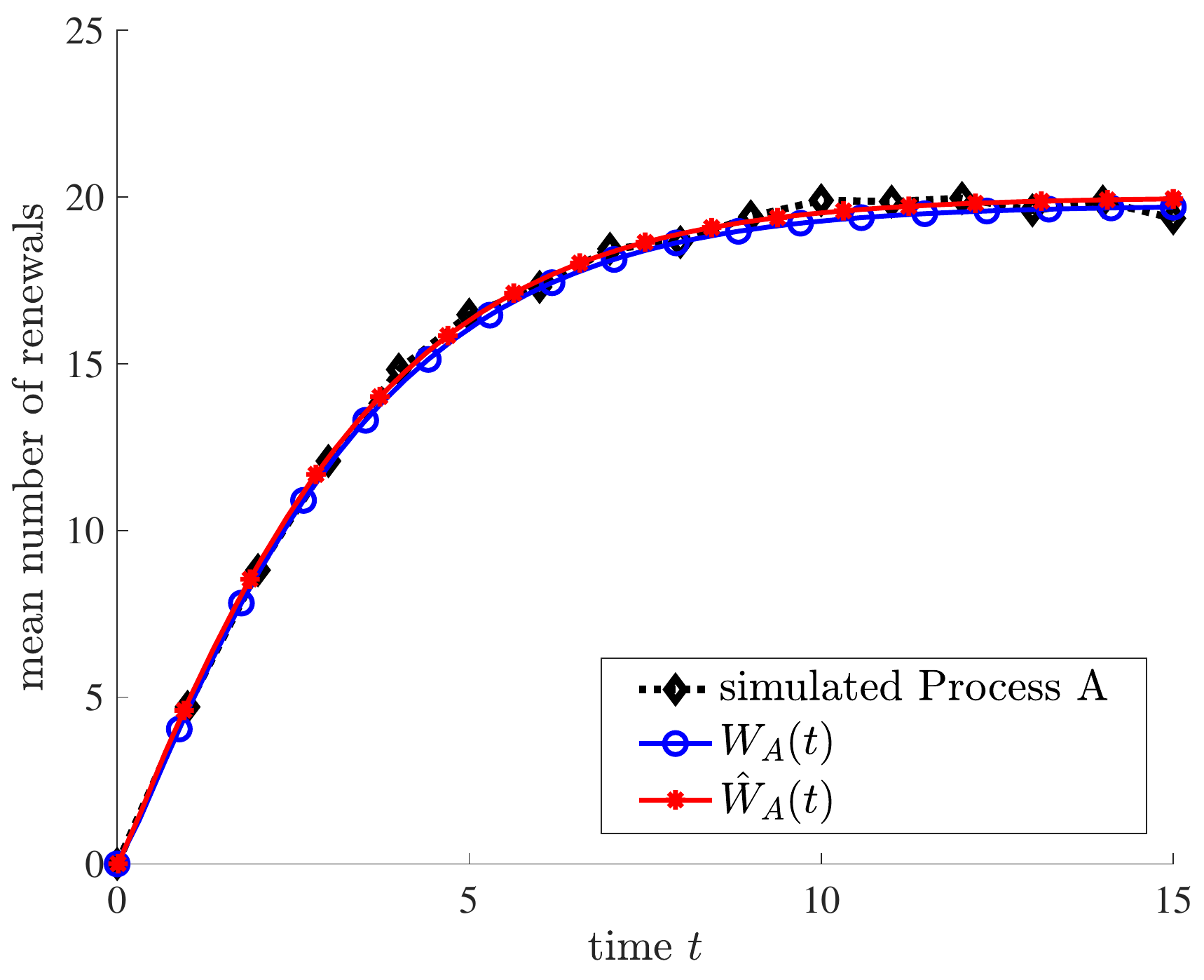}}
\subfigure[$T \sim$ Gamma($\kappa_1=2, \theta_1=1$); $\tilde T \sim$ Gamma($ \kappa_2=2, \theta_2=0.1$).]{
\label{Fig.A.4}
\includegraphics[width=0.23\textwidth]{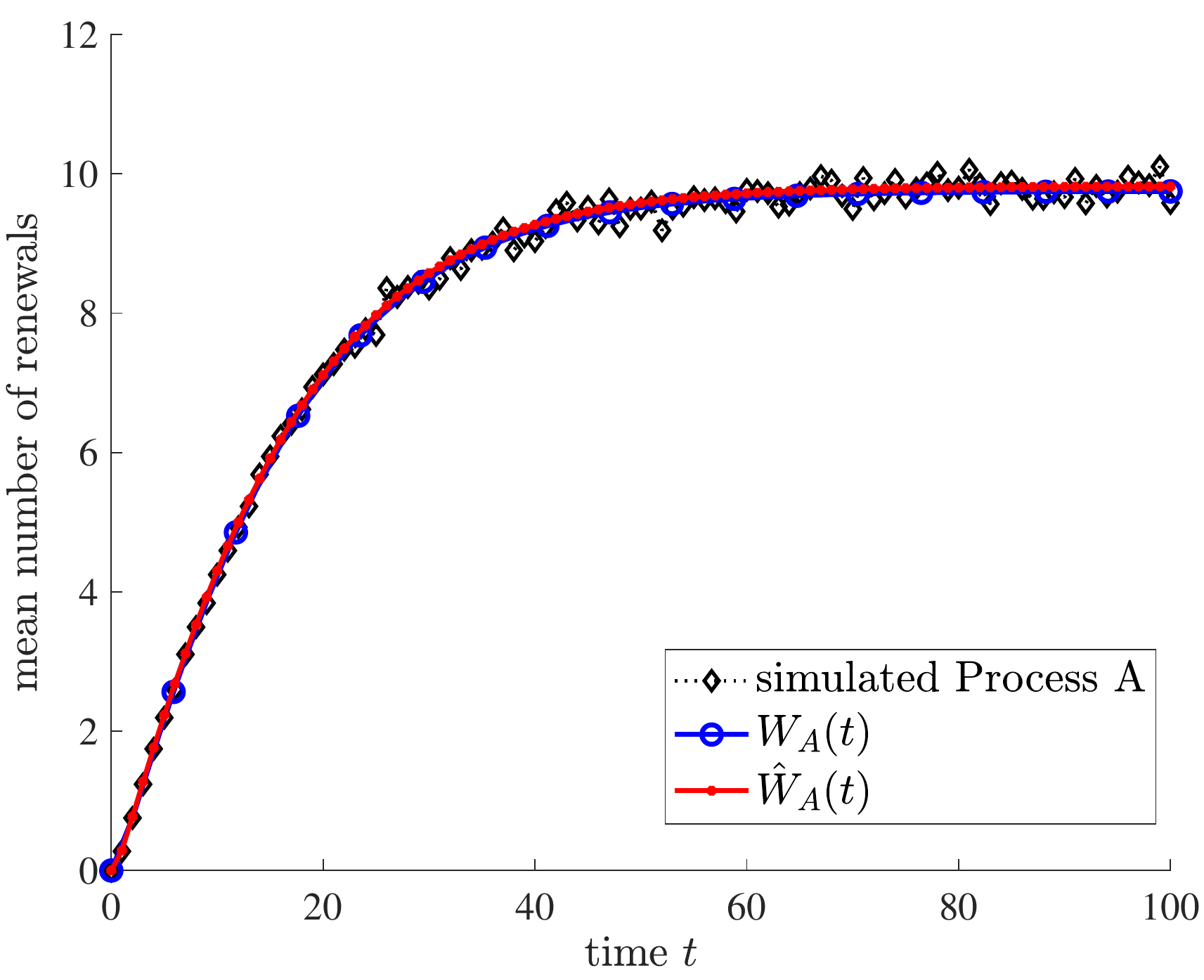}}
\subfigure[
$T \sim$ Rayleigh($\lambda=10$);\qquad \qquad $\tilde T \sim$ Gamma($\kappa_2=2, \theta=0.5$).]{
\label{Fig.A.5}
\includegraphics[width=0.23\textwidth]
{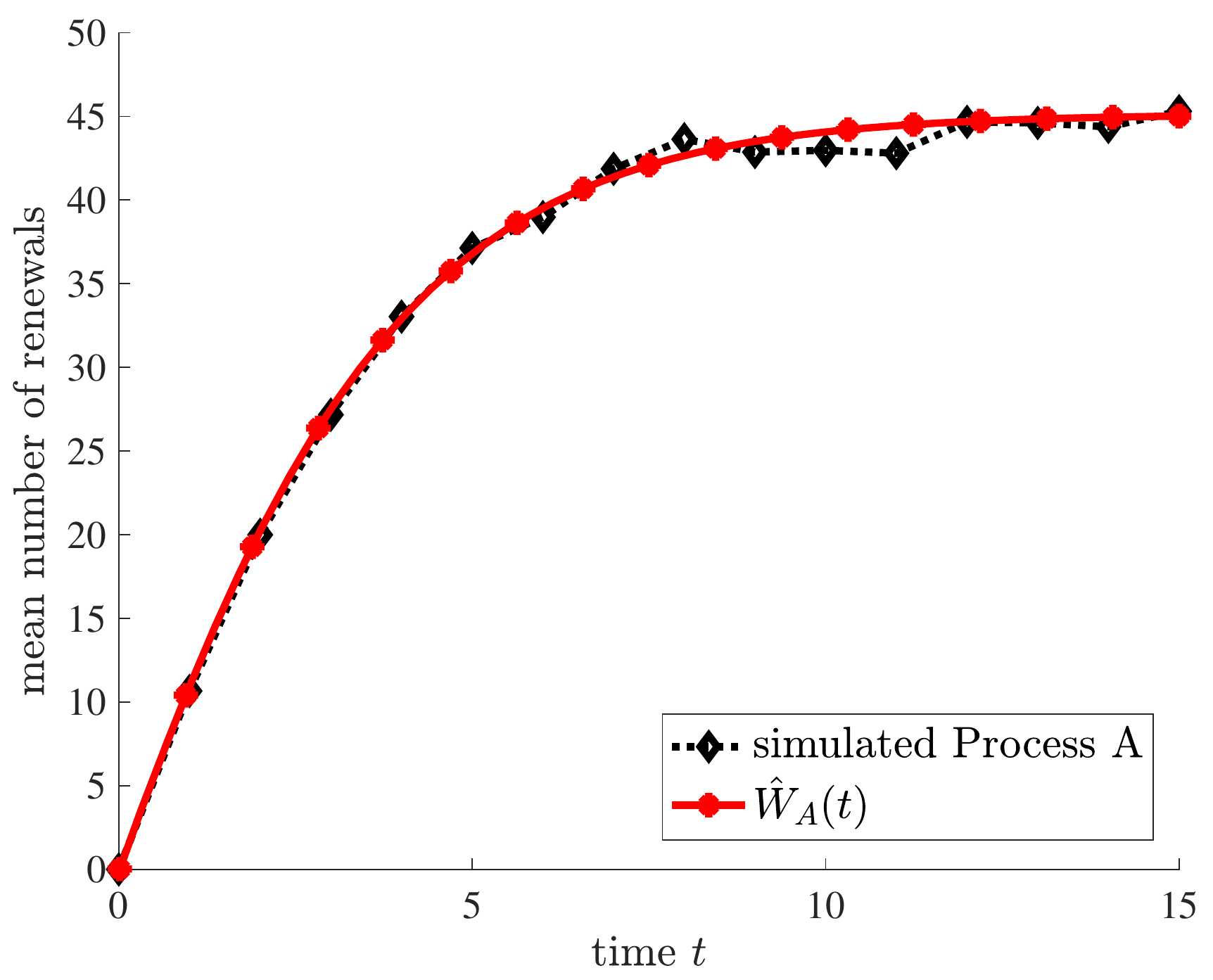}}
\subfigure[$T \sim$ Rayleigh($\lambda=1$);\qquad \qquad$\tilde T \sim$ Gamma($\kappa_2=2, \theta=0.1$).]{
\label{Fig.A.6}
\includegraphics[width=0.23\textwidth]
{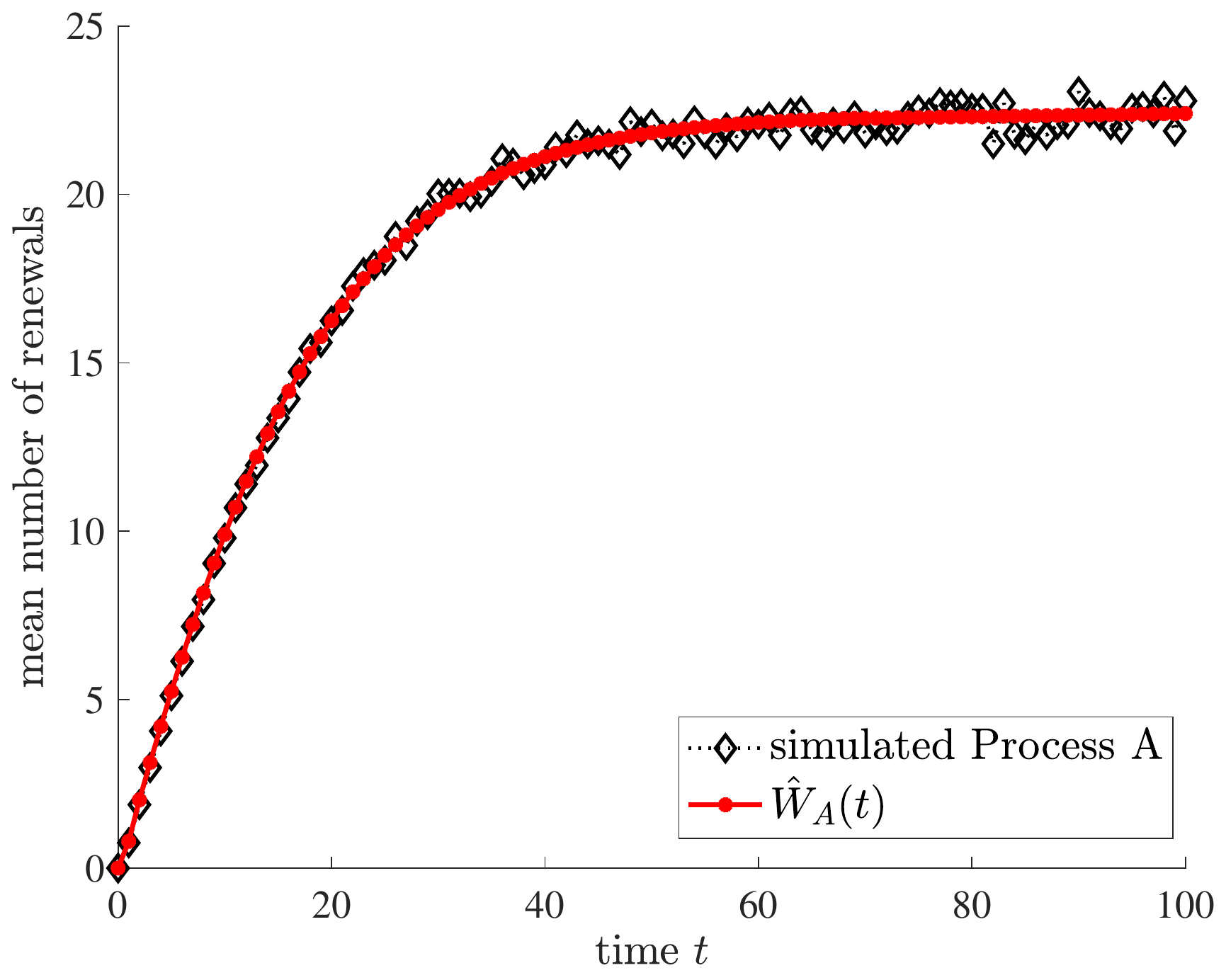}}
\caption{Example plots for Process A.}
\vspace{-6mm}\label{Fig.exampleA}
\end{figure}

\section{Bounded renewal process for  minimally repaired NC under limited lifetime of CR \\ (Process A-MR).}\label{S:ModelAMR} 
 
\subsection {Model assumptions.}

(a) The NC and CR start working ``as good as new.''

(b) When NC fails, it is  minimally repaired (``as bad as old''). 

(c) When NC fails and is minimally repaired, CR is simply left to operate  as is (with reevaluation of its remaining lifetime or not).

(d) When CR fails, the whole system operation stops (its mission comes to an end).

(e) All replacements/repairs take negligible time.

%An example for this model would be a sensor transceiver that is remotely placed and is inaccessible (or accessible with forbidding cost) that communicates with a controller in the operation room. If the sensor transceiver stops working (e.g., due to battery depletion or some other malfunction), the system is considered to have completely failed as the sensing mission cannot be accomplished any more. On the other hand, a malfunction in the controller is  assumed to be easily fixed (e.g., by resetting it) so that it is considered again as-bad-as-old.

The difference between Process A-MR and Process A is that when a failure occurs in NC (denoted by  ``x" in Fig. \ref{Fig.PA}), NC is minimally repaired  whereas it is  fully replaced/perfectly repaired in Process A. In both processes, the treatment of CR is the same (i.e., it simply keeps working  whenever NC fails). Similarly to Lemma \ref{lemma:A}, it can be shown that:

\begin{lemma}\label{lemma:AMR}
The probability  $P(  N_{A_{MR}}(t) \geq n)$ of having at least $n$ renewals by time $t$ in Process A-MR does not depend on whether CR is reevaluated or not at every time that NC fails.
\end{lemma}

\subsection {Derivation of renewal function for Process A-MR.}

Let $S_i$, $i\ge 0$, be the time that NC has been in operation just before its $i$th minimal (and instantaneous) repair. 
The fact that  NC is minimally repaired each time means that its $n$th lifetime  $T_n$ depends on $S_{n-1}$.
As was done in \cite{KIJIMA1988194}, we assume this dependence is only on the value $y=S_{n-1}$ and independent of $n$. That is, the corresponding CDF   is $F_{ T}(x|y)=P( T_{n}\hiderel \leq x | S_{n-1}= y)=P( S_{n}\hiderel \leq x+y | S_{n}>S_{n-1}= y)=\frac{F_{ T}(x+y)-F_{ T}(y)}{\bar F_{ T}(y)}$.  
Likewise, 
the corresponding survival function and pdf are: 
$
\bar F_{ T}(x|y)=\frac{\bar F_{ T}(x+y)}{\bar F_{ T}(y)},
$
and 
$
f_{ T}(x|y)=\frac{f_{T}(x+y)}{1-F_{ T}(y)},
$
where $F_{ T}(x)$ (pdf $f_{ T}(x)$) is the initial lifetime  distribution of NC. The corresponding functions for CR remain the same as in Process A.

%We are now giving the renewal function of ``Process A-MR".

\begin{theorem}\label{theo:AMR}
The mean number of renewals $  W_{A_{MR}}(t)$ for Process A-MR is given by:
\begin{dmath} \label{eqn:AMR}
 W_{A_{MR}}(t)\!\!=M_g (t)  \bar F_{\tilde T}(t)+\int_0^t  M_g(x) f_{\tilde T}(x)dx
.\end{dmath}
where $M_g (t)=F_T(t|0) + \int_0^t M_g(y) f_T(t-y|y)dy$ is the g-renewal function that involves NC only.
\end{theorem}

%{\color{blue}{See Appendix \ref{app:theorem2_proof} for the proof.}}

Note that if the failure rate (or hazard rate) for CR is zero, the $ W_{A_{MR}}(t)$ equals the g-renewal function $M_g (t)$ involving only NC. Note also that $ W_{A_{MR}}(t)$ (in Eq. \ref{eqn:AMR}) generalizes the $ W_A(t)$ (in Eq. \ref{eqn:theorem1}), as $M_g(t)$ generalizes  classical renewal function.

\subsection {Approximation of bounded renewal function for Process A-MR.}

% For the purpose of this work, we approximate using a less expensive computational method based on \cite{Bartholomew}. 
An approximation for the g-renewal function has been given in \cite{KIJIMA1988194}, but this is still relatively expensive to compute. In our context, we use an easier to compute approximation denoted by  $\hat M_g (t)$:
 \begin{eqnarray}\label{eqn:approxg}
M_g (t) \!\!\approx \!\!\hat M_g (t)\!\! = \!\!F_T(t|0) \!\!+\!\!\int_0^t \!\!\frac{F_T(u|0) \int_0^{u}  f_T(u\!-\!y| y)dy}{\int_0^{u} \bar F_T(u\!-\!y| y)dy}du.
\end{eqnarray}

The unknown functions are no longer included in the renewal function. By substituting $\hat M_g (t)$ for $M_g (t)$, we get an approximation formula $ \hat W_{A_{MR}}(t)$ for the renewal function of Process A-MR (Eq. \ref{eqn:AMR}) as follows:
\begin{dmath} \label{eqn:approxAMR}
 W_{A_{MR}}(t)\approx \hat W_{A_{MR}}(t)\!\!\hiderel=\hat M_g (t)  \bar F_{\tilde T}(t)+\int_0^t  \hat M_g(x) f_{\tilde T}(x)dx
.\end{dmath}

\subsection{Numerical results.}
Examples are shown in Figure \ref{Fig.ExampleAMR}. No exact renewal function can be determined for all combinations of lifetimes in these situations. The proposed approximations, as can be seen, closely match the simulated outcomes.

\begin{figure}[tb]
\centering  %图片全局居中
\begin{comment}
\subfigure[ {$T \sim$ Gamma($\theta=10, \kappa=2$);\\ $\tilde T \sim$ Exp($\beta=0.5$).}]{
\label{Fig.ExampleAMR.1} 
\includegraphics[width=0.23\textwidth]{A_MR_gam_exp_10_0.5.pdf}}
\subfigure[$T \sim$ Gamma($\theta=1, \kappa=2$);\\$\tilde T \sim$ Exp($\beta=0.1$).]{
\label{Fig.ExampleAMR.2}
\includegraphics[width=0.23\textwidth]{A_MR_Gam_Exp_1_0.1.pdf}}
\end{comment}
\subfigure[ $T \sim$ Gamma($ \kappa_1=2, \theta_1=10$); $\tilde T \sim$ Gamma($\kappa_2=2, \theta_2=0.5$).]{
\label{Fig.ExampleAMR.3} 
\includegraphics[width=0.23\textwidth]{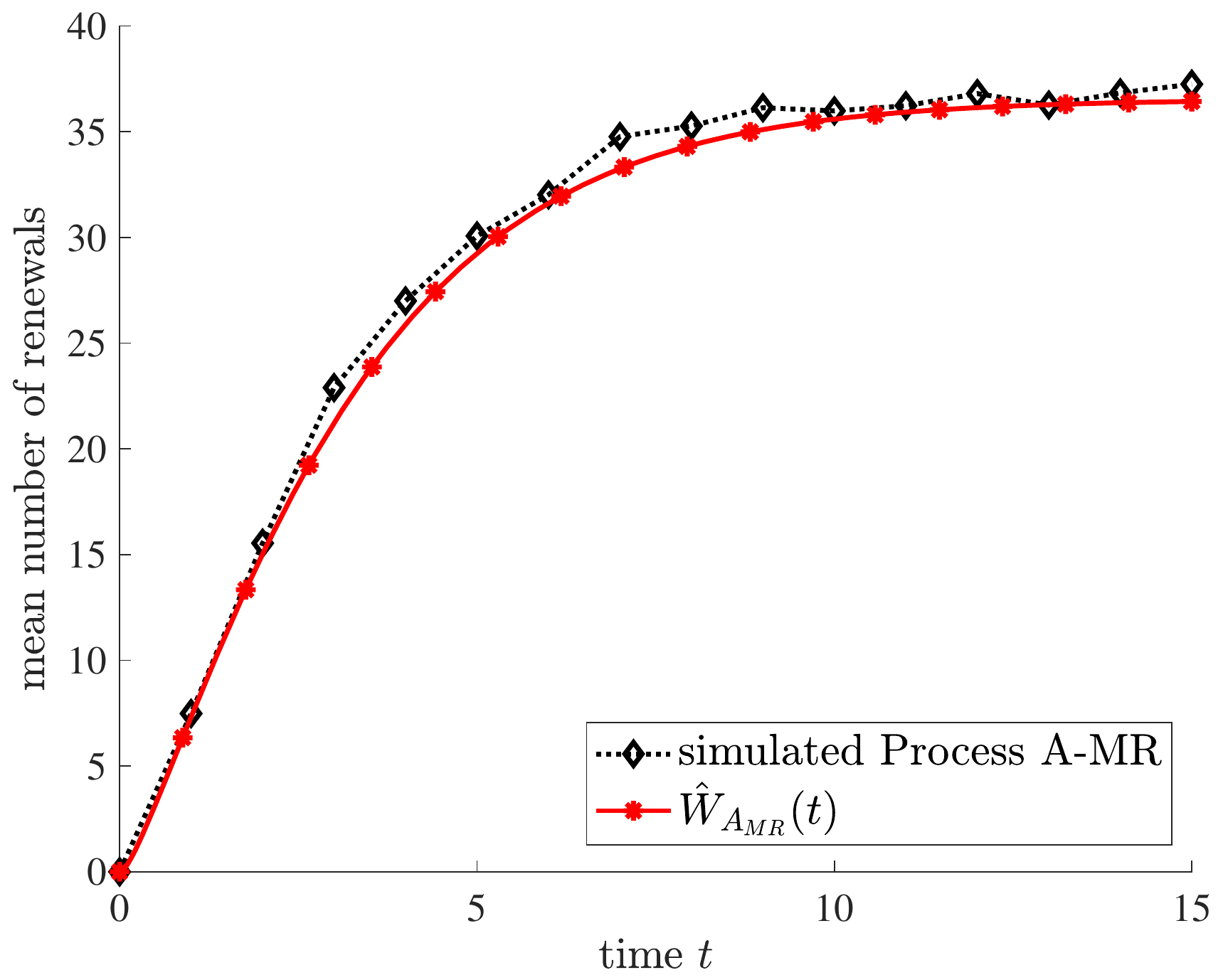}}
\subfigure[$T \sim$ Gamma($ \kappa_1=2, \theta_1=1$); $\tilde T \sim$ Gamma($ \kappa_2=2, \theta_2=0.1$).]{
\label{Fig.ExampleAMR.4}
\includegraphics[width=0.23\textwidth]{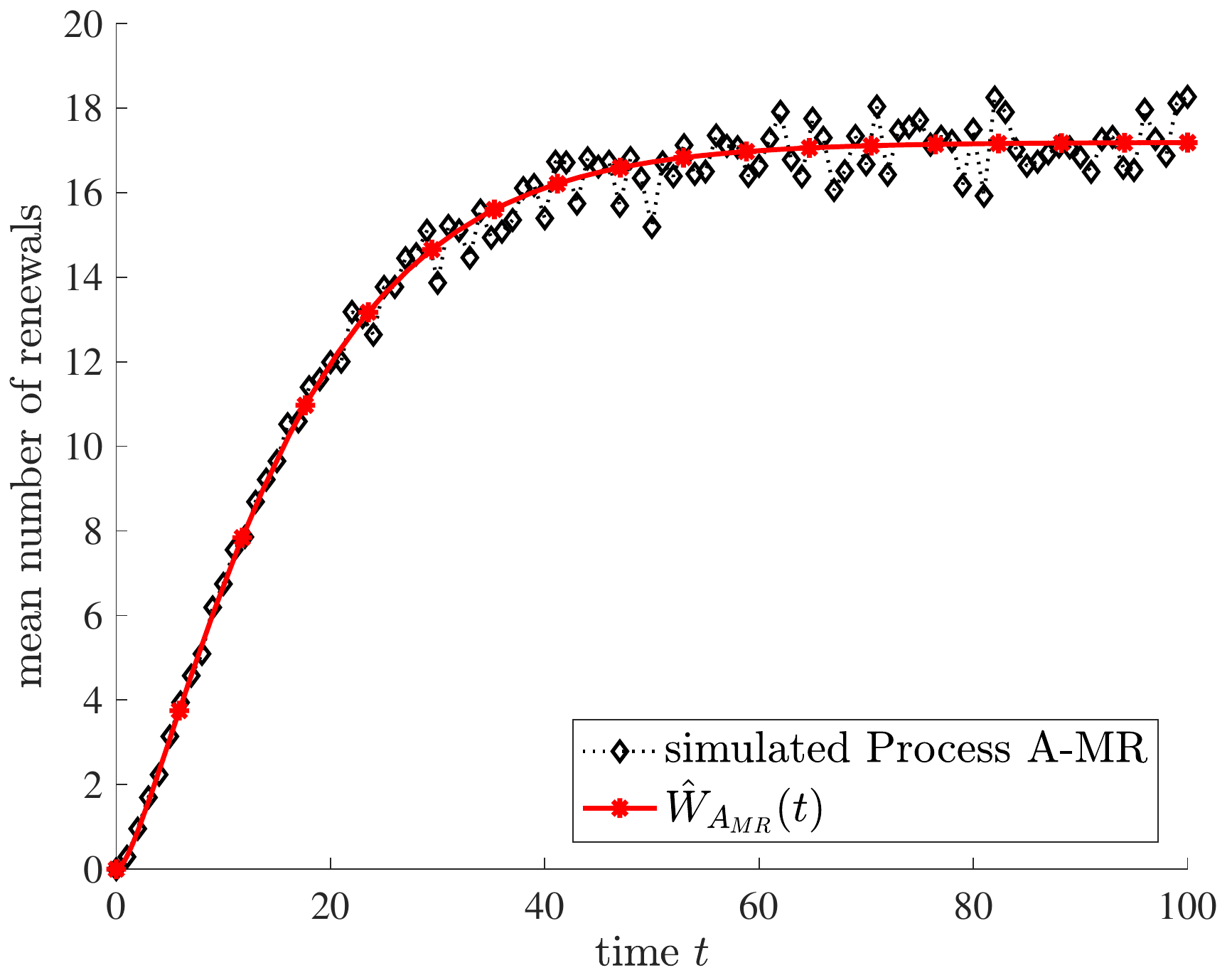}}
\subfigure[$T \sim$ Rayleigh($\lambda=0.1$);\qquad \qquad $\tilde T \sim$ Gamma($ \kappa=2, \theta=0.1$).]{
\label{Fig.ExampleAMR.5}
\includegraphics[width=0.23\textwidth]{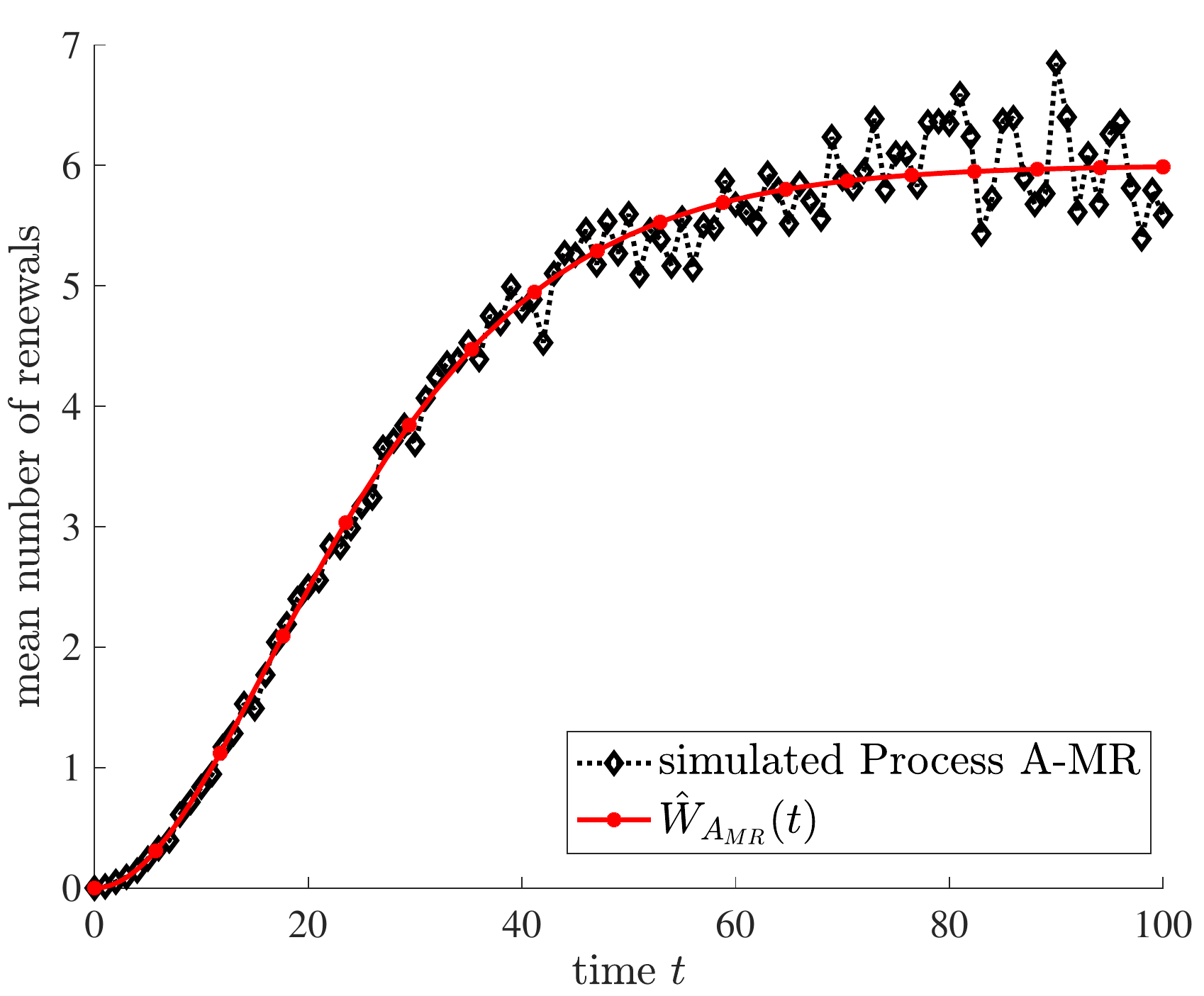}}
\subfigure[$T \sim$ Rayleigh($\lambda=0.1$);\qquad \qquad $\tilde T \sim$ Gamma($ \kappa=2, \theta=0.3$).]{
\label{Fig.ExampleAMR.6}
\includegraphics[width=0.23\textwidth]{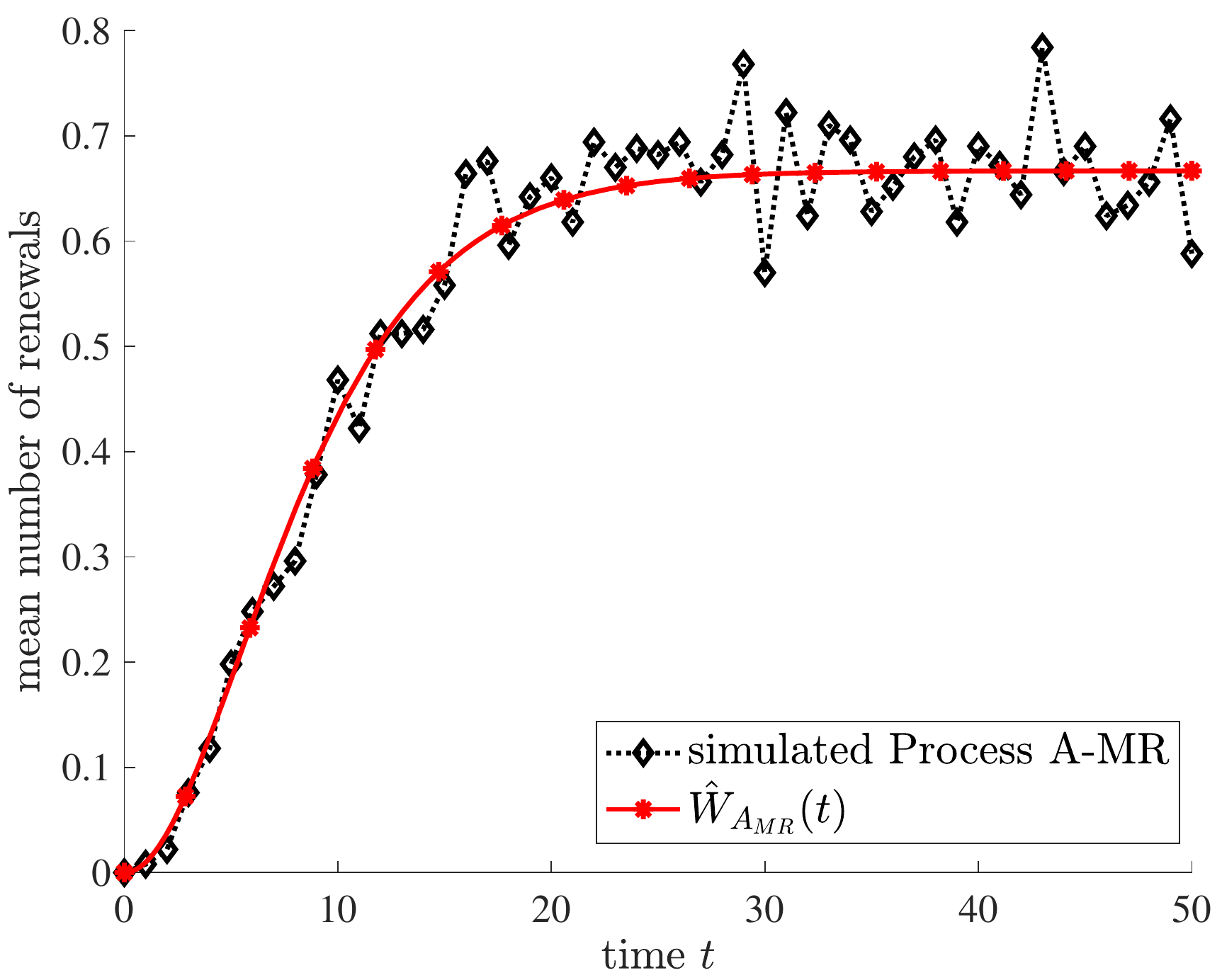}}
\caption{Example plots for Process A-MR}
\label{Fig.ExampleAMR}
\end{figure}

\section{Bounded renewal process for perfectly repaired NC with preventively replaced CR \\(Process B).}\label{S:ModelB} 
\subsection {Model assumptions.}

(a) The NC and CR start working ``as good as new.''

(b) When NC fails, it is  fully replaced (perfectly repaired). 

(c) When NC fails and is fully replaced, CR is also fully replaced in a preventive fashion.

(d) When CR fails, the whole system operation stops (its mission comes to an end).

(e) All replacements/repairs take negligible time.

An example for this model is a battery operated sensor node that monitors some phenomenon, and which communicates with a battery-operated relay node. Both the sensor node and the relay node are easily accessible but the sensor node is assumed to be critical for the system mission and when its battery dies, the system mission terminates. If the battery of the relay node dies while the sensor node is still in operation, the batteries of both the relay and the sensor node are replaced (correctively for the relay node, and preventively for the sensor node). 
%An example for this model would be a sensor transceiver that   is accessible and communicates with a microcontroller in the operation room. If the sensor transceiver stops working, the system is considered to shut down since the sensing mission cannot be accomplished any more.
%Whenever a failure occurs in the microcontroller, we assume it is fixed by replacement  by a new one, and at the same time  the sensor is also replaced (preventively) with a new one so that the sensing mission can be prolonged. 
Another example is shown in Section \ref{S:casestudy2}.

\subsection {Derivation of renewal function for Process B.}
An illustration of Process B is given in Fig.  \ref{fig:system_model2}, 
Every time  NC experiences a failure (marked by ``x" in Fig. \ref{fig:system_model2}), NC is fully repaired/replaced but in contrast with Process A, CR is preventively maintained (fully replaced) at the same time (marked by a circle in Fig. \ref{fig:system_model2}). 
Four renewals have occurred up to time $t_1$ in Fig. \ref{fig:system_model2}, because, in this case, CR stops working before NC dies. For any other time after $t_1$, the number of renewals of NC  remains at 4.

\begin{figure}[!t]\centering \vspace{-4mm}
  	\def\svgwidth{230pt} 
  	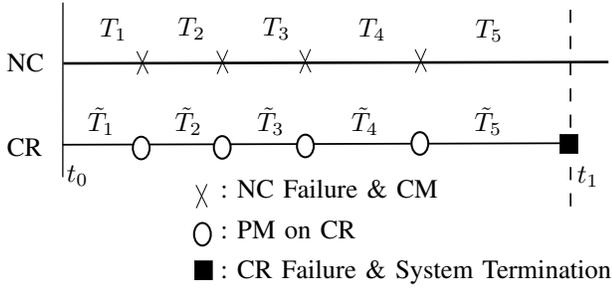 \vspace{-7mm}
  	\caption{{{ Illustration of Process B}}}\vspace{-4mm} \label{fig:system_model2}
  \end{figure}

Let the $i$-th lifetime of CR, be noted by $\tilde T_i$ with CDF, survival function, and pdf denoted by $F_{\tilde T}(t) = P(\tilde T_i \hiderel \leq t)$, $\bar F_{\tilde T}(t)= 1-P(\tilde T_i \hiderel \leq t)=1- F_{\tilde T}(t)$, and $f_{\tilde T}(t) = d(F_{\tilde T}(t))/dt$. The corresponding functions for  NC are denoted by $F_{T}(t)$, $\bar F_{ T}(t)$, and $f_{ T}(t)$.

\begin{theorem}\label{theo:B}
For Process B, the mean number of renewals $W_B(t)$ is:
\begin{dmath}\label{eqn:WAtGolden}
 W_{B}(t)=  F_{B}(t)+\int_0^t  W_{B}(t-x) f_{B}(x)dx
,\end{dmath}
where $ F_{B}(t)\!\!=\!\! \int_0^tf_T(x)\bar F_{ \tilde T}(x)dx$, and $ f_{B}(x)\!\!=\!\! f_T(x)\bar F_{ \tilde T}(x)$.
\end{theorem}

%{\color{blue}{See Appendix \ref{app:Theorem3_proof} for the proof.}}
%The renewal density function can be written as:
%\begin{dmath}\label{eqn:wctgolden}
 %w_{B}(t)= f_{B}(t)+\int_0^t  w_{B}(t-x) f_{B}(x)dx.\end{dmath}
%where $ F_{B}(t)\!\!=\!\! \int_0^tf_T(x)\bar F_{ \tilde T}(x)dx$, and $ f_{B}(t)\!\!=\!\! f_T(t)\bar F_{ \tilde T}(t)$.\qed

\subsection {Approximation for renewal function of Process B.}

The approximation for the average number of renewals in Process B (Eq. \ref{eqn:WAtGolden}) is more complicated than that for Process A and Process A-MR, since Eq. \ref{eqn:WAtGolden} is recursive (in terms of $ W_{B}(t)$). Moreover, approximations that have been proposed for the classical renewal function (see section \ref{S:class}) are not applicable in this ``bounded" renewal case, as the curve of classical renewals keeps increasing forever, whereas the curve of bounded renewals eventually flattens. 
 In the following, we show how to approximate Eq. \ref{eqn:WAtGolden} of Theorem \ref{theo:B}. We present first a modification to  the approximate solution of  \cite{Bartholomew} for  the classical renewal function (Eq.  \ref{eqn:class}) that involves NC alone.Then, in order to achieve a better degree of accuracy, we propose a more precise approximation for process B.

%Assume that we have $n$ independent trials, and each trial has two possible outcomes $\mathcal B$ and $\mathcal B^*$, where $\mathcal B$ denotes the event 
Let $p_B$ denote the probability that the lifetime of  NC $T$ is less than that of CR $\tilde T$ , i.e.,
\begin{eqnarray}\label{eqn:PB}
 p_B\!\!\!\!\!\!\!\!\!\!&&=P( \tilde T > T)
=\int_0^\infty P( \tilde T  > x) f_{ T}(x) dx \nonumber\\
&& = \int_0^\infty  \bar F_{\tilde T}(x)f_{ T}(x)dx
.\end{eqnarray}

Let $\hat G_{B}(t)=\int_t^\infty f_T(x)\bar F_{ \tilde T}(x)dx$. Since 
\begin{dmath}
\!\!\int_t^\infty \!\!\!\!\!\!f_T(x)\bar F_{ \tilde T}(x)dx\!+\!\int_0^t\!\!\!f_T(x)\bar F_{ \tilde T}(x)dx\!\hiderel=\!\!\int_0^\infty\!\!\!\!\!\! f_T(x)\bar F_{ \tilde T}(x)dx,
\end{dmath}
we have that $\hat G_{B}(t)+\hat F_{B}(t)=p_{B}$.

Then Eq. \ref{eqn:WAtGolden} can  be rewritten as:
\begin{dmath}
 W_{B}(t)
= F_{B}(t)\hiderel+\!\!\int_0^t \!\! w_{B}(t\!\!\hiderel-\!\!x)(p_{B}\!\!\hiderel- G_{B}(x))dx
=p_{B} W_{B}(t)\!\!\hiderel+ F_{B}(t)\hiderel-\int_0^t\!\! G_{B}(x)  w_{B}(t\!\!\hiderel-\!\!x) dx.
\end{dmath}

Since moreover
\begin{dmath}
\frac{F_{B}(t)+(p_{B}-1) W_{B}(t)}{\int_0^t   w_{B}(t-x) G_{B}(x)dx}=1,
\end{dmath}
we have that 
\begin{dmath}
 W_{B}(t)= F_{B}(t)+R\cdot\Big((p_{B}-1) W_{B}(t)+ F_{B}(t)\Big)
,\end{dmath}
where
\begin{dmath}
R=\frac{\int_0^t  w_{B}(t-x) F_{B}(x)dx}{\int_0^t  w_{B}(t-x) G_{B}(x)dx}
.\end{dmath}

Our approximation is obtained  by substituting  $R$ with  $R'$
given by:
\begin{dmath}
R'=\frac{\int_0^t F_{B}(x)dx}{\int_0^t  G_{B}(x)dx}
.\end{dmath}

 As a result, the function $\hat W_{B_0}(t)$ presented below can be used to approximate the renewal function Eq.\ref{eqn:WAtGolden}:
\begin{dmath} \label{eqn:BarWA}
 \hat W_{B_0}(t)=\frac{ F_{B}(t)+R' F_{B}(t)}{1+(1-p_{B})R'}
.\end{dmath}

As shown by simulation results (subsection \ref{sec:exampleB}), the approximation $ \hat W_{B_0}(t)$  is not very precise. A far better approximation is obtained by the theorem below (this theorem has a long proof given in the Appendix, and it is based on an intermediate approximation $\hat W_{B_1}(t) =\frac{p_B}{1-p_B}(1-\E{p_B^{N(t)}})$ obtained in Lemma 3 in the Appendix.)

%Define $X_t$ to be a random variable that indicates the number of renewals in "Process B" at a time of $t$. If we assume that the probability $P(\mathcal B) = p_B$ is the same in all $n$ trials, there is a chance that the NC could survive the chance of the possible outcome $\mathcal B^*$ until time $t$ by working collaboratively with CR, which can be expressed as:

%\subsubsection{Approximated renewal function 1}

%{\color{blue}{See Appendix \ref{app:lemma5_proof} for the proof.}}

  \begin{theorem}\label{theo:Bapp}
The renewal function $W_B(t)$ for Process B can be approximated by:
\begin{dmath}\label{eqn:TB3}
  \hat W_{B_2}(t) =\frac{p_B}{1-p_B}\Big(1-(p_B)^{M(t)}\Big)
,\end{dmath}
where $M(t)=\E{N(t)} $ denotes  the classical renewal function (Eq. \ref{eqn:class}) involving only component NC.
\end{theorem}

%{\color{blue}{See Appendix \ref{app:Theorem4_proof} for the proof.}}

%\begin{lemma}
%$ \hat W_{B_3}(t)$ can be also approximately expressed by
%\begin{dmath}\label{eqn:WA4F}
% \hat W_{B_3}(t)
%\approx M(t)p_B^{\lceil M(t) \rceil}+\sum_{n=0}^{\lceil M(t) \rceil-1}np_B^n(1-p_B) \;\;\text{for}\;\;t\hiderel>0,
%.\end{dmath}
%\end{lemma}
%\begin{proof}

% Since $\frac{p_B}{1-p_B}(1-p_B^{M(t)})$ can be approximated by the summation for geometric sequence, that is $\sum_{n=1}^{\lceil M(t) \rceil}p_B^n$ (where $ \lceil M(t) \rceil$ is used to obtain the discrete value), $ \hat W_{B_3}(t)$ can be re-written as
% \begin{dmath}\label{eqn:WA4F}
%  \hat W_{B_3}(t)=\frac{p_B}{1-p_B}(1-p_B^{M(t)})\hiderel\approx\sum_{n=1}^{\lceil M(t) \rceil}p_B^n=\sum_{n=1}^{\lceil M(t) \rceil}np_B^n-\sum_{n=1}^{\lceil M(t) \rceil}(n-1)p_B^n
% \approx M(t)p_B^{\lceil M(t) \rceil}+\sum_{n=0}^{\lceil M(t) \rceil-1}np_B^n-\sum_{n=0}^{\lceil M(t) \rceil-1}np_B^{n+1}
% =M(t)p_B^{\lceil M(t) \rceil}+\sum_{n=0}^{\lceil M(t) \rceil-1}np_B^n(1-p_B) \;\;\text{for}\;\;t\hiderel>0,
% .\end{dmath}
% 
% \noindent and $ \hat W_{B_3}(t)=0$ when $t=0$.
% 
% \end{proof}
% 

We note that for the special case where the lifetime of CR follows the exponential distribution (with the distribution of the lifetime of NC being still arbitrary), the renewal functions of   Process A and Process B become identical. This is of course due to the fact that  the lack of memory in the exponential distribution cannot distinguish between ongoing operations and newly initiated operations.

%{\color{blue}{See Appendix \ref{app:Theorem5_proof} for the proof.}}

\subsection{Numerical results.}\label{sec:exampleB}
The proposed renewal function for Process B (Eq. \ref{eqn:WAtGolden}) and its approximations as well as the intermediate approximation $\hat W_{B_1}(t) =\frac{p_B}{1-p_B}(1-\E{p_B^{N(t)}})$ (Lemma 3 in the Appendix) are evaluated numerically with respect to simulated results. For the examples in Figure  \ref{Fig:ExampleB}(a-b), the specific distributions of NC and CR allow for the computation (via
Laplace transforms) of the exact bounded renewal function
curve, against which the simulation and approximation curves
can be compared.  We first see that the improvement over the Bartholomew-based approximation of Eq. \ref{eqn:BarWA} is dramatic. By zooming in the figure, we can see that $W_B(t)>\hat W_{B_2}(t)>\hat W_{B_1}(t)$. For the examples in Figure \ref{Fig:ExampleB}(c-d), the simulation curves can only be compared against the approximation curves, since the lifetime distributions of NC and CR make it impossible to calculate the exact renewal functions. This emphasizes of course the significance of having an approximation to work with.
By zooming in the figure, we can see that $W_B(t)>\hat W_{B_2}(t)>\hat W_{B_1}(t)$. 

\begin{figure}[t]
 \vspace{-3mm}
\centering  %图片全局居中
\begin{comment}
\subfigure[$T \!\!\sim\!\! \text{Exp}(\alpha\!\!=\!\!1)$, $\tilde T \!\!\sim \!\!\text{Exp}(\beta\!\!=\!\!0.05)$.]{
\label{Fig.ExampleB.1} 
\includegraphics[width=0.23\textwidth]{B_expexp_1_0.05.pdf}}
\subfigure[$T \!\!\sim\!\! \text{Exp}(\alpha\!\!=\!\!1)$, $\tilde T \!\!\sim\!\! \text{Exp}(\beta\!\!=\!\!0.1)$.]{
\label{Fig.ExampleB.2} 
\includegraphics[width=0.23\textwidth]{B_expexp_1_0.1.pdf}}
\end{comment}
\subfigure[$T \!\sim\! \text{Gamma}( \kappa_1\!=\!2,\theta_1\!=\!1)$, $\tilde T\! \sim\! \text{Gamma}( \kappa_2\!=\!2, \theta_2\!=\!0.05)$.]{
\label{Fig.ExampleB.3} 
\includegraphics[width=0.23\textwidth]{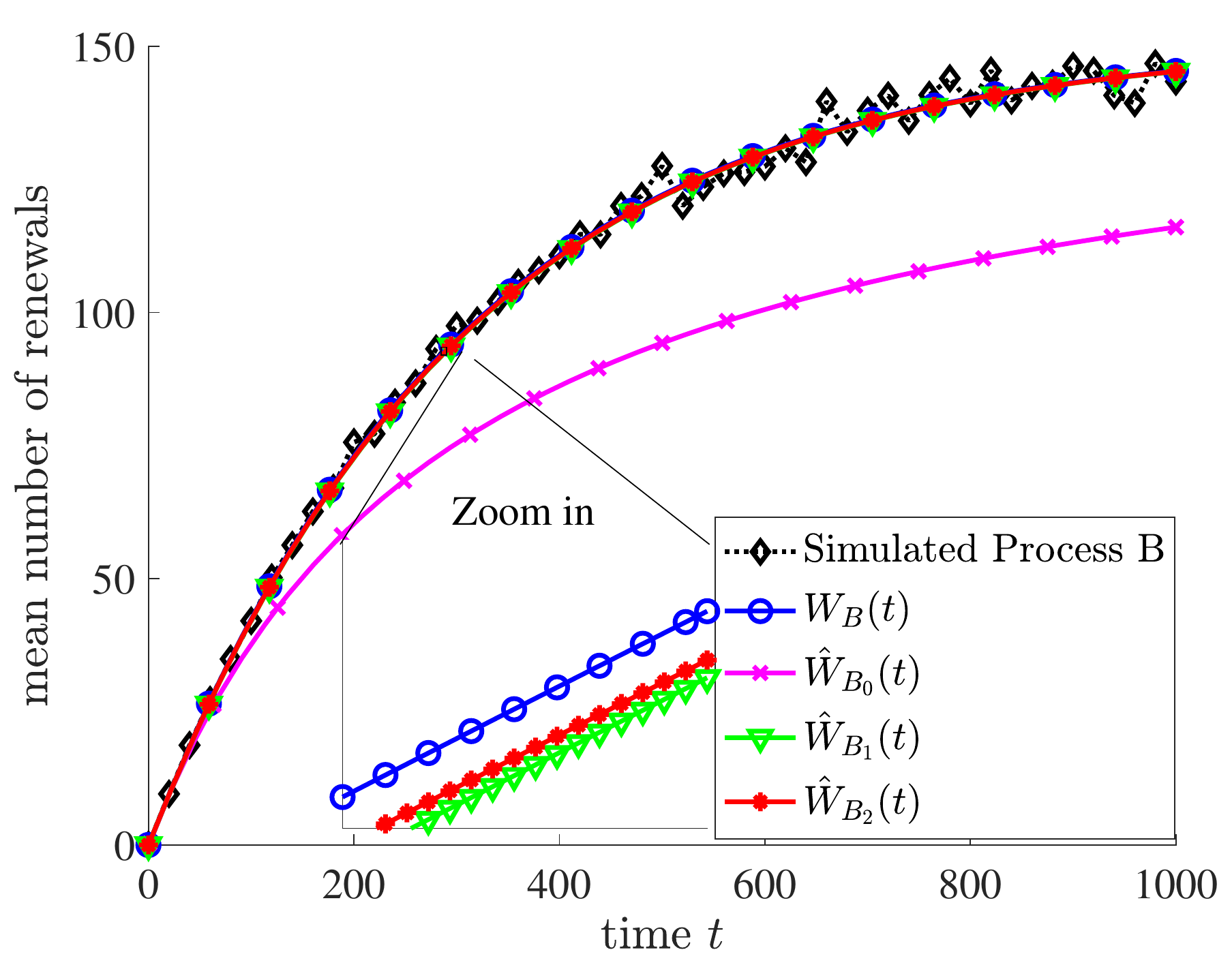}}
\subfigure[$T \!\sim\! \text{Gamma}( \kappa_1\!=\!2,  \theta_1\!=\!1)$, $\tilde T\! \sim\! \text{Gamma}( \kappa_2\!=\!2, \theta_2\!=\!0.1)$.]{
\label{Fig.ExampleB.4} 
\includegraphics[width=0.23\textwidth]{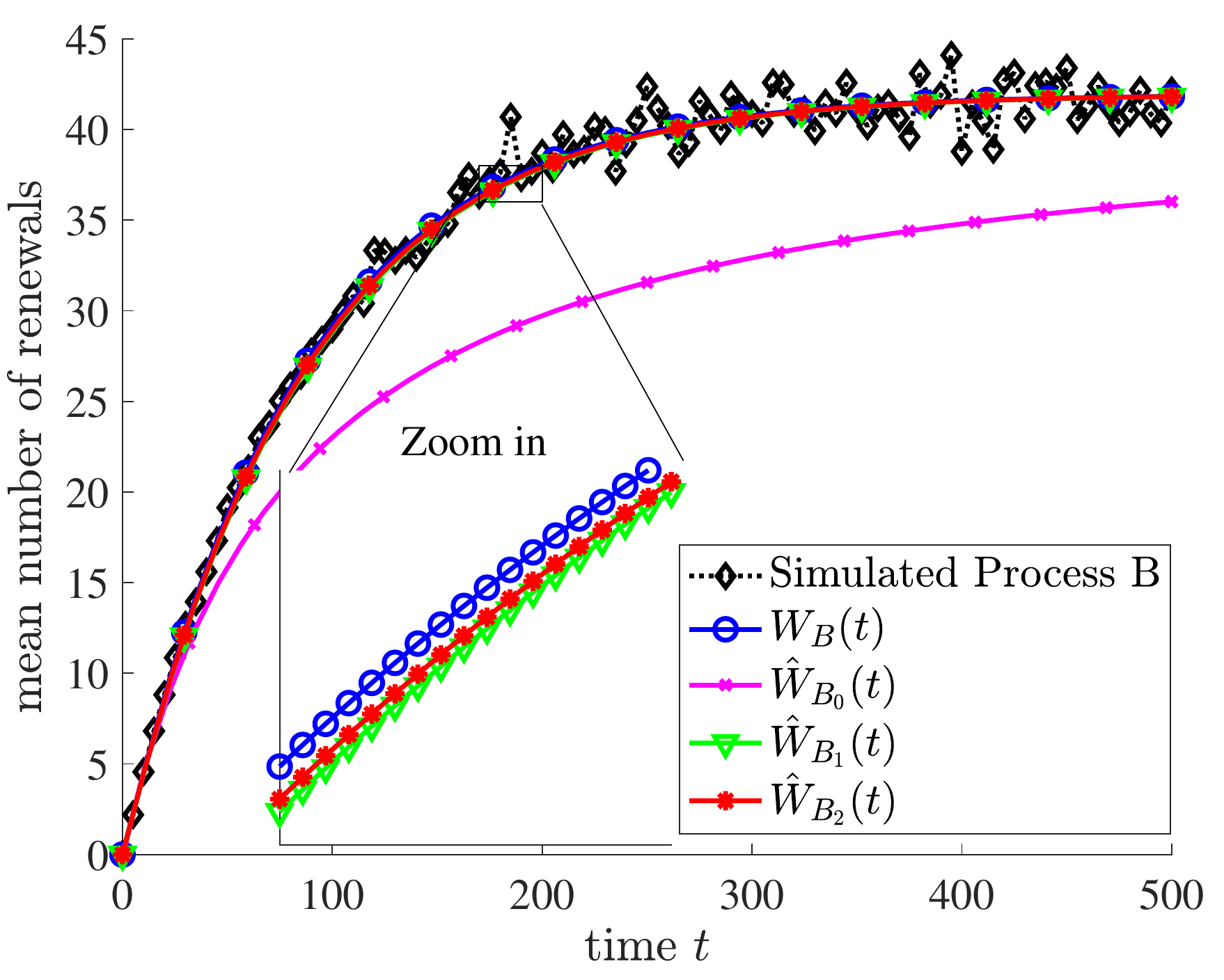}}
\subfigure[$T \!\sim\! \text{Gamma}( \kappa\!=\!2, \theta\!=\!1)$, $\,\,\,\,\,\,\,\,\,\,\,\tilde T\! \sim\! \text{Rayleigh}(\lambda\!=\!0.05)$.]{
\label{Fig.ExampleB.5}
\includegraphics[width=0.23\textwidth]{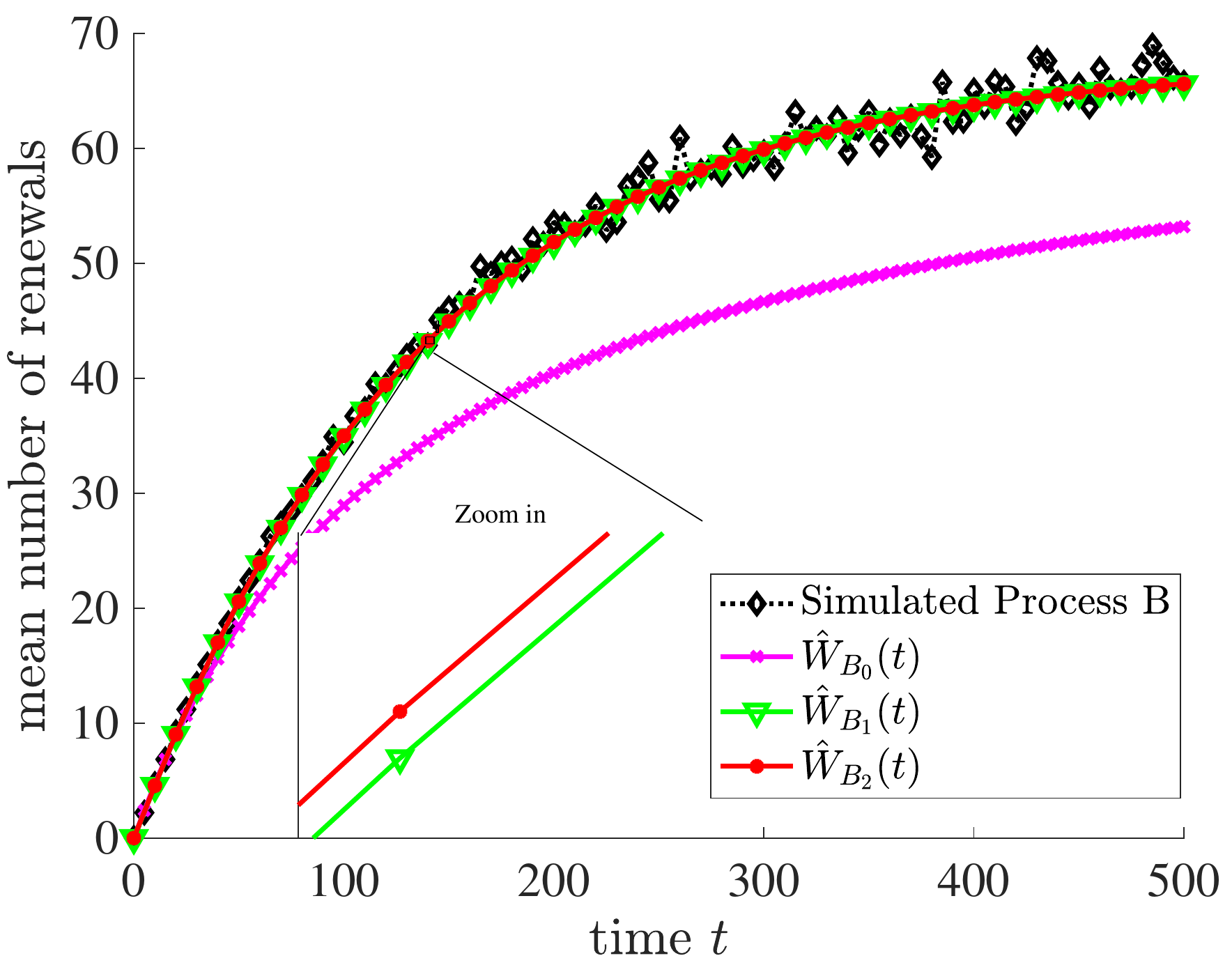}}
\subfigure[$T \!\sim\! \text{Gamma}( \kappa\!=\!2, \theta\!=\!1)$, $\,\,\,\,\,\,\,\,\,\,\,\tilde T\! \sim\! \text{Rayleigh}(\lambda\!=\!0.1)$.]{
\label{Fig.ExampleB.6}
\includegraphics[width=0.23\textwidth]{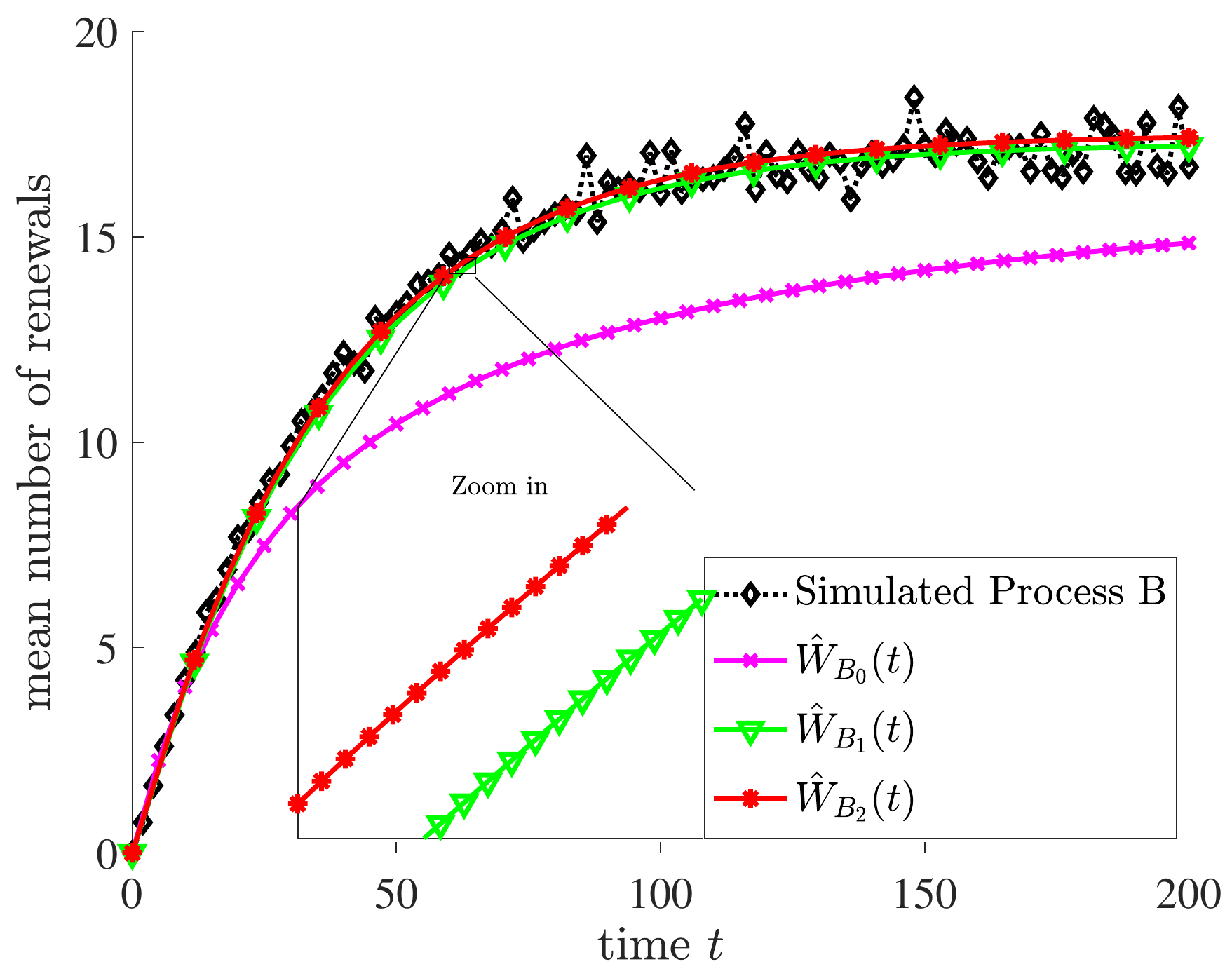}}
 \vspace{-3mm}
\caption{Example plots for Process B.}
 \vspace{-4mm}\label{Fig:ExampleB}
\end{figure}

%\begin{table}[thb]
%\centering
%\begin{tabular}{|c| c |c |c |c |} 
% \hline
%& \multicolumn{4}{c|}{MSE} \\
%\cline{3-5}
%Example number  & $W_B(t)$ & $W_{B_0}(t)$ & $\hat  W_{B_2}(t)$ & $\hat  W_{B_3}(t)$ \\  
% \hline
% 7 &  0.1479 & 2.2708 & 0.1539 & 0.1497 \\
% \hline
% 8 &  0.4228 & 31.3163 & 0.4375 & 0.4301 \\
% \hline
% 9 &  - & 2.485 & 0.068 & 0.0642 \\
% \hline
%\end{tabular}
%\caption{MSE of analytic function with simulated Process B.}
%\vspace{-5mm} 
%\label{table:1}
%\end{table}

\section{Bounded renewal process for  minimally repaired NC with preventively replaced CR \\ (Process B-MR).}\label{S:ModelBMR} 
\subsection {Model assumptions.}

(a) The NC and CR start working ``as good as new.''

(b) When NC fails, it is  minimally repaired. 

(c) When NC fails and is minimally repaired, CR is fully replaced in a preventive fashion.

(d) When CR fails, the whole system operation stops (its mission comes to an end).

(e) All replacements/repairs take negligible time.

%An example of this model would be a sensor transceiver that is accessible and  communicates with a microcontroller in the operation room. If the sensor transceiver stops workingg, the system mission is considered to have failed. Whenever a malfunction occurs in the microcontroller (which moreover is  assumed to be minimally fixed e.g., by simply resetting it),  the sensor can be at that time also replaced (preventively) so that the sensing mission can be prolonged.

The difference between Process B-MR and Process B is that when a failure  takes place at NC (shown by  an ``x" in Fig. \ref{fig:system_model2}), NC is minimally repaired  whereas it is  fully replaced/perfectly repaired in Process B. In both processes, the treatment of CR is the same (i.e., preventive full replacement whenever NC fails). 
\subsection {Derivation of renewal function for Process B-MR.}

\begin{theorem}\label{theo:BMR}
For Process B-MR, the mean number of renewals, $ W_{B_{MR}}(t)$, is:
\begin{dmath}\label{eqn:theorem2}
 W_{B_{MR}}(t)\!= \! F_{B_{MR}}(t|0)\!+\!\!\int_0^t\!\!  W_{B_{MR}}(y) f_{B_{MR}}(t\!-\!y|y)dy
.\end{dmath}
where $ F_{B_{MR}}(t|y)\!\!=\!\! \int_0^t\!\!\bar F_{ \tilde T}(x)f_T(x|y)dx$, 
and $ f_{B_{MR}}(t|y)\!\!=\!\!\bar F_{ \tilde T}(t) f_T(t|y)$.
\end{theorem}

%{\color{blue}{See Appendix \ref{app:Theorem6_proof} for the proof.}}

\subsection {Approximation for renewal function of Process B-MR.}

Because of the recursion and the kernel form $ f_{B_{MR}}(t|y)$, the exact $W_{B_{MR}}(t)$ (Eq. \ref{eqn:theorem2}) is in general hard to compute. Also, the existing approximation methods for the  g-renewal function that only involves NC do not apply to the bounded case.  In the following, we provide an approximation for Process B-MR.

We start by defining a probability, referred to as $ p_{\mathcal {B_{MR}}}(y)$, that the lifetime $\tilde T$ of CR is greater than the lifetime $T$ of NC, given that NC has been working for total time $y$:
\begin{eqnarray}\label{eqn:ProBMR}
p_{\mathcal {B_{MR}}}(y)\!\!\!\!\!\!\!\!&&
=P(    \tilde T\hiderel \geq T |y)
%&&=P(    \tilde T\hiderel \geq S_n-S_{n-1} |S_n>S_{n-1}=y)\nonumber\\
=  \int_0^\infty P( \tilde T \hiderel \geq  t) f_{  T}(t|y) dt\nonumber\\
&&=\int_0^\infty \bar F_{ \tilde T} (t)f_{  T}(t|y)dt
.\end{eqnarray}

We assume that in each of the renewals $i$, $ i\geq 1$, the probability of CR surviving NC given that NC has been working for time $y=S_{i-1}$ is approximately
equal to  Eq. \ref{eqn:ProBMR}.
The probability of CR surviving NC for $n$ consecutive renewals given that NC has been working for time $y=S_{i-1}$ is given as: \begin{eqnarray}\label{eqn:rho}
\rho_{S}(n)=\prod_{i=1}^{n} p_{ {B_{MR}}}(S_{i-1}),
\end{eqnarray}
with initial value $S_{0}=0$. 

The quantities $S_{i-1}$, $ i\geq 1$, are of course unknown, but 
since $S_{i-1}$ is the time that the first $i-1$ renewals take, and $M_g (t)$ gives the average number of renewals by time $t$,  $S_{i-1}$ can in principle be approximated by  the inverse function of the g-renewal function $M_g (t)=i-1$, i.e., $S_{i-1} \approx M^{-1}_g(i-1)$. In the case where  $M_g^{-1}(\cdot)$ can be computed,  we have that  $\rho_{S}(n)$ is approximated by 
$ \rho_M(n)=\prod_{i=1}^{n} p_{ {B_{MR}}}(M^{-1}_g(i-1))$.
In the case where $M_g^{-1}(\cdot)$ cannot be computed, we can approximate $\rho_{S}(n)$ as follows:  
 The approximate average time until the first renewal in Process B-MR is $\tau (1)=\int_0^\infty tf_T(t)dt$,  the approximate average time until the second renewal is $\tau (2)=\tau (1)+\int_0^\infty tf_T(t|\tau (1))dt$, and so on. Therefore, the approximate average time until the $i$th renewal of Process B-MR is given by
\begin{eqnarray}\tau (i)=\tau (i-1)+\int_0^\infty tf_T(t|\tau (i-1))dt,
\end{eqnarray}
where  $\tau(0)=0$.
We note that the $\tau(i)$ can be iteratively computed (using a simple loop rather than actual recursion) starting with 
$\tau (1)=\tau(0)+\int_0^\infty tf_T(t|\tau(0))dt = \int_0^\infty tf_T(t)dt$.
Therefore, the approximation for $\rho_S(n)$ becomes:
\begin{eqnarray}\label{eqn:rhotau}
\rho_S(n)\approx \rho_\tau(n)=\prod_{i=1}^{n} p_{ {B_{MR}}}(\tau(i-1)).
\end{eqnarray}

Now we are ready for the approximation of the renewal function of Process B-MR.

\begin{theorem}\label{theo:BMRapp1}
The renewal function $  W_{{B_{MR}}}(t)$ of Process B-MR  can be approximated by the function $ \hat W_{{B_{MR}}\_1}(t)$ given by:
\begin{eqnarray}\label{eqn:BMR}
\!\!\!\!\!\!\!\!\!\!\!\!\!&&\hat W_{B_{MR}\_1}(t)
=\sum_{k=1}^{\lceil M_g (t) \rceil} \rho_\tau (k).
\end{eqnarray}

%$\hat \rho (k)=\prod_{i=1}^{k} p_{\mathcal {B_{MR}}}(M^{-1}_g(i-1))$, $ M_g (t)$ is the g-renewal function, and $ M^{-1}_g(n)$ is its inverse function.
\end{theorem}

%{\color{blue}{See Appendix \ref{app:Theorem7_proof} for the proof.}}

 Because the approximation in Eq. \ref{eqn:BMR} results in a non-continuous (``zigzag'') curve, we propose an alternative approximation based on the following theorem:

\begin{theorem}\label{theo:BMRapp2}
The renewal function $  W_{{B_{MR}}}(t)$ of Process B-MR  can also be approximated by the function $ \hat W_{B_{MR}\_2}(t)$  given by:
\begin{subequations}
\begin{eqnarray}\label{eqn:BMR1}
\!\!\!\!\!\!\!\!\!\!\!\!\!\!\!\!\!\!\!\!&&\hat W_{B_{MR}\_2}(t)\nonumber\\
\!\!\!\!\!\!\!\!\!\!\!\!\!\!\!\!\!\!&&= M_g (t) \rho_{\tau} (\lceil M_g (t) \rceil)+\!\!\!\!\!\!\sum_{k=1}^{\lceil M_g (t) \rceil-1}\!\!\!\!\!\!\!\!k \Big(\rho_{\tau} (k)\!-\!\rho_{\tau} (k+1)\Big)\\
\!\!\!\!\!\!\!\!\!\!\!\!\!\!\!\!\!\!\!&&= M_g (t) \rho_{\tau} (\lceil M_g (t) \rceil)+\!\!\!\!\!\!\sum_{k=1}^{\lceil M_g (t) \rceil-1}\!\!\!\!\!\!\!\!k \Big(1-p_{\mathcal B_{MR}}(k)\Big) \rho_{\tau} (k).\label{eqn:BMR1_B}
\end{eqnarray}
\end{subequations}

%$\hat \rho (k)=\prod_{i=1}^{k} p_{\mathcal {B_{MR}}}(M^{-1}_g(i-1))$, $ M_g (t)$ is the g-renewal function, and $ M^{-1}_g(n)$ is its inverse function.
\end{theorem}

Eq. \ref{eqn:BMR1} or \ref{eqn:BMR1_B} produces a continuous smooth curve. We note that the particular form of  Eq. \ref{eqn:BMR1_B} reveals a rather intuitive meaning of the approximation for Process B-MR, namely that the average number of renewals $W_{{B_{MR}}}(t)$ in Process B-MR is approximated by the contribution of the average number $M_g(t)$ of the unbounded g-renewal process multiplied  by the probability $\rho_{\tau} (\lceil M_g (t) \rceil)$ that CR survives NC $\lceil M_g (t) \rceil$ times, plus the contribution of all integer values $k$ less than $\lceil M_g (t) \rceil$, each one multiplied by the probability that CR survives NC  exactly $k$ times (i.e, by the probability $\rho_\tau(k)\cdot(1-p_{\mathcal B_{MR}}(k))$.) 

We note that in Theorems \ref{theo:BMRapp1} and \ref{theo:BMRapp2}, if  the exact $M_g(t)$ cannot be computed (which is in general the case), the approximation function $\hat M_g(t)$  defined in Eq. \ref{eqn:approxg} is used.

We finally note that for the special  case where the lifetime of CR
follows the exponential distribution, the renewal functions of  Process B-MR and Process A-MR become identical, and also that that for the special case where the lifetime of NC
follows the exponential distribution, the renewal functions of  Process B-MR and Process B become identical.

\subsection{Numerical results.}
We demonstrate numerically how nicely the suggested approximation renewal function (Eq. \ref{eqn:BMR1} that uses $\rho_\tau (k)$) behaves in regards to the results of simulations in this subsection. We also provide  results  by using the $\rho_M (k)$ instead of $\rho_\tau (k)$ in Eq. \ref{eqn:BMR1} if $\rho_M (k)$ is available to compute (which is defined as $\hat W_{B_{MR}\_2'}$ in the plots). By comparing the $\hat W_{B_{MR}\_2}$ with $\hat W_{B_{MR}\_2'}$, we can see that $\tau(n)$  approximates well the $\hat M^{-1}_g(n)$. Also, a comparison between $\hat M^{-1}_g(n)$ and $\tau(n)$ is given in Table \ref{table:1}. 
Examples are shown in Figure \ref{Fig:ExampleBMR}.  The computation of an  exact renewal function is impossible in these cases. As can be seen, the proposed approximation formulas follow closely the simulation curves. We also observe that as the $\hat M^{-1}_g(n)$ is a little greater than $\tau(n)$,  the $\hat W_{B_{MR}\_2'}$ is a little higher than $\hat W_{B_{MR}\_2}$.

\begin{table}[tb]
\centering
\begin{tabular}{|l|ll|ll|}
\hline
\multirow{2}{*}{$n$} & \multicolumn{2}{l|}{Gamma($\theta\!=\!1,\kappa\!=\!2$)}& \multicolumn{2}{l|}{Rayleigh($\lambda=0.1$)} \\ \cline{2-5} 
& \multicolumn{1}{l|}{$M^{-1}_g(n)$} & $\tau(n)$ & \multicolumn{1}{l|}{$M^{-1}_g(n)$} & $\tau(n)$ \\ \hline
1& \multicolumn{1}{l|}{2.1462} & 2 & \multicolumn{1}{l|}{10}& 8.8623                     \\\hline
2& \multicolumn{1}{l|}{3.5052 }&3.3333 & \multicolumn{1}{l|}{14.1421}&  12.9459   \\ \hline
3& \multicolumn{1}{l|}{4.7490 }& 4.5641 & \multicolumn{1}{l|}{17.3205}&   16.1260  \\ \hline
10& \multicolumn{1}{l|}{12.6109 }& 12.3987& \multicolumn{1}{l|}{31.6228}&  30.5587   \\ \hline
50& \multicolumn{1}{l|}{54.0075 }& 53.7803& \multicolumn{1}{l|}{70.7107}&   69.9627 \\\hline
100& \multicolumn{1}{l|}{104.6602 }& 104.4308& \multicolumn{1}{l|}{100}&   99.3858  \\\hline
\end{tabular}\caption{Comparison between $\hat M^{-1}_g(n)$ and $\tau(n)$.}
\vspace{-0mm} 
\label{table:1}
\end{table}

\begin{figure}[t]
\centering  %图片全局居中
\vspace{-1mm} 
\begin{comment}
\subfigure[{$T \sim$ Gamma($\theta=10, \kappa=2$);\\ $\tilde T \sim$ Exp($\beta=0.5$).}]{
\label{Fig.sub8.1} 
\includegraphics[width=0.23\textwidth]{B_MR_GamExp_10_0.5.pdf}}
\subfigure[{$T \sim$ Gamma($\theta=1, \kappa=2$);\\ $\tilde T \sim$ Exp($\beta=0.1$).}]{
\label{Fig.sub8.2}
\includegraphics[width=0.23\textwidth]{B_MR_Gam_Exp_1_0.1.pdf}}
\end{comment}
\subfigure[$T \sim$ Gamma($ \kappa_1=2, \theta_1=10$); $\tilde T \sim$ Gamma($\kappa_2=2, \theta_2=0.5$).]{
\label{Fig.ExampleAMR.3MR} 
\includegraphics[width=0.23\textwidth]{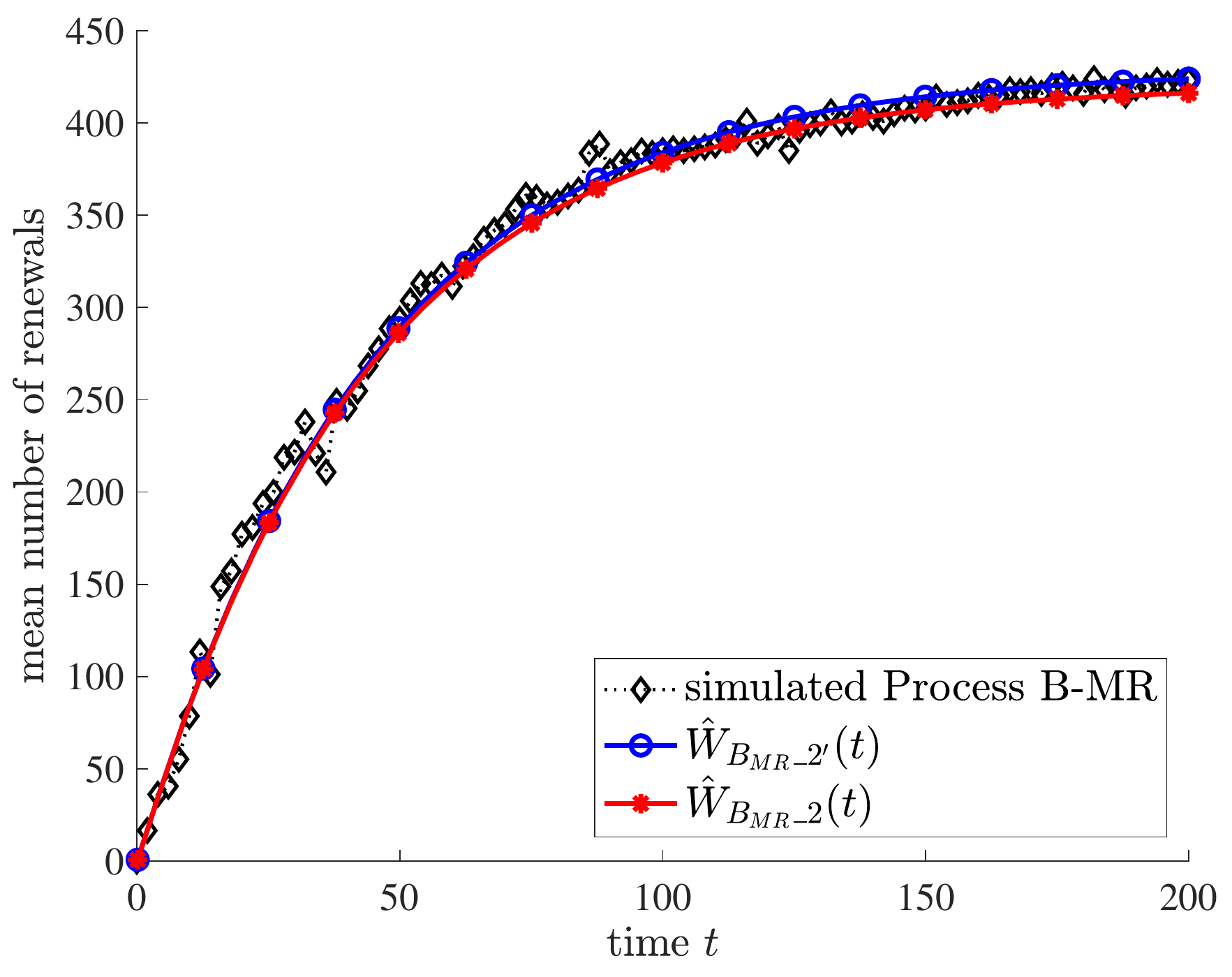}}
\subfigure[$T \sim$ Gamma($ \kappa_1=2, \theta_1=1$); $\tilde T \sim$ Gamma($\kappa_2=2, \theta_2=0.1$).]{
\label{Fig.ExampleAMR.4MR}
\includegraphics[width=0.23\textwidth]{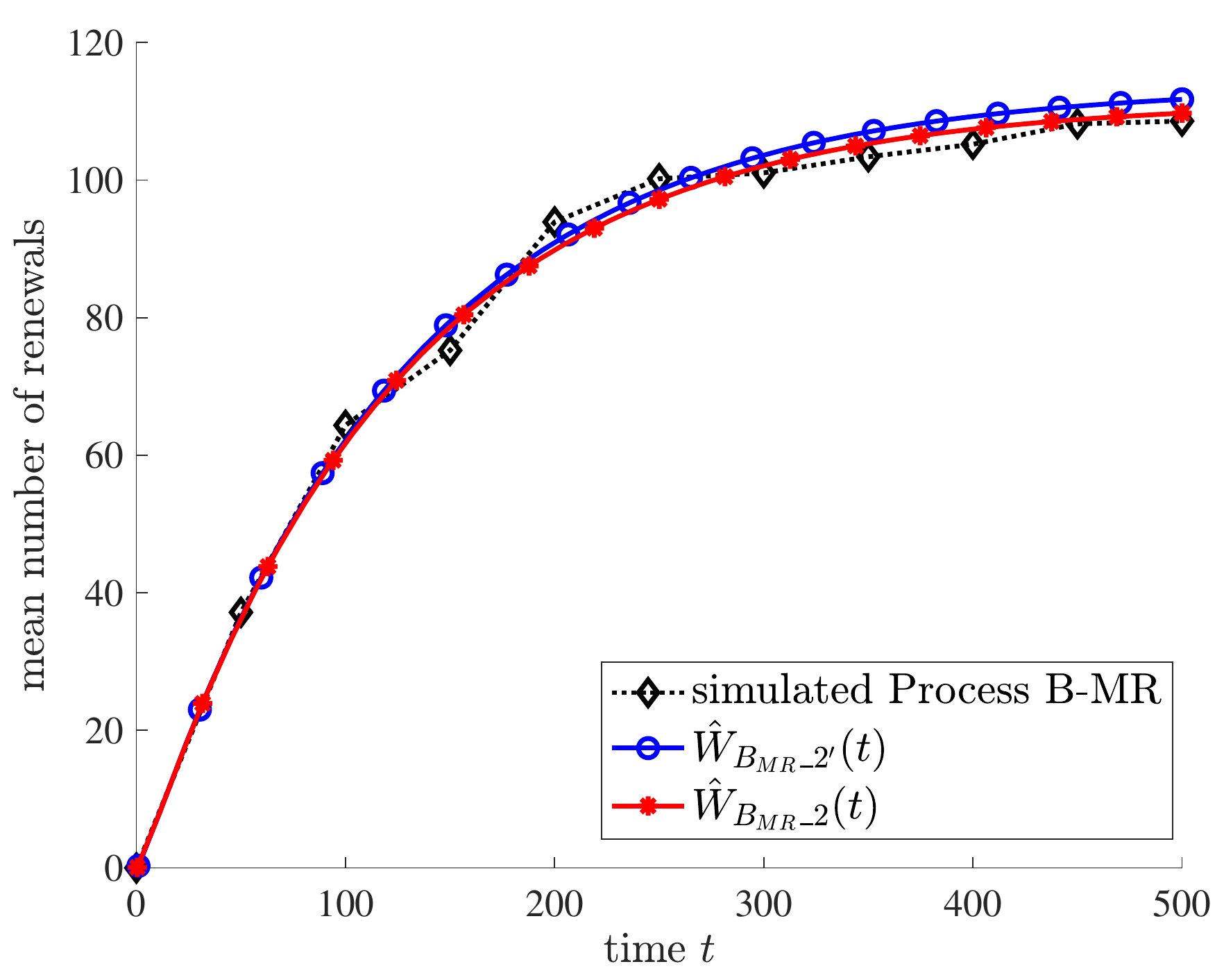}}
\subfigure[$T \sim$ Rayleigh($\lambda=0.1$);\qquad \qquad $\tilde T \sim$ Gamma($ \kappa=2, \theta=0.2$).]{
\label{Fig.ExampleAMR.5MR}
\includegraphics[width=0.23\textwidth]{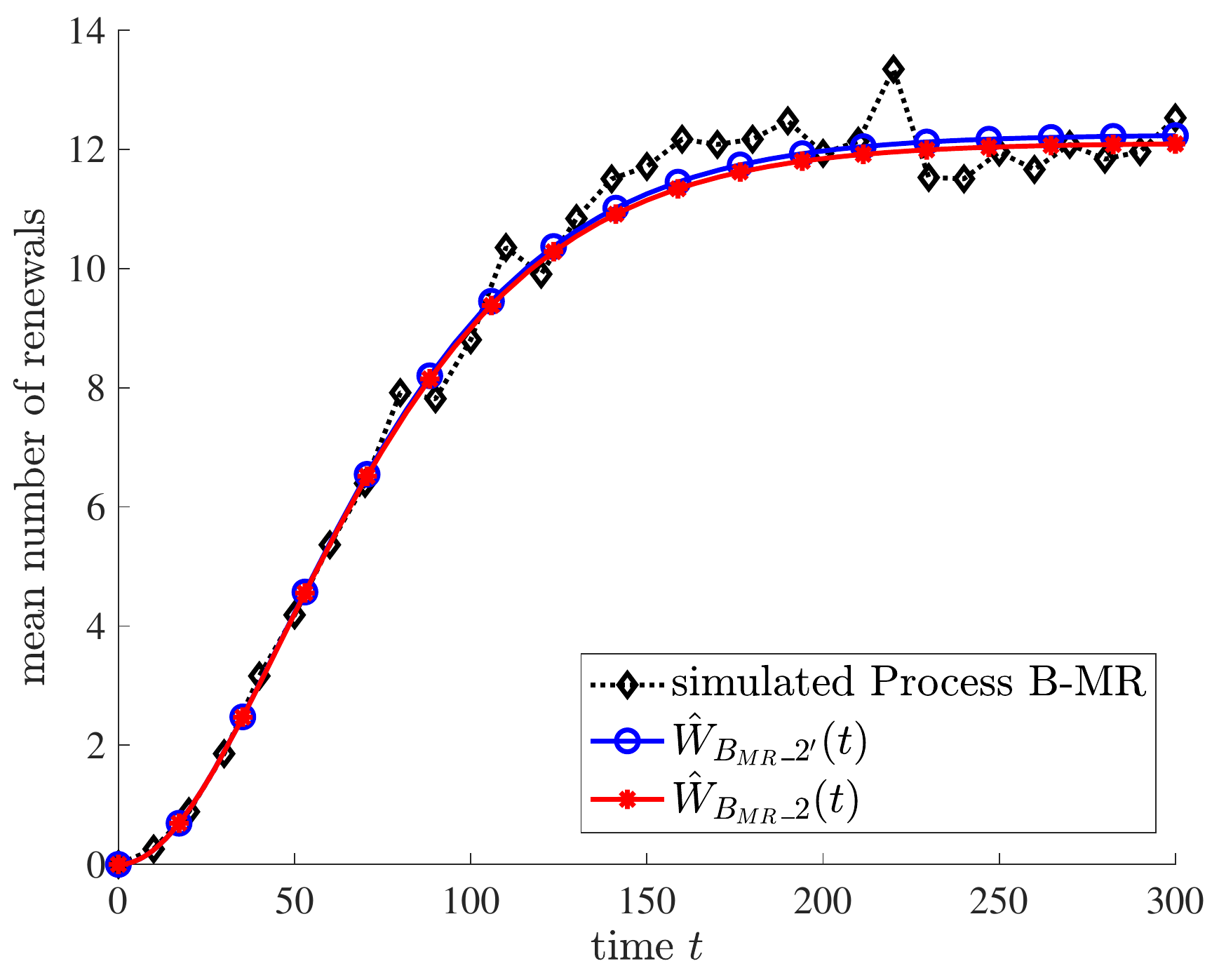}}
\subfigure[$T \sim$ Rayleigh($\lambda=0.1$);\qquad \qquad $\tilde T \sim$ Gamma($\kappa=2, \theta=0.3$).]{
\label{Fig.ExampleAMR.6MR}
\includegraphics[width=0.23\textwidth]{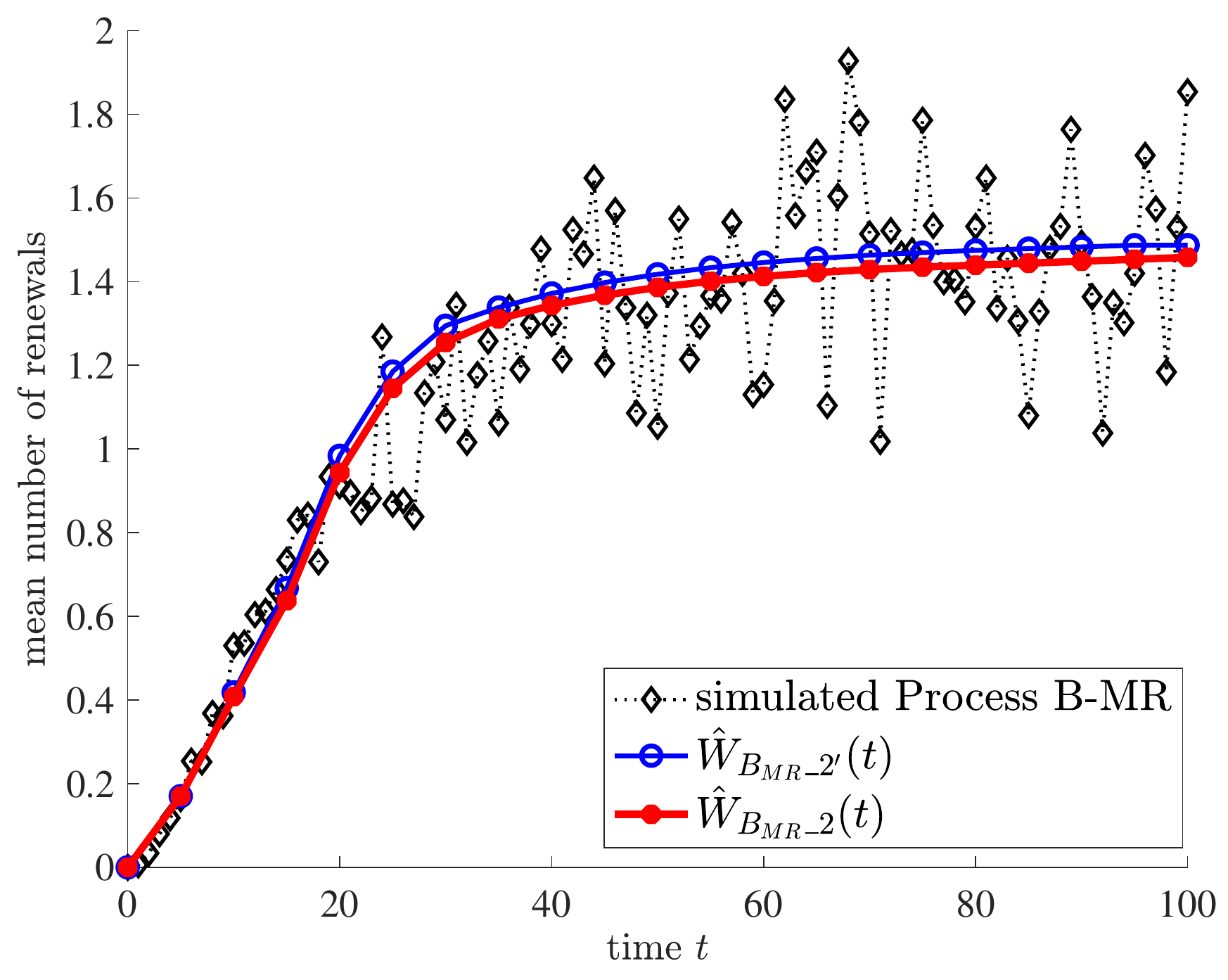}}
\caption{Example plots for Process B-MR.}
\vspace{-4mm} 
\label{Fig:ExampleBMR}
\end{figure}

\section{Illustration on wind turbine components.}\label{S:case} 
In this section, we illustrate the application of our formulas on actual lifetime distributions that have been obtained for wind turbine components (see, e.g., \cite{shafiee2015optimal}, \cite{WANG20191}).

\subsection{Blade and motor.}
\label{S:casestudy1}
We suppose the wind turbine blade as the NC component and the wind turbine motor as the CR component and are interested in the average number of  renewals of the blade until the motor stops working, without any maintenance action ever taken on the motor. Process A applies if the blade is replaced whenever it fails, while process A-MR applies if the blade whenever it malfunctions, is only minimally repaired. In this case, the failure of the blade is actually  taken to be the situation where its deterioration (due to cracks/delamination) exceeds a given level $D$ \cite{WANG20191}. 

The lifetime (deterioration process) of the blade has been modeled (see e.g., \cite{WANG20191}) by a homogeneous gamma process with scale parameter $\theta$ and shape parameter $\kappa$ and its pdf is given by
\begin{eqnarray}
\label{eqn:pdfblade}
f(t) = \frac{\theta^{\kappa t}}{\Gamma(\kappa t)} D^{\kappa t-1} e^{-\theta D}, t \geq 0; \theta, \kappa >0,
\end{eqnarray}
\noindent and the cdf is
\begin{eqnarray}
\label{eqn:cdfblade}
F(t) = \frac{\gamma(\kappa t, \theta D)}{\Gamma(\kappa t)}, t \geq 0; \theta, \kappa >0,
\end{eqnarray}
\noindent where $\gamma(. \, , .)$ and $\Gamma(.)$  represent the incomplete and complete gamma function respectively, i.e.
$\gamma(a,b) = \int_{b}^{\infty} z^{a-1} e^{-z} dz, a,b>0,$ and $\Gamma(a) = \int_{0}^{\infty} z^{a-1} e^{-z} dz, a>0.$

The specific values of the parameters in this modeling are $\kappa = 0.542$ and $\theta = 1.147$. The value of the  critical degradation level is $D = 20$cm for replacement (Process A), and  $D = 5$cm for minimal repair (Process A-MR). 

The time to failure of the motor has been modeled \cite{WANG20191} by a  Weibull distribution with scale parameter $\lambda$ and shape parameter $\nu$ and pdf  given by
\begin{equation}
 f(t;\lambda,\nu) = \begin{cases} \frac{\nu}{\lambda} \left( \frac{t}{\lambda} \right)^{\nu-1} e^{-( t/\lambda)^{\nu}} \text{  for $t\geq0$} \\0 \text {  otherwise} \end{cases}
\end{equation}

The specific values of the parameters are $\lambda= 3.5$ and  $\nu= 0.209$.
The plots of the simulation and the proposed approximation renewal function for Process A (Eq. \ref{eqn:approxA1}) are given in Fig. \ref{Fig:CaseStudy1}(a)  and the plots of the simulation and the proposed approximation renewal function for Process A-MR (Eq. \ref{eqn:approxAMR}) are given in Fig. \ref{Fig:CaseStudy1}(b). As can be seen the proposed approximations match the simulated results very well.
 In these plots, we provide also for reference/comparison the approximation curves of the standard renewal function of the blade in the case where the blade is considered as a stand-alone component without the limitation (``bound'')  imposed on it when working in conjunction with the CR component.
In Fig.  \ref{Fig:CaseStudy1}(a), the  standard renewal function  is  the classical renewal function whose approximation is given in Eq. \ref{eqn:appB}, and in Fig. \ref{Fig:CaseStudy1}(b) the  standard renewal function  is  the g-renewal function whose approximation is given in Eq. \ref{eqn:approxg}. For instance,
$t=100$, the average number of renewals of the blade would be 2.312 in the stand-alone case if the blade is fully replaced each time (Fig. \ref{Fig:CaseStudy1}(a)), whereas the actual number given by Process A is 0.3897. Also, by time $t=100$, the average number of renewals of the blade would be  519.9.  in the stand-alone case if the blade is minimally repaired each time (Fig. \ref{Fig:CaseStudy1}(b)), whereas the actual number given by Process A-MR is 82.02 (We note that the plot of $\hat W_{A_{MR}}(t)$ in Fig.\ref{Fig:CaseStudy1}(b) starts flattening after $2\times 10^5$ time units).

\begin{figure}[t]
\centering  %图片全局居中
\vspace{-1mm} 
\subfigure[Process A]{
\label{Fig.CaseStudy1.1} 
\includegraphics[width=0.23\textwidth]{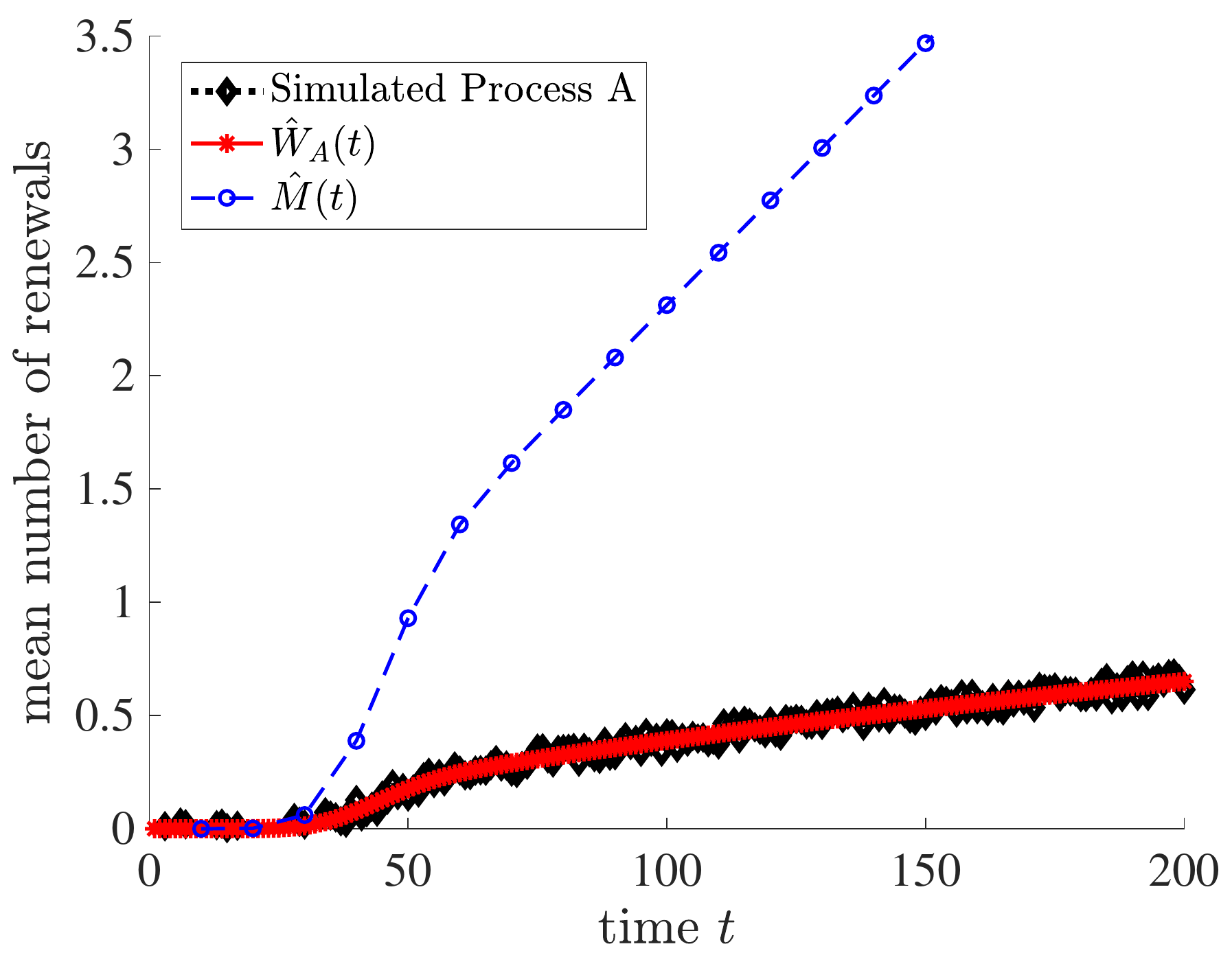}}
\subfigure[Process A-MR]{
\label{Fig.CaseStudy1.2} 
\includegraphics[width=0.23\textwidth]{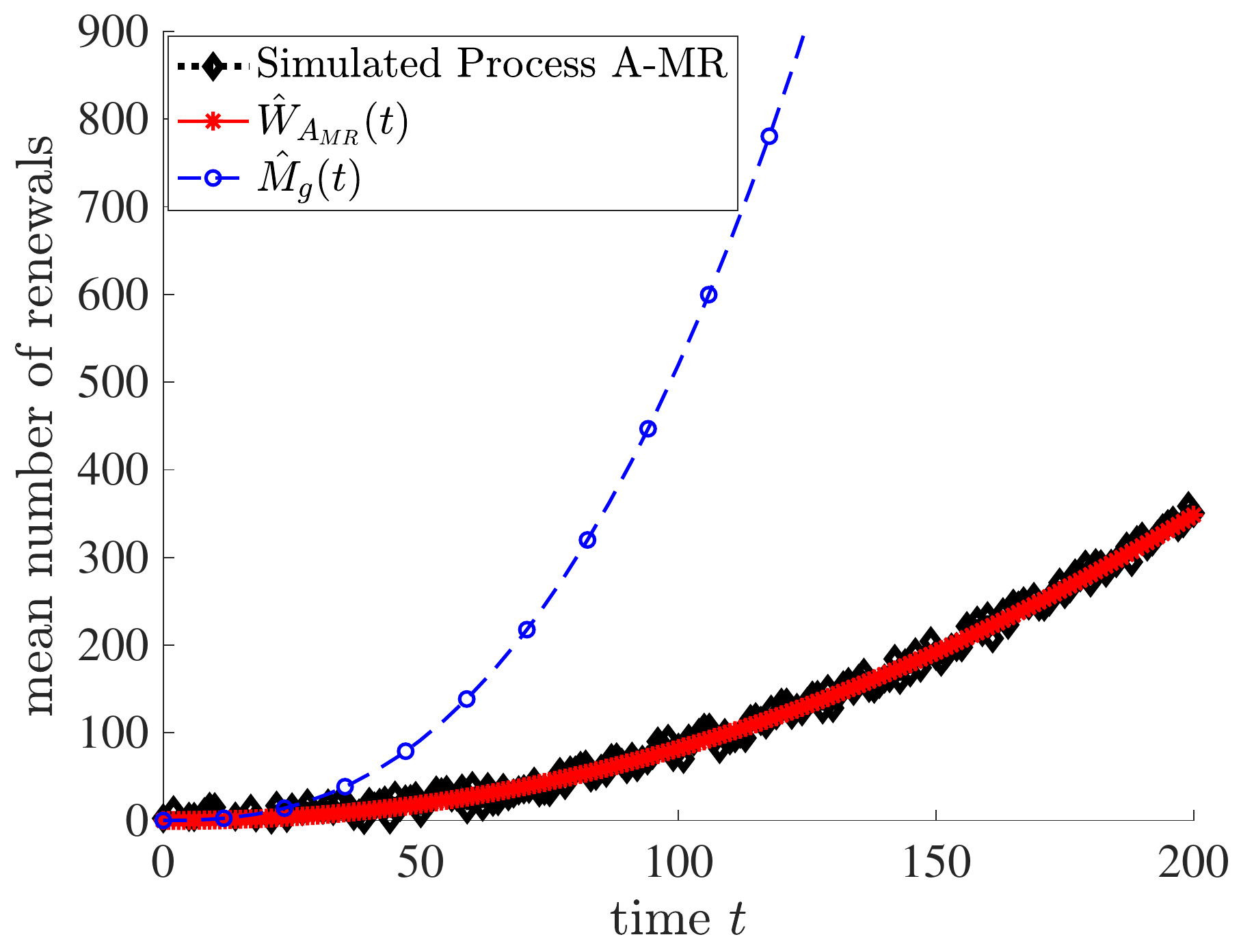}}
\caption{Process A and Process A-MR on wind turbine blade (as NC) and motor (as CR). 
%For reference, the standard renewal function for the blade as a stand-alone component is also plotted.
}
\vspace{-4mm} 
\label{Fig:CaseStudy1}
\end{figure}

\subsection{Blade and HSS bearing.}
\label{S:casestudy2}
 In the previous illustration, we considered the number of the renewals of the blade (where ``renewals'' are either replacements or minimal repairs) until the motor collapses. No preventive replacement was done on the motor every time the blade failed since  the cost of replacing the motor every  time is presumably forbidding. In the following illustration, we present a situation where process B and process B-MR apply, that is, a situation where it is more plausible to preventively replace the CR component  every time the NC component fails. For this purpose, we consider the blade as the NC component and  the ``gearbox high speed shaft (HSS)'' bearing \cite{shafiee2015optimal}  as the CR component.  The gearbox HSS bearing is considered to be one of the most critical components in wind turbines \cite{shafiee2015optimal}.  Similarly to \cite{shafiee2015optimal}, where when one of the bearings fails (more precisely, its  degradation level  reaches its
critical fracture size),  corrective replacement is done on it  while 
 a PM action is taken for the non-failed bearings (such as the pitch bearings, the main bearing, the generator bearing, or the gearbox intermediate speed shaft  bearing \cite{shafiee2015optimal}), we assume  that whenever the blade fails then the gearbox HSS bearing is preventively replaced at that time, and we are interested in the average number of  renewals of the blade until the gearbox HSS bearing fails. If the policy is to replace the  blade every time it fails,  then process B applies, whereas if the policy is to minimally repair the blade every time it fails, then Process B-MR applies. The lifetimes (deterioration process) of either the blade or the gearbox HSS bearing  have been modeled  (see, e.g., \cite{shafiee2015optimal}) by the same overall gamma process as in Eq. \ref{eqn:pdfblade} and \ref{eqn:cdfblade}. The specific parameters  for the lifetime of the blade are $\kappa= 0.542$, and $\theta = 1.147$, and  its critical degradation level is $D = 20cm$ for replacement and $D = 5cm$ for minimal repair.
The specific parameters  for the lifetime of the gearbox HSS bearing are $\kappa= 0.724$, and $\theta = 1.52$, and  its critical degradation level  is $D = 30cm$ \cite{shafiee2015optimal}. 
The plots of the simulation and the proposed approximation renewal function for Process B (Eq. \ref{eqn:TB3}) are given in Fig. \ref{Fig:CaseStudy2}(a)  and the plots of the simulation and the proposed approximation renewal function for Process B-MR (Eq. \ref{eqn:BMR1})  are given in Fig. \ref{Fig:CaseStudy2}(b). As can be seen the proposed approximations match the simulated results very well.
 Similarly to Fig. \ref{Fig:CaseStudy1}, the approximation curves of the standard renewal function of
the blade when considered as a stand-alone component are also plotted.   For instance, by time $t=200$, the average number of renewals of the blade would be  4.625 in the stand-alone case if the blade is fully replaced each time (Fig. \ref{Fig:CaseStudy2}(a)), whereas the actual number given by Process B is 3.415. Also, by time $t=60$, the average number of renewals of the blade would be 144.96 in the stand-alone case if the blade is minimally repaired each time (Fig. \ref{Fig:CaseStudy2}(b)), whereas the actual number given by Process B-MR is 75.94.

\begin{figure}[t]
\centering  %图片全局居中
\vspace{-1mm} 
\subfigure[Process B]{
\label{Fig.CaseStudy2.1} 
\includegraphics[width=0.23\textwidth]{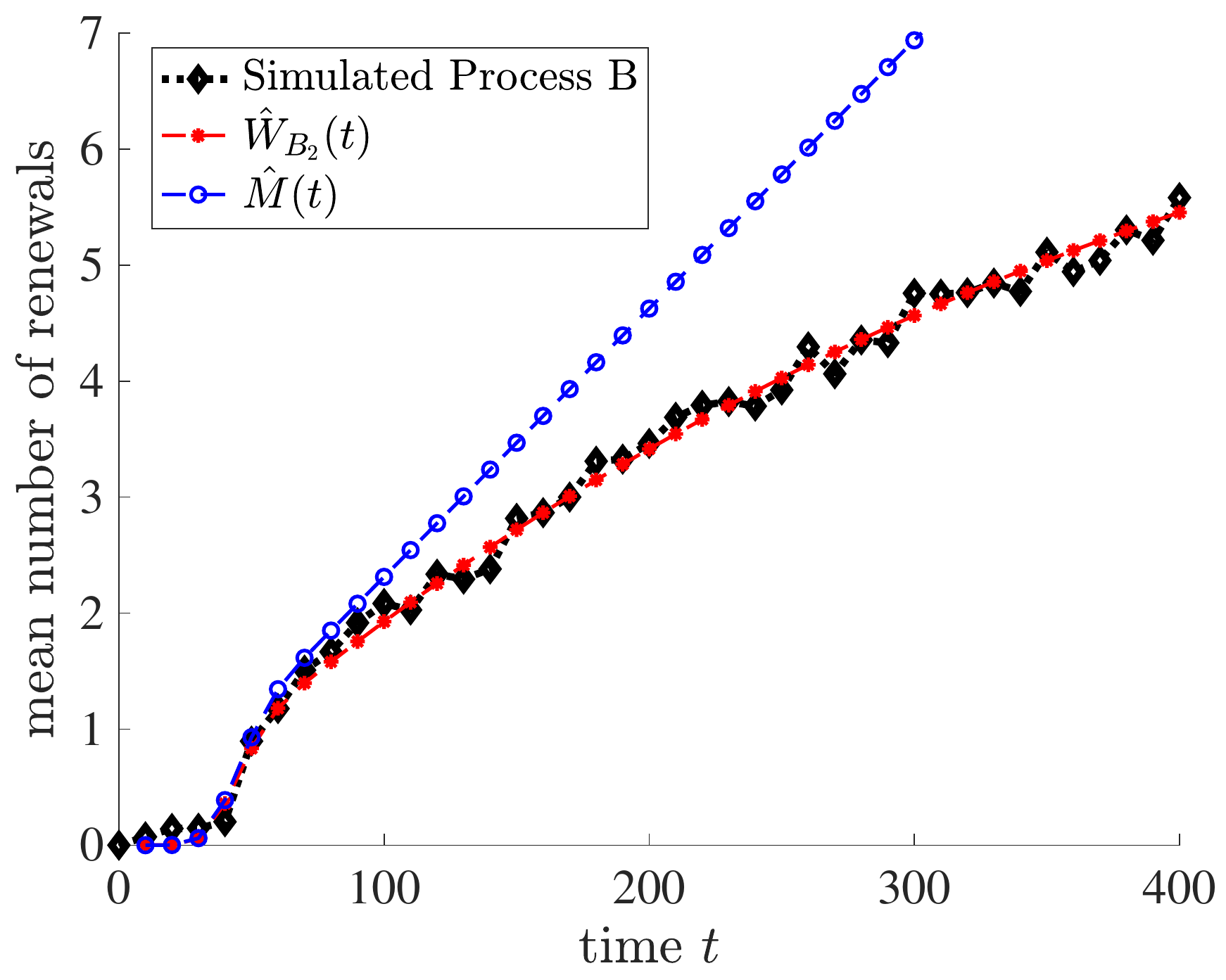}}
\subfigure[Process B-MR]{
\label{Fig.CaseStudy2.1} 
\includegraphics[width=0.23\textwidth]{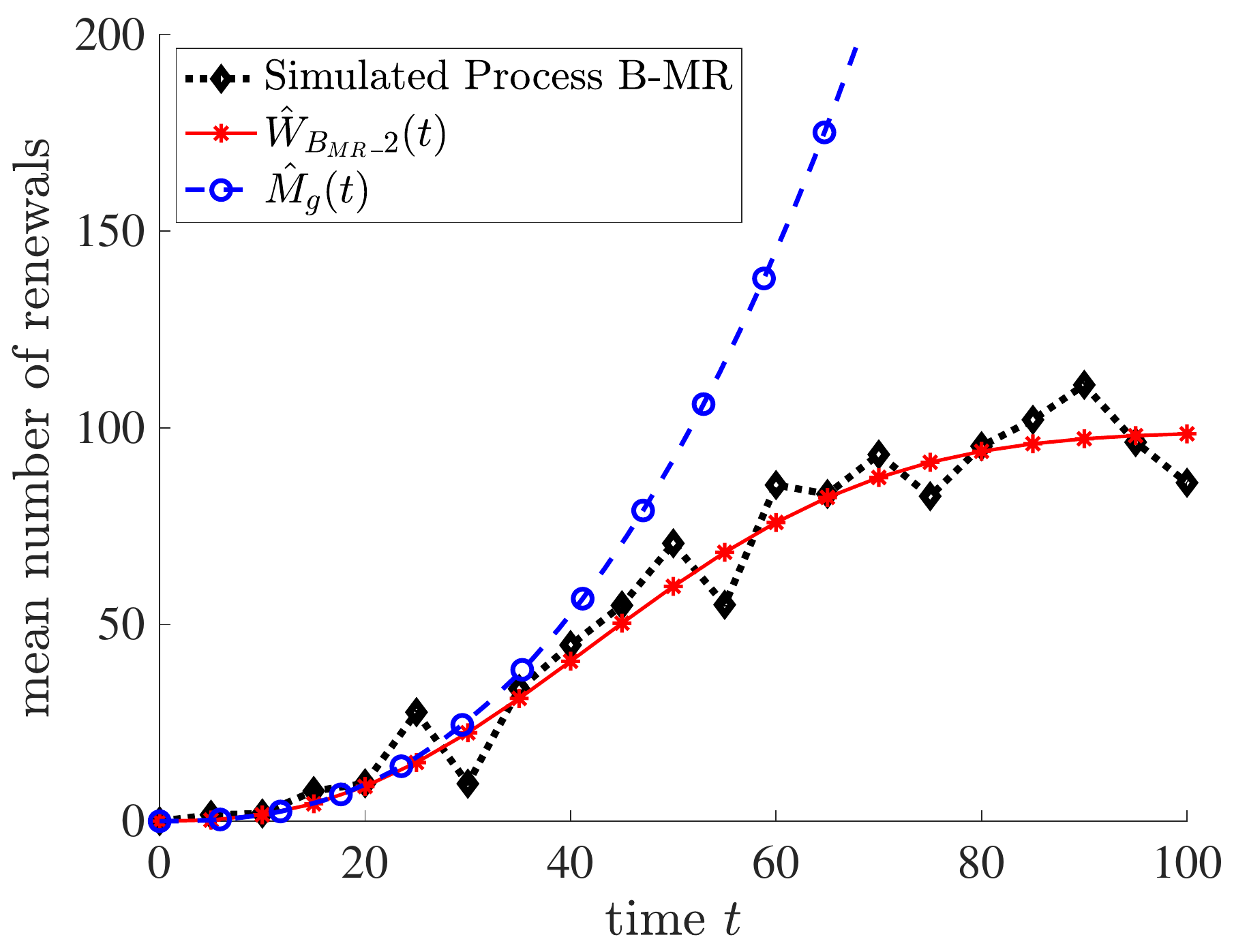}}
\caption{Process B and Process B-MR on wind turbine blade (as NC) and gearbox HSS bearing (as CR).}
\vspace{-4mm} 
\label{Fig:CaseStudy2}
\end{figure}

\section{conclusion.}\label{S:Conclusion} 

In this study, we developed new ``bounded" renewal theory functions for a  system consisting of a critical (CR) and non-critical (NC) component.
As long as the  CR component is still in operation, the NC component  can undergo corrective maintenance (replacement or minimal repair) whenever it fails. Whenever
the NC component fails, the CR component can optionally
be preventively replaced, but when the CR component fails, the system operation is irrevocably terminated.
  We formulated the new bounded renewal functions based on existing (classical and generalized) renewal theory functions and obtained  close approximations for them, validated by  simulations on various combinations of distributions, including actual distributions for wind turbine components.
In our upcoming work, we intend to explore standby arrangements as well as periodic/opportunistic maintenance policies.

\section{Appendix.}
\label{S:Appendix}

The proofs of the  theorems (and the lemmas they depend on) are given below.

\subsection{Proof of Lemma 1.}
\label{app:lemma1_proof}
\begin{proof}

Consider the counting procedure $N_{A}^{\rm nr}(t)$ , which counts the number of renewals in Process A without reevaluation of CR. 

The probability of $ N_{A}^{\rm nr}(t)\hiderel \geq n$ can be expressed as:
\begin{eqnarray}  \label{eqn:NASt}
P( N_{A}^{\rm nr}(t) \!\! \geq \!\!n)\!\!\!\!\!\!\!\!\!\!&&=P\Big((T_1 \!\leq \!t)
\!\cap\!(T_2 \!\leq \!t\!-\!T_1)\!\cap\!...\!\cap\!\nonumber \\
&&(T_n\! \leq\! t\!-\!S_{n-1}) \!\cap\!( \tilde T \!\geq\! T_n \!+\!S_{n-1})\Big).
\end{eqnarray}

Using the continuous lifetime distribution functions, this probability becomes:
\begin{eqnarray}\label{eqn:condiB}
 \!\!\!\!\!\!\!\!\!\!\!&&P( N_{A}^{\rm nr}(t)\!\!\hiderel \geq\!\! n)
 \!= \!\!\int_0^tf_T(x_1) \int_0^{t-x_1}\!\!\!\! f_T(x_2)...\nonumber \\
\!\!\!\!\!\!\!\!\!\!\!&&\int_0^{t-x_1-...-x_{n-1}}\!\!\!\!\!\!\!\!\!\!\!\!\!\!\!\!\!\!\!\!f_T(x_n)\bar F_{ \tilde T}(x_1\hiderel+...+x_n)dx_n... dx_2dx_1.
\end{eqnarray}

Now let $N_{A}^{\rm wr}(t)$ be the counting process that calculates the number of failures in Process A with a subsequent evaluation of CR after each occurrence of NC failure.

In this case, $P(N_{A}^{\rm wr}(t) \ge n)$ is given by:
\begin{eqnarray}
\!\!\!\!\!\!\!\!&&P(N_{A}^{\rm wr}(t)\hiderel \geq n)=P \Big((T_1\hiderel \leq t) \cap( \tilde T\hiderel \geq T_1 ) \cap\nonumber\\
\!\!\!\!\!\!\!\!&&(T_2\hiderel \leq t -T_1) \cap( \tilde T\hiderel \geq T_1+T_2|\tilde T\hiderel \geq T_1) \cap\nonumber\\
\!\!\!\!\!\!\!\!&&(T_3\!\!\hiderel \leq t \hiderel-T_1\hiderel-T_2) \hiderel\cap( \tilde T \hiderel \geq T_1\hiderel+T_2\hiderel+T_3|\tilde T \hiderel \geq T_1+T_2)\nonumber\\
\!\!\!\!\!\!\!&&\cap\,.\,.\,. \cap\nonumber\\
\!\!\!\!\!\!\!\!&&
( T_n\!\!\hiderel \leq t-S_{n-1})\cap( \tilde T \hiderel \geq T_n+S_{n-1}|\tilde T \hiderel \geq S_{n-1})\Big).
\end{eqnarray}
 
Using the continuous lifetime distribution functions, this probability is given by: 
\begin{eqnarray}\label{eqn:condiB}
 \!\!\!\!\!\!\!\!\!\!\!\!\!\!\!\!&&P( N_{A}^{\rm wr}(t)\!\!\hiderel \geq\!\! n) \!\!= \!\!\int_0^t\!\!\!f_T(x_1)\bar F_{ \tilde T}(x_1)\!\!\int_0^{t-x_1} \!\!\!\!\! \!\!f_T(x_2)\bar F_{ \tilde T}(x_2|x_1)... \nonumber\\
\!\!\!\!\!\!\!\!\!\!\!\!\!\!\!\!&&\int_0^{t-x_1-...-x_{n-1}}\!\!\!\!\!\!\!\!\!\!\!\!\!\!\!\!\!\!\!\!f_T(x_n)\bar F_{ \tilde T}(x_n|x_1+...+x_{n-1})dx_n... dx_2dx_1.
\end{eqnarray}

Given that $\bar F_{ \tilde T}(x|y)=\frac{\bar F_{ \tilde T}(x+y)}{\bar F_{ \tilde T}(y)}$, we have that
\begin{eqnarray}\label{eqn:condiB1}
 &&P( N_{A}^{\rm wr}(t)\!\!\hiderel \geq\!\! n) \nonumber\\
 &&= \!\!\int_0^tf_T(x_1)\bar F_{ \tilde T}(x_1)\int_0^{t-x_1} \!\! \!\!f_T(x_2)\frac{\bar F_{ \tilde T}(x_1+x_2)}{\bar F_{ \tilde T}(x_1)}... \nonumber\\
&&\int_0^{t-x_1-...-x_{n-1}}\!\!\!\!\!\!\!\!\!\!\!\!\!\!\!\!\!\!\!\!f_T(x_n)\frac{\bar F_{ \tilde T}(x_1+...+x_{n})}{\bar F_{ \tilde T}(x_1+...+x_{n-1})}dx_n... dx_2dx_1 \nonumber\\
&&= \!\!\int_0^tf_T(x_1) \int_0^{t-x_1}\!\!\!\! f_T(x_2)...\int_0^{t-x_1-...-x_{n-2}}\!\!\!\!\!\!\!\!\!\!\!\!\!\!\!\!\!\!\!\!f_T(x_{n-1}) \nonumber\\
&& \int_0^{t-x_1-...-x_{n-1}}\!\!\!\!\!\!\!\!\!\!\!\!\!\!\!\!\!\!\!\!\!\!\!\!f_T(x_n)\bar F_{ \tilde T}(x_1\!\hiderel+\!x_2\!\hiderel+\!...\!+\!x_n)dx_n... dx_2dx_1. \nonumber\\
&&=  P(N_{A}^{\rm nr}(t)\!\!\hiderel \geq\!\! n).
\end{eqnarray}
\end{proof}

\subsection{Proof of Theorem \ref{theo:A}.}
\label{app:theorem1_proof}

\begin{proof}
According to the definition of the mean number of renewals in any renewal process, the mean number of renewals in Process A is:
$$W_A(t) = \sum_{n=1}^\infty P(N_A(t)=n) \cdot n = \sum_{n=1}^\infty P(N_A(t)\ge n).$$

Due to Lemma \ref{lemma:A}, we can compute the probability $P(N_A(t)=n)$ without considering the reevaluation of CR every time NC fails. That is, $P(N_A(t)\ge n) = P(N_{A}^{\rm nr}(t) \ge n)$  from Eq. \ref{eqn:NASt}. Therefore,
\begin{dmath}  \label{eqn:PANA}
P(N_{A}(t)\ge n) = P(N_{A}^{\rm nr}(t)  \hiderel \geq n)\nonumber \\
=P\Big((T_1 \hiderel\leq t)
\cap(T_2 \hiderel\leq t-T_1)\cap...
\nonumber \\
\cap (T_n \hiderel\leq t-S_{n-1}) \cap( \tilde T \hiderel\geq T_n +S_{n-1})\Big)\nonumber \\
= P( T_1+T_2+... +T_n \hiderel\leq t\cap \tilde T \hiderel\geq  S_n) \nonumber \\
= P( S_{n}   \hiderel\leq  t \cap \tilde T \hiderel\geq   S_n)\nonumber \\
=P(S_n \hiderel\leq  \min(t,\tilde T)).
\end{dmath}

%The renewal counting process $ N_{A}(t)$ now can be denoted as  
%\begin{dmath} 
% N_{A}(t) = \max\{ n\hiderel: S_n \!\! \hiderel \leq \!\! \min(t,\tilde T) \}
%.\end{dmath}

We note that when the event $\{\tilde T \hiderel <t\}$ occurs, the random variable $\min(t,\tilde T)$ equals $\tilde T$ and, as a result, has the same distribution function as $\tilde T$ in the range from 0 to $t$. The random variable $\min(t,\tilde T)$ also equals $t$ when the event $\{\tilde T \hiderel \geq t\}$ occurs.
Therefore, Eq. \ref{eqn:PANA} can be rewritten as
\begin{eqnarray} \label{eqn:} 
P(  N_{A}(t)  \geq n) \!\!\!\!\!\!\!\!&&=
 P(\tilde T \geq  t) P(S_n \leq  t|\tilde T  \geq   t) \nonumber \\ 
 &&+ P(\tilde T \hiderel < t) P(S_n \leq  \tilde T |\tilde T < t).
%\!\!\!\!\!\!\!\!\!\!\!\!\!\!\!\!\!\!&&=\bar F_{ \tilde T}(t)P(S_n\!\!\hiderel \leq\!\!  t)\hiderel+\!\! \int_0^t \!\!P(S_n\hiderel \leq  x ) f_{\tilde T}(x)dx 
\end{eqnarray}
Let $  W_{A}(t)$ be the average number of renewals in $(0,t]$. The bounded renewal function can thus be expressed as follows:
\begin{eqnarray} \label{eqn:BCR1} 
W_{A}(t)\hiderel \!\!\!\!\!\!\!\!\!\!&&= \E { N_{A}(t)} \hiderel=\sum_{n=1}^{\infty}P(  N_{A}(t) \hiderel \geq n)\nonumber\\
&&= P(\tilde T\hiderel \geq   t)\sum_{n=1}^{\infty}P(S_n\hiderel \leq  t|\tilde T\hiderel \geq   t) \nonumber\\
&&+P(\tilde T \hiderel < t)\sum_{n=1}^{\infty}P(S_n\hiderel \leq  \tilde T |\tilde T \hiderel <\!\! t). 
\end{eqnarray}

Since $\tilde T $ and $S_n$ are independent, %we can get $P(S_n\hiderel \leq  t|\tilde T\!\!\hiderel \geq   \!\!t) =P(S_n\hiderel \leq  t)$.
$P(S_n\hiderel \leq  t|\tilde T\!\!\hiderel \geq   \!\!t) =P(S_n\hiderel \leq  t)$ holds.
In addition, $\sum_{n=1}^{\infty}P(S_n\hiderel \leq  t)$ is equivalent to the basic classical renewal equation $M(t)$ (see Eq. \ref{eqn:class}). 

Therefore, Eq. \ref{eqn:BCR1} equals
\begin{eqnarray}
 W_{A}(t)\!\!\!\!\!\!\!\!&&=M(t)  \bar F_{\tilde T}(t)\!\!+\!\!\sum_{n=1}^{\infty}P(S_n\!\!\hiderel \leq\!\!  \tilde T \cap \tilde T\!\! \hiderel <\!\! t) \nonumber\\
&&=M(t)  \bar F_{\tilde T}(t)\!\!+\!\! \int_0^t \!\!\sum_{n=1}^{\infty}P(S_n\hiderel \leq  x ) f_{\tilde T}(x)dx \nonumber\\
\iff  W_{A}(t)\!\!\!\!\!\!\!\!&&\hiderel=M(t)  \bar F_{\tilde T}(t)+\int_0^t  M(x)  f_{\tilde T}(x)dx.\nonumber
\end{eqnarray}
\end{proof}

\subsection{Proof of Theorem \ref{theo:AMR}.}
\label{app:theorem2_proof}

\begin{proof}
%Due to Lemma \ref{lem:AMR}, we are now using the case of reevaluating CR at every time that NC fails to derive the renewal function for Process A-MR.

Let  $ N_{A_{MR}}(t)$ denote the counting process that counts the number of Process A-MR renewals that occurred during the time interval $(0,t]$.
 
Let $T_i$ indicate the lifetime of NC between the ($i-1$)th and $i$th renewals (minimal repairs). We know that $P(T_i \le t) = F_T(t|y=S_{i-1})$, where now $S_i=\sum^i_{j=1}{T_j}$.

Due to Lemma \ref{lemma:AMR}, we can compute the probability $P(N_{A_{MR}}(t)=n)$ without considering the reevaluation of CR every time NC fails. 

The probability of $ N_{A_{MR}}(t)\geq n$ is written as:
\begin{eqnarray}
\!\!\!\!\!\!\!\!\!\!\!\!\!\!\!\!\!\!\!\!&&P( N_{A_{MR}}(t)\hiderel\! \geq\! n)\!=P\Big((T_1 \hiderel\leq t)
\cap(T_2 \hiderel\leq t-T_1)\cap...\cap
\nonumber \\
&& (T_n \hiderel\leq t-S_{n-1}) \cap( \tilde T \hiderel\geq S_{n})\Big)
.\end{eqnarray}
 
By using the corresponding distribution functions, the probability of $ N_{A_{MR}}(t) \geq n$ is 
\begin{eqnarray}\label{eqn:NAMRgpden}
\!\!\!\!\!\!\!\!\!\!\!\!\!\!\!\!\!\!\!&& P( N_{A_{MR}}(t)\!\!\hiderel \geq \!n) \!\!= \!\!\int_0^t\!\!f_T(x_1|0)\!\!\int_0^{t-x_1} \!\!\!\!\!\!\!\!\! \!\!f_T(x_2|x_1)... \nonumber\\
\!\!\!\!\!\!\!\!\!\!\!\!\!\!\!\!\!\!&&\int_0^{t-x_1-x_2-...-x_{n-1}}\!\!\!\!\!\!\!\!\!\!\!\!\!\!\!\!\!\!\!\!\!\!\!\!\!\!\!\!\!\!\!\!\!\!\!\!\!\!\!\!\!f_T(x_n|x_1\!+\!...\!+\!x_{n-1})\bar F_{ \tilde T}(x_1\!+\!...\!+\!x_n)dx_n... dx_{2}dx_1
.\end{eqnarray}
By a change of variables, Eq. \ref{eqn:NAMRgpden} can be rewritten as:
\begin{eqnarray}\label{eqn:NBn}
\!\!\!\!\!\!\!\!\!\! \!\!\!\!\!\!\!\!\!\!&&P( N_{A_{MR}}(t)\!\!\hiderel \geq\!\! n) \nonumber\\
\!\!\!\!\!\!\!\!\!\!\!\!\!\!\!\!\!\!&&=\underbrace{ \!\!\int_0^t\!\!\!\!\!f_T(y_1|0)\!\!\!\int_{y_1}^{t} \!\!\!\! \!\!f_T(y_2\!\!-\!\!y_1|y_{1})...\int_{y_{n-2}}^{t} \!\! \!\!\!\!\!\!\!\!\!\!f_T(y_{n-1}\!\!-\!\!y_{n-2}|y_{n-2})}_{A}\nonumber\\
\!\!\!\!\!\!\!\!\!\!\!\!\!\!\!\!\!\!&&\cdot \underbrace{\int_{y_{n-1}}^{t}\!\!\!f_T(y_n-y_{n-1}|y_{n-1})\bar F_{ \tilde T}(y_n)dy_n}_{B}\underbrace{dy_{n-1}...dy_{2}dy_{1}}_{A}
.\end{eqnarray}

Then, by doing the integration by parts for the last integral ($B$), we obtain:
\begin{eqnarray}\label{eqn:NBny}
B\!\!\!\!\!\!\!\!\!\!&&=\int_{y_{n-1}}^{t}\!\!\!f_T(y_n\hiderel -y_{n-1}|y_{n-1})\bar F_{ \tilde T}(y_n)dy_n \nonumber\\
&&=\underbrace{\bar F_{ \tilde T}(t)F_T(t \hiderel-y_{n-1}|y_{n-1})}_{{C}}\nonumber\\
&&+\underbrace{\int_{y_{n-1}}^{t}F_T(y_n\hiderel-y_{n-1}|y_{n-1})f_{ \tilde T}(y_n)dy_n}_{D}
.\end{eqnarray}

By using  $P_n(y,t)=\int_y^t P_{n-1}(y_1,t)f_T (y_1-y|y)dy_1$ (see Eq. \ref{eqn:pky}), we obtain
\begin{dmath}
A \cdot C
%=\!\!\int_0^tf_T(y_1|0)\int_{y_1}^{t} \!\! \!\!f_T(y_2\hiderel-y_1|y_{1})... \\
%\int_{y_{n-2}}^{t}\!\!\!f_T(y_{n-1}\hiderel-y_{n-2}|y_{n-2})F_T(t \!\!\hiderel-\!\!y_{n-1}|y_{n-1})\bar F_{ \tilde T}(y_n)dy_{n-1}... dy_{2}dy_1\\
=\bar F_{ \tilde T}(t)\int_0^t P_{n-1}(y_1,t)f_T (y_1|0)dy_1=\bar F_{ \tilde T}(t)P_n(0,t)
.\end{dmath}

Furthermore, by change the order of integration, we can also obtain the following equation:
\begin{eqnarray}
%\!\!\int_0^tf_T(y_1|0)\int_{y_1}^{t} \!\! \!\!f_T(y_2\hiderel-y_1|y_{1})...\\ 
%\int_{y_{n-2}}^{t}\!\!\!f_T(y_{n-1}\hiderel-y_{n-2}|y_{n-2})\bar F_{ \tilde T}(y_{n-1})\\\!\!\int_{y_{n-1}}^{t}\!\!\!F_T(y_n\!\!\hiderel-\!\!y_{n-1}|y_{n-1})f_{ \tilde T}(y_n)dy_ndy_{n-1}... dy_{2}dy_1\\
\!\!\!\!\!\!\!\!\!&&A\cdot D=\!\!\int_0^t f_{ \tilde T}(y_n)\!\!\int_0^{y_n}f_T(y_1|0)\int_{y_1}^{y_n} \!\! \!\!f_T(y_2\hiderel-y_1|y_{1})... \nonumber\\ 
\!\!\!\!\!\!&&\int_{y_{n-2}}^{y_n}\!\!\!\!\!\!\!\!f_T(y_{n\!-\!1}\!\!\hiderel-\!\!y_{n\!-\!2}|y_{n\!-\!2}) F_{  T}(y_{n}\!\!-\!\!y_{n\!-\!1}|y_{n\!-\!1})dy_{n\!-\!1}... dy_{2}dy_1dy_n\nonumber\\
\!\!\!\!\!\!\!\!&&=\!\!\int_0^t f_{ \tilde T}(y_n)\int_0^{y_n} P_{n-1}(y_1,y_n)f_T (y_1|0)dy_1dy_n\nonumber\\
\!\!\!\!\!\!\!\!\!\!\!\!&&=\!\!\int_0^t f_{ \tilde T}(y_n)P_{n}(0,y_n)dy_n
.\end{eqnarray}

Since $P_{n}(0,t)=P(N_g(t)\hiderel\geq n)$ (see Eq. \ref{eqn:pky1}), we have that 
\begin{eqnarray}
\!\!\!\!\!\!\!\!\!\!\!\!&&P( N_{A_{MR}}(t)\!\!\hiderel \geq\!\! n)
 \hiderel=A\cdot C+A\cdot D\nonumber\\
\!\!\!\! \!\!\!\!\!\!\!\!&&=\bar F_{ \tilde T}(t)P(N_g(t)\!\!\geq\!\! n)\!\!+\!\!\int_0^t f_{ \tilde T}(y_n)P(N_g(y_n)\!\!\geq \!\!n)dy_n
.\end{eqnarray}

Let $  W_{A_{MR}}(t)$  be the average number of renewals throughout the time period $(0,t]$. The renewal function can thus be expressed as follows:
\begin{eqnarray}\label{eqn:MAMR}
\!\!\!\!&& W_{A_{MR}}(t)=\sum_{n=1}^\infty  P( N_{A_{MR}}(t)\hiderel \geq n)\nonumber\\ 
\!\!\!\!&&=\!\!\sum_{n=1}^\infty \!\!\bar F_{ \tilde T}(t)P(N_g(t)\!\!\hiderel\geq\!\! n)\!\!+\!\!\sum_{n=1}^\infty\!\! \int_0^t \!\!f_{ \tilde T}(y_n)P(N_g(y_n)\!\!\hiderel\geq\!\! n)dy_n\nonumber\\ 
\!\!\!\!&&\iff  W_{A_{MR}}(t)\!\!\hiderel=\!\! \bar F_{ \tilde T}(t)M_g (t)\!\!\hiderel+\!\!\int_0^t\!\! f_{ \tilde T}(y)M_g(y)dy\nonumber
.\end{eqnarray}

\end{proof}

\subsection{Proof of Theorem \ref{theo:B}.}
\label{app:Theorem3_proof}

\begin{proof}

Let $N_B(t)$ be the counting process that counts the number of Process B renewals in  $(0,t]$. 
Based on the definition for any renewal process, Process B has a mean number of renewals given by:
\begin{dmath*}\label{eqn:WA1}
 W_{B}(t)\!\hiderel=\!\E{ N_{B}(t)}\!\hiderel=\!\!\!\sum_{n=1}^\infty\!\! nP( N_{B}(t)\hiderel \!=\!n)
%=\!\!\sum_{n=1}^\infty\!\! n[P( N_{B}(t)\!\!\hiderel \geq\!\!n)\!\!-\!\!P( N_{B}(t)\!\!\hiderel \geq  \!\!n\!\!+\!\!1)]\!\!
\!\hiderel=\!\!\sum_{n=1}^\infty \!\!P( N_{B}(t)\!\hiderel \geq  \!n)
.\end{dmath*}

The probability that CR survives NC in each renewal  is $P( \tilde T_i\hiderel \geq T_i)$, where $i \ge 1$. Consequently, the probability that $N_B(t)$ is larger than or equal to $n$ can be expressed as follows:
\begin{dmath} \label{eqn:prob}
P( N_{B}(t)\!\!\hiderel \geq\!\! n)\!\!\hiderel=\!\!P\Big((T_1\!\!\hiderel \leq\!\! t) \!\!\hiderel\cap\!\!( \tilde T_1\!\!\hiderel \geq \!\!T_1 )\!\!\hiderel\cap\!\!
 (T_2\!\!\hiderel \leq\!\! t\!\! \hiderel-\!\!T_1) \!\!\hiderel\cap\!\!( \tilde T_2\!\!\hiderel \geq \!\!T_2)\!\!\hiderel\cap \\
 (T_3\!\!\hiderel \leq \!\!t \!\!-\!\!T_1\!\!\hiderel-\!\!T_2) \!\! \hiderel\cap\!\! ( \tilde T_3\!\! \hiderel \geq \!\!T_3)\!\! \hiderel\cap...\cap\!\!
( T_n\!\!\hiderel \leq \!\!t \!\!\hiderel-\!\!S_{n-1})\!\! \hiderel \cap\!\! ( \tilde T_n \!\!\hiderel \geq\!\! T_n)\Big) 
.\end{dmath}

Re-expressing in terms of the relevant distribution functions, Eq. \ref{eqn:prob} becomes:
\begin{eqnarray}\label{eqn:NBn}
\!\!\!\!\!\!\!\!&& \!\!P(N_{B}(t)\!\!\hiderel \geq\!\! n) \!\!\nonumber\\ 
\!\!\!\!\!\!\!\!&&= \!\!\int_0^t\!\!\!\!f_T(x_1)\bar F_{ \tilde T}(x_1)\!\!
 \int_0^{t-x_1}\!\!\!\! \!\!\!\!\!\! \!\!f_T(x_2)\bar F_{ \tilde T}(x_2)... \!\!
\int_0^{t-x_1-...-x_{n-2}}\!\!\!\!\!\!\!\!\!\!\!\!\!\!\!\!\!\!\!\!\!\!\!\!\!\!\!\!\!\!\!\!\!\! f_T(x_{n-1})\bar F_{ \tilde T}(x_{n-1})\nonumber\\ 
 \!\!\!\!\!\!\!\!&&\cdot\int_0^{t-x_1-...-x_{n-1}}\!\!\!\!\!\!\!\!\!\!\!\!\!\!\!\!\!\!\!\!f_T(x_n)\bar F_{ \tilde T}(x_n)dx_ndx_{n-1}...dx_2dx_1
.\end{eqnarray}

Letting 
$
 f_{B}(x)=f_T(x)\bar F_{ \tilde T}(x),
$
and 
$ F_{B}(x)=\int_0^xf_T(u)\bar F_{ \tilde T}(u)du,$
Eq. \ref{eqn:NBn} can be rewritten as 
\begin{eqnarray}\label{eqn:NBn_1}
\!\!\!\!\!\!\!\!&& \!\!P( N_{B}(t)\!\!\hiderel \geq\!\! n) \!\!\hiderel= \!\!\int_0^t\!\! f_{B}(x_1)
\int_0^{t-x_1}\!\!\!\! \!\! \!\! f_B(x_2)... \nonumber\\ 
\!\!\!\!\!\!\!\!&&\!\!
\cdot\int_0^{t-x_1-...-x_{n-2}}\!\!\!\!\!\!\!\!\!\!\!\!\!\!\!\!\!\!\!\!\!\!\!\!\!\!\!\!\!\! f_B(x_{n-1})
 F_B((t\!\!-\!\!x_1\!\!-\!\!...\!\!-\!\!x_{n-2})\!\!-\!\!x_{n-1})dx_{n-1}...dx_2dx_1.\nonumber\\ 
 \!\!\!\!\!\!\!\!&&
\end{eqnarray}

Eq. \ref{eqn:NBn_1} can be seen as the convolution of $ F_{B}(x)$ with $ f_{B}(x)$ for $n-1$ times. 
The  Laplace transform $\mathcal L( )$ of  probability $P( N_{B}(t)\!\!\hiderel \geq\!\! n)$ is equal to  $\mathcal L( P( N_{B}(t)\!\!\hiderel \geq\!\! n))=  f_{B}^*(s)^{n-1}\frac{ f_{B}^*(s)}{s}=\frac{1}{s}(  f_{B}^*(s))^n$, and therefore,
the Laplace transform of $ W_{B}(t)$ is equal to 
\begin{dmath}\label{eqn:LapNAt}
 W_{B}^*(s)=\sum_{n=1}^\infty \frac{1}{s} ( f_{B}^*(s))^n\hiderel=\frac{ f_{B}^*(s)}{s(1- f_{B}^*(s))}
.\end{dmath}

Applying the inverse Laplace transform of $ W_{B}^*(s)=\frac{ f_{B}^*(s)}{s(1- f_{B}^*(s))}\Rightarrow W_{B}^*(s)=\frac{ f_{B}^*(s)}{s}+ W_{B}^*(s) f_{B}^*(s)$, we finally obtain the renewal function $ W_{B}(t)$ of Process B  as:
\begin{dmath}
 W_{B}(t)=  F_{B}(t)+\int_0^t  W_{B}(t-x) f_{B}(x)dx=  F_{B}(t)+\int_0^t  w_{B}(t-x) F_{B}(x)dx
.\end{dmath}
\end{proof}

\subsection{Proof of Theorem \ref{theo:Bapp}.}
\label{app:Theorem4_proof}

The proof of Theorem 4 depends on the following lemmas.

%{\color{blue}{See Appendix \ref{app:lemma3_proof} for the proof.}}

\begin{lemma} \label{lem:WB1}
The renewal function $W_B(t)$ for  Process B can be approximated by the function 
\begin{dmath}\label{eqn:TB2}
  \hat W_{B_1}(t) 
  \hiderel=\frac{p_B}{1-p_B}(1-\E{p_B^{N(t)}})
 %=\sum_{n=1}^\infty P(N(t)\hiderel \geq n)p_B^{ n}
 %=\sum_{n=1}^\infty \frac{p_B}{1-p_B}(1-p_B^n)P(N(t) \hiderel =n)
 ,\end{dmath} where $N(t)$ indicates the number of renewals in a classical renewal process involving only component NC.
\end{lemma}
\begin{proof}

%The probability of $X_t \hiderel =n$, where the first n trials result in ``$\mathcal B$'' and the trial $(n+1)$ results in ``$\mathcal B^*$,'' is given by
%\begin{dmath}\label{eqn:Aequaln}
%P(N_{B}(t)=n)= p_B^{ n}(1-p_B) \,\,\,\,\text{for} \,\,\,\,n\hiderel=1,2,...
%.\end{dmath}

The probability $P(N_{B}(t) = n)$ of having exactly $n$ renewals by time $t$ in Process B  can be related to the probability $P(N(t) = n)$ of having exactly $n$ renewals by time $t$ in a classical renewal process where only the NC component is present, as follows: Every renewal in the classical renewal process corresponds to a renewal in Process B as long as CR survives NC. That is, $P(N_{B}(t) = n)$ equals $P(N(t) = n)$ if  CR survives NC for $n$ consecutive renewals. 
Assuming that in each of these renewals the probability of CR surviving NC is approximately\footnote{We note that this is an approximation and not exact because $p_B$ (Eq. \ref{eqn:PB}) is computed for $t\rightarrow \infty$, i.e.,  $p_B=P( \tilde T > T)= \int_0^\infty  f_{ T}(x)\bar F_{\tilde T}(x)dx = \lim_{t\rightarrow \infty}\int_0^t  f_{ T}(x)\bar F_{\tilde T}(x)dx$, whereas each renewal has a finite time duration which is, moreover, different in each renewal.}   equal to  $p_B=\int_0^\infty  f_{ T}(x)\bar F_{\tilde T}(x)dx$ (Eq. \ref{eqn:PB}), we have that this event has probability $p_B^n$. On the other hand, if CR survives NC for $n$ consecutive renewals and dies before NC during the $(n+1)$th renewal, then the number of renewals in Process B is going to remain at $n$,  while the number of renewals in the classical renewal process keeps increasing for ever. 
Therefore, $P( N_{B}(t)=n)$ can be approximated by 
 $$P(N_{B}(t)=n)\approx P(N(t)=n)p_B^{n}+P(N(t) >n)p_B^{ n}(1-p_B).$$
Hence, the mean number of renewals in Process B $ W_{B}(t) =\sum_{n=0}^\infty n P( N_{B}(t)\hiderel =n)$ 
can be approximated by a function $\hat W_{B_1}(t)$ given as:
\begin{eqnarray}\label{eqn:AappW1}
\!\! \!\!\!\! \hat W_{B_1}(t) \!\! \!\!\!\!\!\!\!\!\!\!&&=\!\!\sum_{n=0}^\infty n[P(N(t)\hiderel =n)p_B^{ n}\!+\!P(N(t)\hiderel \!\!>\!\!n)p_B^{ n}(1\!-\!p_B)]\nonumber \\
\!\!\!\!\!\! &&=\sum_{n=0}^\infty n\{[P(N(t)\!\!\hiderel \geq \!\!n)\!\!-\!\!P(N(t)\!\!\hiderel \geq\!\! n+1)]p_B^{ n}\nonumber \\
 \!\!\!\!\!\!&&+P(N(t)\hiderel \geq n+1)p_B^{n}(1-p_B)\}\nonumber \\
\!\!\!\!\!\! &&=\sum_{n=0}^\infty n[P(N(t)\hiderel \geq n)p_B^{ n}-P(N(t)\hiderel \geq n\hiderel+1)p_B^{ n+1}]\nonumber \\
\!\!\!\!\!\!&&  = 1\,[P(N(t)\hiderel \geq 1)p_B-P(N(t)\hiderel \geq 2)p_B^{ 2}]\nonumber \\
\!\!\!\!\!\!&& \,+\, 2\,[P(N(t)\hiderel \geq 2)p_B^{ 2}-P(N(t)\hiderel \geq 3)p_B^{ 3}]\nonumber \\
\!\!\!\!\!\!&& \,+ \,3\,[P(N(t)\hiderel \geq 3)p_B^{ 3}-P(N(t)\hiderel \geq 4)p_B^{ 4}]\nonumber \\
\!\!\!\!\!\!&& \,+ \,.\,.\,.\nonumber \\
\!\!\!\!\!\!&&=\sum_{n=1}^\infty P(N(t)\hiderel \geq n)p_B^{ n}
.\end{eqnarray}
%By taking the inverse Laplace transform for both sides of equation (\ref{eqn:AappW1}), we obtain
%\begin{dmath}\label{eqn:AappLapA1}
%  \hat W_{B_1}^*(s) =\frac{1}{s}\sum_{n=1}^\infty (p_Bf_T(s))^n=\frac{p_Af_T(s)}{s(1-p_Bf_T(s))}
% .\end{dmath}
%By simple modified (\ref {eqn:AappLapA1}), we get
%\begin{dmath}\label{eqn:LapNAt1}
% \hat W_{B_1}^*(s)=\frac{p_Af_T(s)}{s}+ \hat W_{B_1}^*(s)p_Af_T(s)
%.\end{dmath}

%Thus, by taking the inverse Laplace transform, we finally obtain the mean number of renewals function $ W_B(t)$:
%\begin{dmath}\label{eqn:WA1t}
% \hat W_{B_1}(t)=  p_B F_T(t)+p_B\int_0^t  %\hat W_{B_1}(t-x) f_T(x)dx

%\subsubsection{Approximated renewal function 2.}

 %Since the unknown function $ \hat W_{B_1}(t-x)$ still exist as the one part of recursion, we improved the approximated renewal function further in the following theorem. 
Furthermore, Eq. \ref{eqn:AappW1} can be written as:
\begin{eqnarray}\label{eqn:AappW2}
 \!\! \!\!\!\!\!\!\!\!&& \hat W_{B_1}(t)  =\sum_{n=1}^\infty P(N(t)\hiderel \geq n)p_B^{ n}\nonumber\\
&&= p_B[P(N(t)\!\!\hiderel = \!\!1)+P(N(t)\!\!\hiderel = \!\!2)+...+P(N(t)\!\!\hiderel = \!\!n))]\nonumber\\
&& +~p_B^2[P(N(t)\!\!\hiderel = \!\!2)+P(N(t)\!\!\hiderel = \!\!3)+...+P(N(t)\!\!\hiderel = \!\!n))]\nonumber\\
&& +~p_B^3[P(N(t)\!\!\hiderel = \!\!3)+P(N(t)\!\!\hiderel = \!\!4)+...+P(N(t)\!\!\hiderel = \!\!n))]\nonumber\\
 &&+ \,.\,.\,.\nonumber\\
 &&=\sum_{n=1}^\infty  P(N(t) \hiderel =n)\sum_{k=1}^{n} p_B^{ k}
.\end{eqnarray}

Since $\sum_{k=1}^{n} p_B^{ k}=\frac{p_B}{1-p_B}(1-p_B^n)$, Eq. \ref{eqn:AappW2} can be further rewritten  as:
\begin{eqnarray}\label{eqn:BappWB1}
 \hat W_{B_1}(t) \!\!\!\!\!\!\!\! &&=\frac{p_B}{1-p_B}\sum_{n=1}^\infty (1-p_B^n)P(N(t) \hiderel =n)\nonumber\\
  &&=\frac{p_B}{1-p_B}\sum_{n=0}^\infty (1-p_B^n)P(N(t) \hiderel =n)\nonumber\\
  &&=\frac{p_B}{1-p_B}[1-\sum_{n=0}^\infty p_B^nP(N(t) \hiderel =n)]\nonumber\\
 && =\frac{p_B}{1-p_B}(1-\E{p_B^{N(t)}}).
 \end{eqnarray}
\end {proof}

\begin{lemma}\label{lem:normol}

Let $\mu_T$ and $\sigma_T$ represent the mean and standard deviation of the common underlying lifetime distribution in the classical renewal process. Then

\begin{dmath}\label{eqn:TB3_1}
 \frac{N(t)-M(t)}{\frac{\sigma_T}{\mu_T} \sqrt{M(t)}}\sim \mathcal N(0,1),
\end{dmath}
and
\begin{dmath}\label{eqn:TB3_2}
 P(\frac{N(t)-M(t)}{\frac{\sigma_T}{\mu_T} \sqrt{M(t)}}<y) \approx \Phi(y), \text{\,\,\,as $t\rightarrow \infty$,}
\end{dmath}
where $M(t)=\E{N(t)} $ is the  classical renewal function (Eq. \ref{eqn:class}), $\mathcal N(0,1)$ denotes the standard normal distribution, and $\Phi(\cdot)$ denotes the distribution function of the standard normal distribution $\mathcal N(0,1)$.
\end{lemma}

%{\color{blue}{See Appendix \ref{app:lemma4_proof} for the proof.}}

\begin{proof}
Following a similar proof as in \cite{ross1996stochastic}, let $y$ be any given number  and let $k=\lfloor M(t)+\frac{\sigma_T}{\mu_T}y \sqrt{M(t)}\rfloor$. Then, 
\begin{eqnarray}\label{eqn:N(t) nomal}
 P{( \frac{N(t)-M(t)}{\frac{\sigma_T}{\mu_T} \sqrt{M(t)}}<y)}\!\!\!\!\!\!\!\!\!\!&&=P\Big(N(t) \hiderel< M(t)+\frac{\sigma_T}{\mu_T}y\sqrt{M(t)}\Big)\nonumber\\
 &&= P\Big(N(t) \hiderel< \lfloor M(t)+\frac{\sigma_T}{\mu_T}y \sqrt{M(t)}\rfloor\Big)\nonumber\\
 &&= P{(N(t)<k)}\hiderel=P{(S_k>t)}\nonumber\\
 &&=P{( \frac{S_k-k \mu_T}{\sigma_T \sqrt{k}}>\frac{t-k \mu_T}{\sigma_T \sqrt{k}})}.
\end{eqnarray}

By using the fact from the   elementary renewal theorem (see, e..g, \cite{rausand2003system}) that $\lim_{t\rightarrow \infty}M(t)= \frac{t}{\mu_T}$, and substituting into the expression for $k$, we have that, as $t \rightarrow \infty$,
%$$\frac{t-\mu_TM(t)-y\sigma_T \sqrt{M(t)}}{\sigma_T \sqrt{\frac{\sigma_T}{\mu_T}y \sqrt{M(t)}+M(t)}}\rightarrow -y(1+\frac{y\sigma_T}{\sqrt{\mu_T t}})^{-\frac{1}{2}}\rightarrow -y.

$$\frac{t-k\mu_T}{\sigma_T \sqrt{k}}\rightarrow -y(1+\frac{y\sigma_T}{\sqrt{t\mu_T }})^{-\frac{1}{2}}\rightarrow -y.
$$

Now, by the central limit theorem, the distribution of $( \frac{S_k-k \mu_T}{\sigma_T \sqrt{k}})$ converges to the standard normal distribution as $t\rightarrow \infty$ (and thereby $k \rightarrow \infty$).
Therefore, based on Eq. \ref{eqn:N(t) nomal}, we have that for large $t$, the probability 
$P{( \frac{N(t)-M(t)}{\frac{\sigma_T}{\mu_T} \sqrt{M(t)}}<y)}$  is approximated by the probability $1 -\Phi(-y)$, where $\Phi()$ is the cdf of the standard normal distribution, and  since $1-\Phi(-y) = \Phi(y)$,
we have that 
$ \frac{N(t)-M(t)}{\frac{\sigma_T}{\mu_T} \sqrt{M(t)}}\sim \mathcal N(0,1).$
\end{proof}

\begin{lemma}\label{cor:greater}
The exact renewal function $W_B(t)$ of Process B is asymptotically equal to $\frac{p_B}{1-p_B}$ as $t\rightarrow \infty$, i.e., 
\begin{dmath}
\lim_{t\rightarrow \infty} W_{B}(t)=
\frac{ p_B}{1- p_B}
.\end{dmath}
\end{lemma}

\begin{proof}
From Eq. \ref{eqn:LapNAt}, we know that  the Laplace transform of $ W_{B}(t)$ is $W_{B}^*(s)=\frac{ f_{B}^*(s)}{s(1- f_{B}^*(s))}$. Also, using the derivative property, we have that $f_{B}^*(s)=sF_B^*(s)-F_B(0)=sF_B^*(s)-0=sF_B^*(s)$. Since
$\lim_{t\rightarrow \infty}F_B(t)=\lim_{t\rightarrow \infty}\int_0^tf_T(u)\bar F_{ \tilde T}(u)du\hiderel \\ =\int_0^\infty f_T(u)\bar F_{ \tilde T}(u)du\hiderel=p_B,
$
we also have, based on the final value theorem, that
$
\lim_{s\rightarrow 0}sF_B^*(s)\hiderel=\lim_{t\rightarrow \infty}F_B(t)=p_B.
$

By using the final value theorem again, we have that
\begin{eqnarray}
\lim_{t\rightarrow \infty} W_{B}(t)\!\!\!\!\!\!\!\!\!\!&&=\lim_{s\rightarrow 0} s W_{B}^*(s)=\lim_{s\rightarrow 0}\frac{ f_{B}^*(s)}{1- f_{B}^*(s)} \nonumber \\
&& =\lim_{s\rightarrow 0}\frac{ sF_B^*(s)}{1- sF_B^*(s)}=\frac{ p_B}{1- p_B}.
\end{eqnarray}
 
\end{proof}

The proof of Theorem 4 based on Lemmas  \ref{lem:WB1},   \ref{lem:normol}, and  \ref{cor:greater} is now  given.

\begin{proof}
From Lemma \ref{lem:WB1}, we know that the renewal function of Process B can be approximated by Eq. \ref{eqn:TB2}, i.e., 
\begin{eqnarray}\label{eqn:BappWB2}
 \hat W_{B_1}(t)  =\frac{p_B}{1-p_B}\cdot\Big\{1-\E{p_B^{N(t)}}\Big\}
.\end{eqnarray} 

From Lemma \ref{lem:normol}, we know that $N(t)$ is asymptotically normal with mean $\E{N(t)}=M(t)$ and standard deviation  $\sigma_{N(t)}=\frac{\sigma_T}{\mu_T} \sqrt{M(t)}$, where $\mu_T$ and $\sigma_T$ denote the mean and standard deviation of the lifetime of component NC.
%, where $\mu_T$ and $\sigma_T$ is the mean and standard deviation of random variable $T_i$ respectively (i.e., lifetime of NC) [see textbook \cite{rausand2003system}],  
Therefore, the random variable $p_B^{N(t)}=\exp{N(t)\ln{p_B}}$ is asymptotically log-normally distributed and its mean is given by 
\begin{eqnarray}
\E{p_B^{N(t)}}\!\!\!\!\!\!\!\!\!
&&=\exp{M(t)\ln{(p_B)}+[\frac{\sigma_T}{\mu_T} \sqrt{M(t)}\ln{(p_B)}]^2/2}\nonumber\\
&&=\exp{M(t)\ln{(p_B)}(1+\frac{\sigma_T^2}{2\mu_T^2}\ln{(p_B)})}
=(p_B)^{M(t)\cdot \delta},~~
\end{eqnarray}
where $\delta = (1+\frac{\sigma_T^2}{2\mu_T^2}\ln{(p_B)})$.
%and variance $\text{Var}[p_B^{N(t)}]=\exp{2N(t)\ln{p_B}+[\sigma_{N(t)}\ln{(p_B)}]^2}(\exp{[\sigma_{N(t)}\ln{(p_B)}]^2}-1)$. 
Therefore, Eq. \ref{eqn:BappWB2} is approximately equal to:

\begin{dmath}\label{eqn:W2=W3}
  \hat W_{B_1}(t) 
\approx \frac{p_B}{1-p_B}\Big(1-(p_B)^{\delta M(t)}\Big).
\end{dmath}

 %Now we are giving the relationship with $  W_{B}(t)$ and $  \hat W_{B_2}(t)$. 

%\begin{theorem}
%The approximated renewal function for Process B for $t>0$ can also be written as:
%\begin{dmath}\label{eqn:TB4}
% \hat W_{B_3}(t)=\frac{p_B}{1-p_B}(1-p_B^{M(t)})\approx M(t)p_B^{\lceil M(t) \rceil}+\sum_{k=0}^{\lceil M(t) \rceil-1}kp_B^k(1-p_B) ,
%.\end{dmath}
%\end{theorem}
%\noindent while noting that $ \hat W_{B_3}(0)=0$, and $\lceil \cdot \rceil$ is the ceiling function.

 From Lemma \ref{cor:greater}, we know that 
$W_B(t)$ is asymptotically equal to $\frac{p_B}{1-p_B}$ as  $t \rightarrow \infty$.
Now, since
\begin{eqnarray}\label{eqn:small}
&&\frac{p_B}{1-p_B}-\frac{p_B}{1-p_B}(1-(p_B)^{\delta M(t)})\nonumber\\
&&\geq 
\frac{p_B}{1-p_B}-\frac{p_B}{1-p_B}(1-(p_B)^{M(t)})
\end{eqnarray}
(due to the fact that $0\leq p_B\leq 1$, $\ln{(p_B)}\leq 0$, and $\delta \le 1$), we have that the expression $\frac{p_B}{1-p_B}(1-(p_B)^{M(t)})$ is asymptotically closer to $\lim_{t\rightarrow \infty} W_{B}(t)=
\frac{ p_B}{1- p_B}$ than $\frac{p_B}{1-p_B}(1-(p_B)^{\delta M(t)})$. Therefore, we use the former expression (denoted by $\hat{W}_{B_2}(t)$) as our approximation for $W_B(t)$, i.e.,
\begin{equation}
    W_B(t) \approx \hat W_{B_2}(t) =\frac{p_B}{1-p_B}\Big(1-(p_B)^{M(t)}\Big).  
\end{equation}

\end{proof}

\subsection{Proof of Theorem \ref{theo:BMR}.}
\label{app:Theorem6_proof}

\begin{proof}

Let $ N_{B_{MR}}(t)$  be the counting process that counts the number of renewals for Process B-MR in the time interval$(0,t]$.
 Let $S_n$ indicate the time that NC has been in operation just before its $n$th minimal (and instantaneous) repair. 

Since the NC is minimally repaired to as-bad-as-old state,  the $i$th lifetime of  NC $T_i$ (for $i>0$), which indicates the lifetime of NC between the ($i-1$)th and $i$th renewals (minimal repairs), follows  $P(T_i \le t) = F_T(t|y=S_{i-1})=\frac{F_{ T}(x+y)-F_{ T}(y)}{1- F_{ T}(y)}$, where now $S_i=\sum^i_{j=1}{T_j}$.
%$F_{ T}(x|y)=P( T_n\hiderel \leq x | S_{n-1}= y)=P( S_n\hiderel \leq x+y | S_n>S_{n-1}= y)=\frac{F_{ T}(x+y)-F_{ T}(y)}{1- F_{ T}(y)}$, which is assumed to depend on only the total working time $y$. 
Likewise, 
the conditional survival function is: 
$
\bar F_{ T}(x|y)=\frac{\bar F_{ T}(x+y)}{\bar F_{ T}(y)},
$
and the conditional pdf is:
$
f_{ T}(x|y)=\frac{f_{ T}(x+y)}{1-F_{ T}(y)},
$
where $F_{ T}(x)$ (pdf $f_{ T}(x)$) is the initial lifetime (real age $S_0=0$) distribution of the NC. The definitions of CDF, survival function and pdf of the lifetime of CR is same as that in Process B.

The probability of $ N_{B_{MR}}(t)\geq n$ is written as:
\begin{eqnarray}
\!\!\!\!\!\!\!\!&&P( N_{B_{MR}}(t)\hiderel \geq n)=P\Big((T_1\hiderel \leq t) \cap( \tilde T_1\hiderel \geq T_1 )\cap \nonumber\\ 
\!\!\!\!\!\!\!\!&&(T_2\!\!\hiderel \leq t - T_1) \cap( \tilde T_2\hiderel \geq T_2)\cap (T_3\!\!\hiderel \leq t -T_2 -T_1) 
\nonumber\\
\!\!\!\!\!\!\!\!&&\cap( \tilde T_3 \hiderel \geq T_3)\cap 
...\cap( T_n\!\!\hiderel \leq\! t-S_{n-1})\! \cap\!( \tilde T_n \hiderel \!\geq\! T_n)\Big)
.\end{eqnarray}
 
By using the corresponding distribution functions, $\bar F_{ \tilde T}(x)$ and $f_{ T}(x|y)$, we obtain the probability for $ N_{B_{MR}}(t) \geq n$ as 
\begin{eqnarray}\label{eqn:NCgpden}
\!\!\!\!\!\!\!\!\!\!\!\!&& P( N_{B_{MR}}(t)\!\!\hiderel \geq \!n) \!\!= \!\!\int_0^t\!\!f_T(x_1|0)\bar F_{ \tilde T}(x_1)\!\!\int_0^{t-x_1} \!\!\!\!\!\!\!\!\! \!\!f_T(x_2|x_1)\bar F_{ \tilde T}(x_2)... \nonumber\\
\!\!\!\!\!\!\!\!\!\!\!&&\int_0^{t-x_1-x_2-...-x_{n-1}}\!\!\!\!\!\!\!\!\!\!\!\!\!\!\!\!\!\!\!\!\!\!\!\!\!\!\!\!\!\!\!\!\!\!\!\!\!f_T(x_n|x_1+x_2+...+x_{n-1})\bar F_{ \tilde T}(x_n)dx_n... dx_{2}dx_1
.\end{eqnarray}
By a change of variables, Eq. \ref{eqn:NCgpden} can also be written as:
\begin{eqnarray}\label{eqn:NBMR}
\!\!\!\!\!\!\!\!\!\!\!\!\!\!\!\!\!\!&& P( N_{B_{MR}}(t)\!\!\hiderel \geq\!\! n)\nonumber\\
 \!\!\!\!\!\!\!\!\!\!\!\!\!\!\!\!\!\!&&= \!\!\int_0^tf_T(y_1|0)\bar F_{ \tilde T}(y_1)\int_{y_1}^{t} \!\! f_T(y_2-y_1|y_1)\bar F_{ \tilde T}(y_2-y_1)... \nonumber\\
\!\!\!\!\!\!\!\!\!\!\!\!\!\!\!\!\!\!&&\int_{y_{n-1}}^{t}\!\!\!\!\!\!f_T(y_n-y_{n-1}|y_{n-1})\bar F_{ \tilde T}(y_n-y_{n-1})dy_n... dy_{2}dy_1
.\end{eqnarray}

Define $B_n(y,t)$ as the probability that the counting process $ N_{B_{MR}}(t)$ increases by $n$ or more in the interval $(y,t]$ given that the system has been working for time $y$:
\begin{eqnarray}
B_n(y,t)\!\!\!\!\!\!\!\!\!\!
&&=P{( N_{B_{MR}}(t)- N_{B_{MR}}(y) \geq n| y)}\nonumber\\
&&=\!\!\int_y^tB_{n-1}(z,t) f_T(z-y|y)\bar F_{ \tilde T}(z-y)dz
.\end{eqnarray}

Then Eq. \ref{eqn:NBMR} can be expressed further as:
\begin{eqnarray}
 P( N_{B_{MR}}(t)\!\hiderel \geq\! n) \!\!\!\!\!\!\!\!\!\!\!\!&&=B_n(0,t)\nonumber\\
 &&= \!\!\int_0^t\!\!\!\!B_{n-1}(y_1,t) f_T(y_1|0)\bar F_{ \tilde T}(y_1)dy_1
.\end{eqnarray}

Therefore, $W_{B_{MR}}(t)$ can be expressed as
\begin{eqnarray} 
&&W_{B_{MR}}(t)=\E{ N_{B_{MR}}(t)}\nonumber\\
&&=\!\sum_{n=1}^\infty P( N_{B_{MR}}(t)\geq n)\!=\!\sum_{n=1}^\infty B_n(0,t)\nonumber\\
&&=\!\!P{( N_{B_{MR}}(t)\!\!\geq\!\! 1)}\!\!\hiderel+\!\!\!\sum_{n=2}^\infty  \!\!\int_0^t\!\! B_{n-1}(y,t) f_T(y|0)\bar F_{ \tilde T}(y)dy\nonumber.
\end{eqnarray} 

Then, by using the commutative property of convolution, we get 
\begin{eqnarray}
 \!\!\!\!\!\!\!\!\!\!\!\!\!&&W_{B_{MR}}(t)\!= \!\!F_{B_{MR}}(t|0)\!\!\hiderel+ \! \!\!\int_0^t\!\! \sum_{n=1}^\infty \!\!B_{n}(0,y) f_T(t\!-\!y|y)\bar F_{ \tilde T}(t\!-\!y)dy\nonumber\\
\!\!\!\!\!\!\!\!\!\!\!\!\!&&\iff \!\! W_{B_{MR}}(t)\!=\!\! F_{B_{MR}}(t|0)+\!\!\int_0^t \!\! W_{B_{MR}}(y) f_{B_{MR}}(t-y|y)dy\nonumber
.\end{eqnarray}
where $ F_{B_{MR}}(t|y)\!\!=\!\! \int_0^t\!\!f_T(x|y)\bar F_{ \tilde T}(y)dx$, and $ f_{B_{MR}}(t|y)\!\!=\!\! f_T(t|y)\bar F_{ \tilde T}(t)$.
\end{proof}

\subsection{Proof of Theorem \ref{theo:BMRapp1}.}
\label{app:Theorem7_proof}

\begin{proof}

%From Lemma \ref{lem:WB3to4}, we can see that the renewal function for Process B can be approximated by the sum of the two cases: (1) the mean number of renewals that might be the same as a classical renewal if the probability $p_{ \mathcal B}$ occurs at least $n$ times. (2) If the occurring of the probability $p_{\mathcal B}$ is less than $n$ times, the mean number of renewals of Process B is less than a classical renewal.

%Based on this Lemma, we give the proof as the following.

The probability $P(N_{B_{MR}}(t) = n)$ of having exactly $n$ renewals by time $t$ in Process B-MR for $n\geq 1$ can be related to the probability $P(N_g(t) = n)$ of having exactly $n$ renewals by time $t$ in a  g-renewal process that only involves the NC, as follows: Every renewal in the g-renewal process corresponds to a renewal in Process B-MR as long as CR survives NC given that NC has been working for time $S_{n-1}=y$. That is, $P(N_{B_{MR}}(t) = n)$ equals $P(N_g(t) = n)$ if  CR survives NC for $n$ consecutive renewals (an event that has probability   $\rho_{\tau}(n)$ (Eq. \ref{eqn:rhotau})).
On the other hand, if CR survives NC for $n$ consecutive renewals and dies before NC during the $(n+1)$th renewal, then the number of renewals in Process B-MR is going to remain at $n$,  while the number of renewals in the g-renewal process keeps increasing forever. 
Therefore, $P( N_{B_{MR}}(t)=n)$ for $n\geq 1$ can be approximated by 
 \begin{eqnarray}
 \!\!\!\!\!\!\!\!\!\!\!\!&&P(N_{B_{MR}}(t)\!=\!n)\approx P(N_g(t)=n)\rho_{\tau}(n)\nonumber \\
   \!\!\!\!\!\!\!\!\!\!\!\!\!\!&&+P(N_g(t) \!>\!n)(1\!-\!p_{ B_{MR}}(S_n))\rho_{\tau}(n).
 \end{eqnarray}
%and $P( N_{B_{MR}}(t)=0)=P(N_g(t)=0)$.

Applying the same explanation as the Process B proof (Lemma \ref{lem:WB1}), we have that the approximation in this case is a variant of  Eq. \ref{eqn:AappW2} and is equal to:
\begin{eqnarray}\label{eqn:BMR_0}
W_{B_{MR}}(t)\!\approx\!\sum_{n=1}^\infty \!P(N_g(t)=n)\sum_{k=1}^{n}\rho_{\tau} (k).
  \end{eqnarray}

Letting $V(n)=\sum_{k=1}^{n}\rho_{\tau} (k)$, we furthermore have that 
\begin{eqnarray}\label{eqn:BMR_E}
W_{B_{MR}}(t)\!\approx \E{V(N_g(t))}.
  \end{eqnarray}
  
 The computation of this average is even more complicated than that for Process B due to the ``history'' hidden in function $N_g(t)$.
However, since ${V(N_g(t))}$ is a convex function (being the summation of probability products that can be approximated by exponential functions), we know that the Jensen gap $\E{V(N_g(t))}-V(\E{N_g(t)})$ is bounded \cite{Mitrinovi1993}, and therefore we use 
$V(\E{N_g(t)}) = V(M_g(t))$ to approximate\footnote{To put this approximation into perspective, we note that for Process B,  $p_{ B_{MR}}(y)=p_B$, and $N_g(t)=N(t)$ since $p_{\mathcal B_{MR}}(y)$ is independent of the past time $y$. Then, $V(n) = \sum_{k=1}^{n}\rho_{\tau} (k)=\sum_{k=1}^{n}p_{ B}^k=\frac{p_{ B}}{(1-p_{ B})}(1-p_{ B}^n)$. Therefore, Eq. \ref{eqn:BMR_E} becomes Eq. \ref{eqn:TB2}, i.e., $\hat W_{B_1}(t)=\frac{p_{ B}}{(1-p_{ B})}(1-\E{p_{ B}^{N(t)}})$, and Eq. \ref{eqn:BMR} becomes Eq. \ref{eqn:TB3}, i.e.,  $\hat W_{B_2}(t)=\frac{p_{ B}}{(1-p_{ B})}(1-p_{ B}^{\E{N(t)}})$, which may be seen as involving the Jensen gap, although Eq.  \ref{eqn:TB3} was obtained from Eq. \ref{eqn:TB2} without appealing to the Jensen gap. }
$\E{V(N_g(t))}$. 
Therefore $W_{B_{MR}}(t)$ is now approximated by the function $\hat W_{B_{MR}\_1}(t)$ as, 
\begin{eqnarray} \label{eqn:WBMR_1}
\!\!\!\!\!\!\!\!\!\!\!\!\!&& 
 \hat W_{B_{MR}\_1}(t) = V(M_g (t)) =\sum_{k=1}^{\lceil M_g (t) \rceil}\rho_{\tau} (k),\nonumber
\end{eqnarray}
where we have used  the ceiling function $\lceil  \cdot \rceil$ to force $M_g(t)$ to be an integer, as required by $V(n)$.
 \end{proof}

\subsection{Proof of Theorem \ref{theo:BMRapp2}.}
\label{app:Theorem8_proof}
\begin{proof}
We note  first that:

 \begin{eqnarray}
 \!\!\!\!\!\!\!\!\!\!&&n\rho_{\tau} (n)+\sum_{k=1}^{n-1}k (1-p_{\mathcal B_{MR}}(k))\cdot \rho_{\tau} (k)\nonumber\\
 \!\!\!\!\!\!\!\!\!\!&&=n\rho_{\tau} (n)+\sum_{k=1}^{n-1}k \cdot \Big(\rho_{\tau} (k)-\rho_{\tau} (k+1)\Big)\nonumber\\
 \!\!\!\!\!\!\!\!\!\!&&=n\rho_{\tau} (n)\!\!+\!\!\Big(\rho_{\tau} (1)\!-\!\rho_{\tau} (2)\Big)\!\!+\!\!2\Big(\rho_{\tau} (2)\!-\!\rho_{\tau} (3)\Big)\!\!
 \nonumber\\
\!\!\!\!\!\!\! &&+3\Big(\rho_{\tau} (3)\!-\!\rho_{\tau} (4)\Big)+...+(n-1)  \Big(\rho_{\tau} (n-1)\!-\!\rho_{\tau} (n)\Big)\nonumber\\
\!\!\!\!\!\!\!\!\!\!&&=\rho_{\tau} (1)+\rho_{\tau} (2)+...+[n\rho_{\tau} (n)-(n-1) \rho_{\tau} (n)]\nonumber\\
\!\!\!\!\!\!\!\!\!\!&&=\sum_{k=1}^{n} \rho_{\tau} (k).
  \end{eqnarray}
  
  That is,
  
  \begin{subequations}
 \begin{eqnarray}
 \!\!\!\!\!\!\!\!\!\!\!\!\sum_{k=1}^{n} \rho_{\tau} (k)\!\!\!\!\!\!\!\!\!\!\!\!&&=n\rho_{\tau} (n)\!\!+\!\!\sum_{k=1}^{n-1}k \cdot \Big(\rho_{\tau} (k)\!-\!\rho_{\tau} (k+1)\Big)\\
&& =n\rho_{\tau} (n)\!\!+\!\!\sum_{k=1}^{n-1}k (1-p_{\mathcal B_{MR}}(k))\cdot \rho_{\tau} (k).
  \end{eqnarray}
    \end{subequations}

Based on the above and  Eq. \ref{eqn:BMR}  we substitute $n$ by $M_g(t)$  and use the ceiling function $\lceil M_g (t) \rceil$ to obtain the discrete value needed for the summations and $\rho_\tau(n)$.
 
 \end{proof}

  \linespread{1.0}
% ===========================================================================
% bibliography
% ===========================================================================
\bibliographystyle{IEEEtran}
\bibliography{References}

\end{document}